%% file: paper.tex
\documentclass[sigplan,10pt]{acmart}

\usepackage{amsthm}
\usepackage{algorithm}
\usepackage[noend]{algpseudocode}
\usepackage{caption}
\usepackage{color}
\usepackage{fancyvrb} 
\usepackage[group-separator={,},binary-units=true]{siunitx}
\usepackage{soul}
\usepackage{subcaption}
\usepackage{xspace}
\usepackage{xifthen}
\usepackage{pifont}

\usepackage{bm}
\usepackage{titlecaps}
\usepackage{mfirstuc}
\usepackage{enumitem}
\usepackage{microtype}
\usepackage{pgf}
\usepackage{tikz}

\usetikzlibrary{arrows,automata,fit,shapes}

\DeclareSIUnit[number-unit-product = ]\percent{\char`\%}

\newcommand{\cmark}{\ding{51}}
\newcommand{\xmark}{\ding{55}}

\algnewcommand\algorithmicmsg[2]{\ifthenelse{\isempty{#2}}{\textit{#1}\xspace}{\textit{#1}(#2)\xspace}}
\algnewcommand\Msg{\algorithmicmsg}
\algnewcommand\algorithmicsend{\textbf{send}\xspace}
\algnewcommand\Send{\algorithmicsend}
\algnewcommand\algorithmicwait{\textbf{wait}\xspace}
\algnewcommand\WaitUntil{\algorithmicwait}
\algnewcommand\algorithmicvar{\textbf{var}\xspace}
\algnewcommand\Var{\algorithmicvar}
\algnewcommand\algorithmicglobal{\textbf{state}\xspace}
\algnewcommand\Global{\algorithmicglobal}
\algnewcommand{\IfThen}[2]{
  \State \algorithmicif\ #1\ \algorithmicthen\ #2}
\algnewcommand{\IfThenElse}[3]{
  \State \algorithmicif\ #1\ \algorithmicthen\ #2\ \algorithmicelse\ #3}
\algnewcommand\algorithmicswitch{\textbf{switch}}
\algnewcommand\algorithmiccase{\textbf{case}}
\algdef{SE}[SWITCH]{Switch}{EndSwitch}[1]{\algorithmicswitch\ #1\ \algorithmicdo}{\algorithmicend\ \algorithmicswitch}
\algdef{SE}[CASE]{Case}{EndCase}[1]{\algorithmiccase\ #1}{\algorithmicend\ \algorithmiccase}
\algtext*{EndSwitch}
\algtext*{EndCase}
\algdef{SE}[DOWHILE]{Do}{DoWhile}{\algorithmicdo}[1]{\algorithmicwhile\ #1}

\newtheorem{thm}{Theorem}
\newtheorem{lem}[thm]{Lemma}

\newcommand{\noeditingmarks}{}

\newcommand{\textred}[1]{\textcolor{red}{#1}}

\ifx\noeditingmarks\undefined
   \newcommand{\pgwrapper}[2]{\textbf{#1: }\textred{\textit{#2}}}
   \newcommand{\newtext}[1]{{\color{blue}#1}}
   
\else
   \newcommand{\pgwrapper}[2]{}
  \newcommand{\newtext}[1]{#1}
   
\fi

\definecolor{worange}{RGB}{245, 128, 37}

\newcommand{\makePasteable}{no}
\ifthenelse{\equal{\makePasteable}{yes}}{%
  \hyphenpenalty 10000
  \exhyphenpenalty 10000
  \setlength{\paperheight}{100in}
  \setlength{\textheight}{100in}
  \onecolumn
}{%
}

\usepackage[compact]{titlesec}
\renewcommand{\noindentparagraph}[1]{\vspace{.5ex} \noindent \textbf{\textit{#1}}}

\titlespacing*{\section} {0pt}{2.5ex plus 1ex minus .2ex}{1ex plus .2ex}

\titlespacing*{\subsection}{0pt}{1.3ex plus 1ex minus .2ex}{.75ex plus .2ex}

\setlist{itemsep=1.5pt,parsep=1.5pt}             

\newcommand{\rs}{RSC}
\newcommand{\rslong}{regular sequential consistency}
\newcommand{\Rslong}{Regular sequential consistency}
\newcommand{\RSlong}{Regular Sequential Consistency}

\newcommand{\rss}{RSS}
\newcommand{\rsslong}{regular sequential serializability}

\newcommand{\RSSlong}{Regular Sequential Serializability}

\newcommand{\gryff}{Gryff}
\newcommand{\gryffrs}{\gryff{}-\rs{}}
\newcommand{\spanner}{Spanner}
\newcommand{\spannerrss}{\spanner{}-\rss{}}

\newcommand{\librss}{\texttt{libRSS}}

\newcommand{\RTBarriers}{Real-Time Fences}

\newcommand{\Rtbarriers}{Real-time fences}
\newcommand{\rtbarrier}{real-time fence}
\newcommand{\rtbarriers}{real-time fences}
\newcommand{\barrier}{fence}
\newcommand{\barriers}{fences}

\DeclareMathOperator*{\states}{\textit{states}}
\DeclareMathOperator*{\start}{\textit{start}}
\DeclareMathOperator*{\sig}{\textit{sig}}
\DeclareMathOperator*{\trans}{\textit{trans}}

\DeclareMathOperator*{\acts}{\textit{acts}}
\DeclareMathOperator*{\localacts}{\textit{local}}
\DeclareMathOperator*{\extacts}{\textit{extacts}}
\DeclareMathOperator*{\actin}{\textit{in}}
\DeclareMathOperator*{\actout}{\textit{out}}
\DeclareMathOperator*{\actint}{\textit{int}}
\DeclareMathOperator*{\sched}{\textit{sched}}
\DeclareMathOperator*{\trace}{\textit{trace}}
\DeclareMathOperator*{\complete}{\textit{complete}}
\DeclareMathOperator*{\system}{\textit{sys}}
\DeclareMathOperator*{\user}{\textit{user}}
\DeclareMathOperator*{\conflicts}{\mathcal{C}}

\DeclareMathOperator*{\send}{\textit{sendto}}
\DeclareMathOperator*{\sent}{\textit{sent}}
\DeclareMathOperator*{\request}{\textit{recvfrom}}
\DeclareMathOperator*{\receive}{\textit{received}}
\DeclareMathOperator*{\sendto}{\send}
\DeclareMathOperator*{\recvfrom}{\request}

\newcommand{\anomaly}[1]{\mathcal{A_\textrm{#1}}}
\newcommand{\invariant}[1]{\mathcal{I_\textrm{#1}}}

\newcommand{\type}{\mathfrak{T}}
\newcommand{\spec}{\mathfrak{S}}
\DeclareMathOperator*{\vals}{\textit{vals}}
\DeclareMathOperator*{\invs}{\textit{invs}}
\DeclareMathOperator*{\resps}{\textit{resps}}
\DeclareMathOperator*{\ops}{\textit{ops}}

\newcommand{\caused}{\rightsquigarrow} 
\newcommand{\rt}{\rightarrow} 

\DeclareMathOperator*{\pastset}{\mathcal{P}}
\DeclareMathOperator*{\lastinv}{\mathcal{L}}
\DeclareMathOperator*{\nextbar}{\text{nf}}

\newcommand{\tmin}{t_{\text{min}}}

\newcommand{\tsnap}{t_{\text{snap}}}
\newcommand{\tprepare}{t_{\text{p}}}
\newcommand{\tcommit}{t_{\text{c}}}

\newcommand{\tread}{t_\text{read}}
\newcommand{\trw}{T_\text{RW}}
\newcommand{\tro}{T_\text{RO}}

\newcommand{\elb}{t_\text{ee}}

\newcommand{\timestamp}{carstamp}
\newcommand{\mts}{\mathit{cs}}

\DeclareMathOperator*{\fq}{\mathit{dp}}
\DeclareMathOperator*{\pp}{\mathit{pp}}
\DeclareMathOperator*{\vp}{\mathit{vp}}
\DeclareMathOperator*{\inv}{\mathit{inv}}
\DeclareMathOperator*{\resp}{\mathit{resp}}

\acmYear{2021}\copyrightyear{2021}
\setcopyright{rightsretained}
\acmConference[SOSP '21]{ACM SIGOPS 28th Symposium on Operating Systems Principles}{October 26--29, 2021}{Virtual Event, Germany}
\acmBooktitle{ACM SIGOPS 28th Symposium on Operating Systems Principles (SOSP '21), October 26--29, 2021, Virtual Event, Germany}
\acmPrice{15.00}
\acmDOI{10.1145/3477132.3483566}
\acmISBN{978-1-4503-8709-5/21/10}

\acmBadgeR{artifacts_available_v1_1}

\begin{document}

\title[\RSSlong{} and \RSlong{}]{\RSSlong{} and \\ \RSlong{}}


\author{Jeffrey Helt}
\affiliation{\institution{Princeton University}\country{United States}}
\email{jhelt@cs.princeton.edu}

\author{Matthew Burke}
\affiliation{\institution{Cornell University}\country{United States}}
\email{matthelb@cs.cornell.edu}

\author{Amit Levy}
\affiliation{\institution{Princeton University}\country{United States}}
\email{aalevy@cs.princeton.edu}

\author{Wyatt Lloyd}
\affiliation{\institution{Princeton University}\country{United States}}
\email{wlloyd@princeton.edu}

\renewcommand{\shortauthors}{Helt et al.}

\begin{abstract}
    Strictly serializable (linearizable) services appear to execute transactions (operations) sequentially, in an order consistent with real time. This restricts a transaction's (operation's) possible return values and in turn, simplifies application programming. In exchange, strictly serializable (linearizable) services perform worse than those with weaker consistency. But switching to such services can break applications.
    
    This work introduces two new consistency models to ease this trade-off: \rsslong{} (\rss{}) and \rslong{} (\rs{}).
    They are just as strong for applications: we prove any application invariant that holds when using a strictly serializable (linearizable) service also holds when using an \rss{} (\rs{}) service.
    Yet they relax the constraints on services---they allow new, better-performing designs.
    To demonstrate this, we design, implement, and evaluate variants of two
    systems, \spanner{} and \gryff{}, relaxing their consistency to \rss{} and
    \rs{}, respectively. The new variants achieve better
    read-only transaction and read tail latency than their counterparts.
\end{abstract}



\settopmatter{printacmref=true, printccs=true, printfolios=false, authorsperrow=4}

\maketitle

\input{sections/intro}
\input{sections/background}
\input{sections/consistency}
\input{sections/spanner}
\input{sections/eval}
\input{sections/gryff}
\input{sections/related_work}
\input{sections/conclusion}

\begin{acks}
    We thank the anonymous reviewers and our shepherd, Rodrigo Rodrigues, for their helpful comments and feedback.
    We are also grateful to Khiem Ngo for his comments on an earlier
    version of this paper. This work was supported by the National Science Foundation under grant CNS-1824130.
\end{acks}

\clearpage
\bibliographystyle{ACM-Reference-Format}
\bibliography{refs,paper,venues}

\clearpage
\appendix

\input{sections/consistency_comparison}
\input{sections/gryff_full_design}
\input{sections/proof}
\input{sections/correctness_proofs}

\end{document}

%% file: sections/intro.tex
\section{Introduction}
\label{sec:intro}

Strict serializability~\cite{papadimitriou1979serializability} and
linearizability~\cite{herlihy1990linearizability} are exemplary consistency
models. Strictly serializable (linearizable) services appear to execute
transactions (operations) sequentially, in an order consistent with real time.
They simplify building correct applications atop them by reducing the number
of possible values services may return to application processes. This, in turn,
makes it easier for programmers to enforce necessary application invariants.

In exchange for their strong guarantees, strictly serializable and
linearizable services incur worse performance than those with weaker
consistency~\cite{hunt2010zookeeper,deCandia2007dynamo,lloyd2011cops,lloyd2013eiger,bailis2016ramp}.
For example, consider a read in a key-value store that returns the value written by a
concurrent write. If the key-value store is weakly consistent, the read imposes
no constraints on future reads. But if the key-value store is strictly serializable,
the read imposes a global constraint on future reads---they all must
return the new value, even if the write has not yet finished. Existing strictly
serializable services guarantee this by blocking
reads~\cite{corbett2013spanner}, incurring multiple round trips between clients
and shards~\cite{zhang2018tapir,zhang2015tapir}, or aborting conflicting
writes~\cite{zhang2018tapir,zhang2015tapir}. These harm service performance,
either by increasing abort rates or increasing latency.



Services with weaker consistency
models~\cite{ahamad1995causal,lloyd2011cops,akkoorath2016cure,lamport1979sequential}, however, offer application programmers with a harsh trade-off. In exchange for
better performance, they may break the invariants of applications built atop them.


\newtext{
This work introduces two new consistency models to ease this trade-off:
\rsslong{} (\rss{}) and \rslong{} (\rs{}). They allow services to achieve
better performance while being \textit{invariant-equivalent} to strict
serializability and linearizability, respectively. For any application that
does not require synchronized clocks, any invariant that holds while
interacting with a set of strictly serializable (linearizable) services also
holds when executing atop a set of \rss{} (\rs{}) services.
}

\newtext{
To maintain application invariants, a set of \rss{} (\rs{}) services must
appear to execute transactions (operations) sequentially, in an order that is
consistent with a broad set of causal constraints (e.g., through message passing). We
prove formally this is sufficient for \rss{} (\rs{}) to be
invariant-equivalent to strict serializability (linearizability).
}

To allow for better performance, \rss{} and \rs{} relax some of strict serializability
and linearizability's real-time guarantees for causally unrelated transactions
or operations, respectively. \newtext{For example, when a read returns the value written by a concurrent write, instead of a global constraint,
\rss{} imposes a causal constraint---only reads that causally follow the first
must return the new value.}

But in addition to helping enforce invariants, strict serializability's (linearizability's) real-time guarantees help applications match their users'
expectations. For instance, from interacting with applications on their local
machine, users expect writes to be immediately visible to all future reads.
Applications built atop weakly consistent services can violate
these expectations, exposing anomalies.

\newtext{
Because \rss{} (\rs{}) relax some of strict serializability's (linearizability's)
real-time constraints, applications built atop an \rss{} (\rs{}) service may
expose more anomalies. But
prior work suggests anomalies are rare in practice~\cite{lu2015existentialConsistency}, and further, \rss{} and
\rs{} include some real-time guarantees to make the chance of observing these new anomalies
small. They should only be possible within short time windows (a few
seconds). Thus, we expect the
difference between \rss{} (\rs{}) and strict serializability (linearizability)
to go unnoticed in practice.
}

\newtext{
To compose a set of \rss{} (\rs{}) services such that they appear to execute transactions (operations) in some global \rss{} (\rs{}) order, each must implement one other mechanism: a \rtbarrier{}. We show how the necessary \barriers{} can be invoked without changing applications.
}


Finally, to demonstrate the performance benefits permitted by \rss{} and \rs{}, we
design, implement, and evaluate variants of two existing services:
\spanner{}~\cite{corbett2013spanner}, Google's globally distributed database,
and \gryff{}~\cite{burke2020gryff}, a replicated key-value store. The variants
implement \rss{} and \rs{} instead of strict serializability and
linearizability, respectively.

\spannerrss{} improves read-only transaction latency by reducing the chances
they must block for conflicting read-write transactions. Instead, \spannerrss{}
allows read-only transactions to immediately return old values in some cases. As
a result, in low- and moderate-contention workloads, \spannerrss{} reduces
read-only transaction tail latency by up to \newtext{\SI{49}{\percent}} without affecting
read-write transaction latency.

\gryffrs{} improves read latency with a different approach. By removing the
write-back phase of reads, \gryffrs{} halves the number of round trips required
between application processes and \gryff{}'s replicas. As a result, for
moderate- and high-contention workloads, \gryffrs{} reduces p99 read
latency by about \SI{40}{\percent}. Further, because \gryffrs{}'s reads always
finish in one round, it offers larger reductions in latency (up to
\SI{50}{\percent}) farther out on the tail.

In sum, this paper makes the following contributions:
\begin{itemize}[leftmargin=*]
\item \newtext{We define \rss{} and \rs{}, the first invariant-equivalent consistency models to strict serializability and linearizability.}
\item \newtext{We prove that for any application not requiring synchronized clocks, any invariant that holds with strictly
  serializable (linearizable) services also holds with \rss{} (\rs{}).}
\item We design, implement, and evaluate \spannerrss{} and \gryffrs{}, which
  significantly improve read tail latency compared to their counterparts.
\end{itemize}


%% file: sections/background.tex
\section{Background and Motivation}
\label{sec:background}

\begin{figure}[!t]
\centering
\includegraphics[width=\linewidth,page=1]{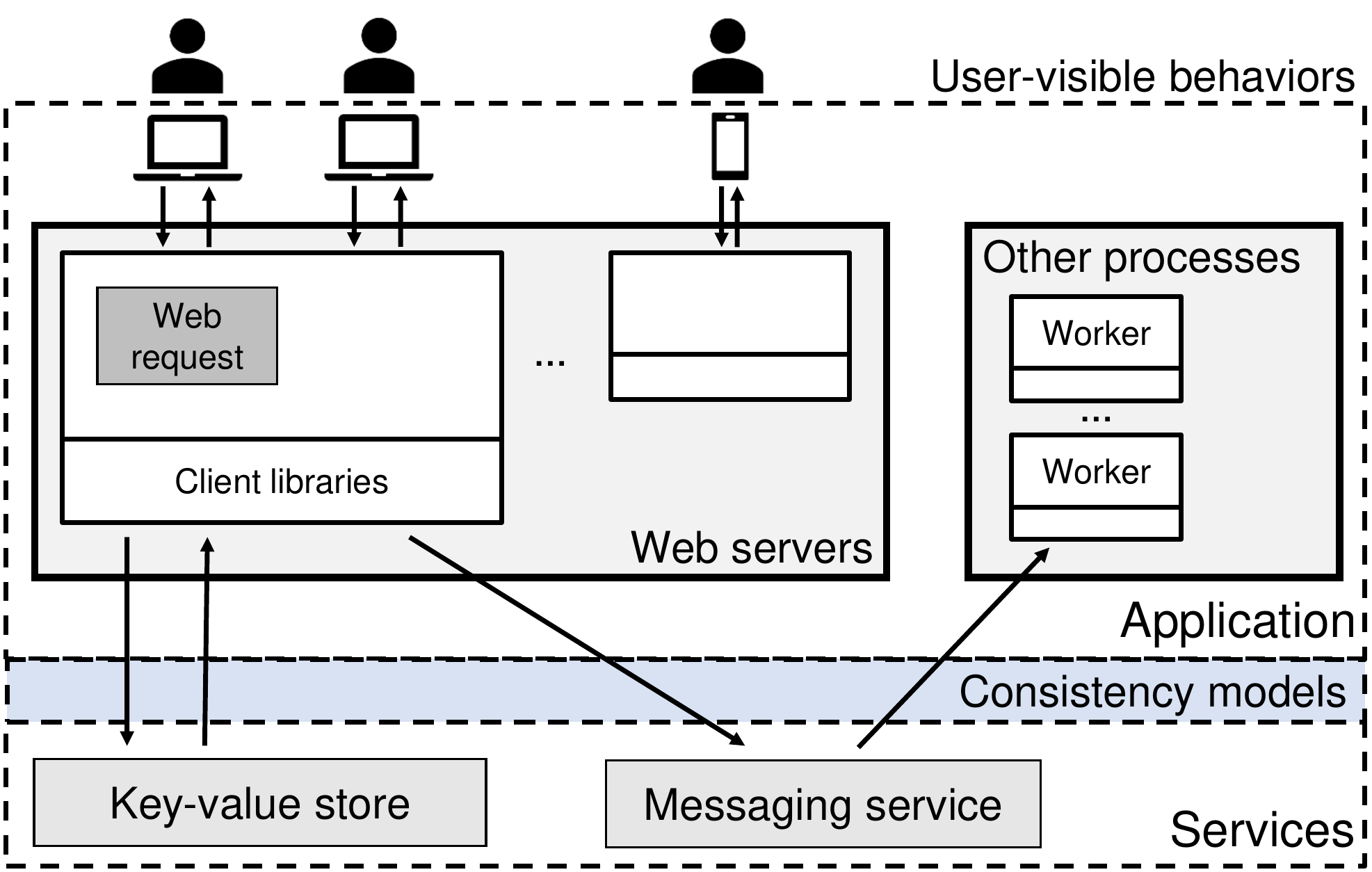}
\caption{An application deployed in a data center. It
  comprises the processes running on user devices, Web servers, and
  asynchronous workers. They are supported by a pair of services. The services'
  consistency models significantly impact application
  correctness and performance.}
\label{fig:application}
\end{figure}

In this section, we first describe the typical structure of applications and
their interaction with supporting services. We then discuss the role
consistency models play in these interactions. Finally, we demonstrate how
existing consistency models offer difficult trade-offs to application
programmers and service designers.

\begin{table*}[!t]
\centering
\include{figures/consistency-table}
\vspace{2ex}
\caption{Comparing consistency models by which invariants hold, which anomalies
  are prevented, and the latency of operations. $\invariant{1}$ states an album
  never contains a photo with null data; $\invariant{2}$ states a worker never
  reads a photo with null data.}
 \label{tbl:consistency-comparison}
\end{table*}

\subsection{Distributed Applications}
\label{sec:background:applications}

Distributed applications can be split into two parts: a set of processes executing
application-specific logic and a set of services supporting them. The
application-specific processes include those that respond interactively to
users, such as those executing on a user's device and those they cooperate
with synchronously or asynchronously, such as Web servers running in a nearby
data center. The services provide generic, reusable functionality, such as data
storage~\cite{lloyd2011cops,lloyd2013eiger,bronson2013tao,corbett2013spanner}
and messaging~\cite{sharma2015wormhole}.

For example, consider a Web application deployed in a data center
(Figure~\ref{fig:application}). A user interacts with a browser on their device.
These interactions define the behaviors of the application. Under the hood, the
browser sends HTTP requests to a set of Web servers. In processing a request, a
server reads and modifies state in a key-value store and renders responses.
There are also
worker processes that these servers invoke asynchronously to perform longer
running tasks~\cite{huang2017sve}. The set of processes running
application-specific logic are the \textit{clients} of the services. The
services are responsible for persisting application state, replicating it across
data centers, and coordinating between application components.

\subsection{Motivating Example: Photo-Sharing App}
\label{sec:background:example}

Throughout the paper, we consider a simple but illustrative example:
a photo-sharing application. The application allows users to backup and share
their photos with other users while handling compression and other photo-processing functions.

In our example application, photos are organized into albums. Both photos and
albums are stored in a globally distributed, transactional key-value store. Each
photo and album has a unique key. A photo's key maps to a binary blob while an
album's key maps to structured data containing the keys of all photos in the
album. If a key is read but is not present, the key-value store simply returns \texttt{null}.

In addition to the transactional key-value store, the application uses a
messaging service to enqueue requests for asynchronous processing. For example,
when a user adds a high-resolution photo, the application enqueues the
photo's key in the messaging service, requesting some worker process create
lower-resolution thumbnails of the image.

When a user adds a new photo to an album, a Web server issues a read-write
transaction: it creates a new key-value mapping for the photo, reads the album,
and writes back the album after modifying its value to include a reference to
the newly added photo. Then it enqueues a request for additional processing.

\subsection{Consistency Models}
\label{sec:background:consistency}

The correctness and performance of an application are heavily
influenced by the consistency models of its supporting services. A \textit{consistency model} is a contract between a service and its
clients regarding the allowable return values for a given set of operations.
Services with stronger consistency are generally restricted to fewer possible return values, so
it is easier for programmers to build a correct application atop them.
Restricting the allowable return values, however, often incurs worse performance
in these services and consequently, in the application.

\noindentparagraph{Invariants and anomalies.} The stronger restrictions of
stronger consistency models enable applications to more easily ensure
correctness and provide better semantics to users. \newtext{The correctness of an
application is determined by its \textit{invariants}, which are logical
predicates that hold for all states of an application (i.e., the combined states of all application processes).} The semantics for
users are determined by the rate of \textit{anomalies}, which
are behaviors the user would not observe while accessing a
single-threaded, monolithic application running on a local machine with no
failures. Table~\ref{tbl:consistency-comparison} shows some invariants
and anomalies for our application.

Application logic relies on invariants to function correctly. For the
photo-sharing example, client-side application logic assumes that if an album
contains a reference to a photo, the photo exists in the key-value store
($\invariant{1}$). Similarly, workers that receive a photo's key through the
messaging service assume fetching the key from the key-value store will
not return \texttt{null} ($\invariant{2}$).

\newtext{
Applications also attempt to present reasonable behaviors to users, which is
quantified by the rate of anomalies. Unlike an invariant violation, the detection of an anomaly may require information that is beyond the application's state. For example, once Alice adds a new photo to an album, Bob not seeing it is an anomaly
($\anomaly{2}$). But detecting $\anomaly{2}$ requires the application either to have synchronized clocks to record the start and end times of Alice and Bob's requests or to somehow know that Alice communicated with Bob.
}

\subsection{Strict Serializability is Too Strong}
\label{sec:strict_serializability}

Strict serializability~\cite{papadimitriou1979serializability} is one of the
strongest consistency models. A service that guarantees \textit{strict serializability} appears to execute
transactions sequentially, one at a time, in an order consistent with the
transactions' real-time order. As a result, only transactions that are
concurrent (i.e., both begin before either ends) may be
reordered. Further, strict serializability is \textit{composable}:
clients may use multiple services, and the resulting
execution will always be strictly serializable because real-time order is universal to
all services~\cite{herlihy1990linearizability}.\footnote{This holds only with
  the reasonable assumption that individual transactions do not span multiple
  services. This makes each strictly serializable service equivalent to a
  linearizable ``object''~\cite{herlihy1990linearizability}.}

Strict serializability ensures a large set of invariants hold. For example, it
ensures both invariants hold for our photo-sharing application. For
$\invariant{1}$, the application logic writes a photo's data and adds it to an
album in a single transaction $T$. Strict serializability then trivially
ensures $\invariant{1}$; because transactions appear to execute sequentially,
any other transaction will be either before $T$ and not see the photo
or after $T$ and see both the photo in the album and its
data. For $\invariant{2}$, the application logic first executes the
add-new-photo transaction and then enqueues a request to process it in the
messaging service. The real-time order and composability of strict
serializability then ensures $\invariant{2}$; the enqueue begins in real time
after the add-new-photo transaction ends, and thus, any process that
sees the enqueued request must subsequently see the writes of the add-new-photo
transaction.

Strict serializability also mitigates anomalies. As shown in
Table~\ref{tbl:consistency-comparison}, $\anomaly{1}$, $\anomaly{2}$, and
$\anomaly{3}$ never occur with a strictly serializable key-value store. Because strictly
serializable transactions appear to execute sequentially, no writes are lost, and its
real-time guarantees ensure Bob's transactions always follow Alice's after
receiving her call.

Yet applications built atop strictly serializability services are not perfect. For instance,
asynchronous networks, transient failures, and a lack of fate sharing among
components can all cause anomalies that are beyond the scope of a
consistency model. $\anomaly{4}$ in Table~\ref{tbl:consistency-comparison} shows
one example, \newtext{which could not occur on a local machine with no failures.}

\noindentparagraph{Strict serializability imposes performance costs.} In
exchange for ensuring invariants and preventing most anomalies, strict
serializability imposes significant performance costs on services. For instance,
consider anomaly $\anomaly{3}$ in Table~\ref{tbl:consistency-comparison}, and
assume Charlie is in the middle of adding a photo when Alice sees it. Strict
serializability mandates that any subsequent read by any application server
includes the photo, even if Bob is on a different continent and Charlie's
transaction has not finished. As a result, the key-value store must
ensure Alice's transaction only includes Charlie's photo once all
subsequent reads will, too.

Existing services provide this guarantee through a variety
of mechanisms. Some block read-only transactions during conflicting read-write
transactions~\cite{corbett2013spanner}. Others incur multiple round trips
between an application server and a set of
replicas~\cite{zhang2015tapir,zhang2018tapir} or abort concurrent read-write
transactions~\cite{zhang2015tapir,zhang2018tapir}. These mechanisms reduce
performance by increasing either read-only transaction latency or abort rates.

\subsection{Process-Ordered Serializability is Too Weak}
\label{sec:process-ordered serializability}

Because strict serializability incurs heavy performance overhead, many services
provide a weaker consistency model. The next strongest is
\textit{process-ordered (PO) serializability}, which guarantees that services appear
to execute transactions sequentially, in an order consistent with each client's
process order~\cite{lu2016snow, daudjee2004lazy}. PO
serializability is weaker than strict serializability because it does not
guarantee that non-concurrent transactions respect their real-time order.
Moreover, PO serializability is not composable. Thus, process
orders across services can be lost.

Because PO serializability is weaker than strict serializability,
it avoids some of its performance costs. For instance, there are read-only
transaction protocols that can always complete in one round of non-blocking
requests with constant metadata in services with PO serializability,
while this is impossible with strict serializability~\cite{lu2020port}.



\noindentparagraph{PO serializability provides fewer invariants.}
In exchange for better performance, weaker consistency models, like
PO serializability, present application programmers with a harsh
trade-off: fewer invariants will hold with reasonable application
logic.\footnote{We say ``reasonable application logic'' because one can always
  write a middleware layer that implements a stronger consistency model, $X$,
  atop a weaker one, $Y$, e.g., by taking the ideas of bolt-on
  consistency~\cite{bailis2013bolton} to an extreme. But the resulting $X$
  middleware on $Y$ service is simply an inefficient implementation of an $X$
  service.} For our photo-sharing example, $\invariant{1}$ holds because like with strict
serializability, PO-serializable services appear to execute
transactions sequentially.

On the other hand,
$\invariant{2}$ does not hold because PO serializability is not
composable. A worker seeing a photo in the message queue does not ensure its
subsequent reads to the key-value store will include the writes of a preceding add-new-photo
transaction because the message queue and key-value store are distinct services.

\subsection{Non-Transactional Consistency Models}
\label{sec:background:non}

Our discussion above focuses on transactional consistency models. The same
tension between application invariants and service performance exists for the
equivalent non-transactional models.
Linearizability~\cite{herlihy1990linearizability} and sequential consistency~\cite{lamport1979sequential} are the non-transactional
equivalents of strict serializability and process-ordered serializability, respectively.
If we temporarily ignore albums and assume
application processes issue a single write to add a photo, invariant $\invariant{2}$ holds with a linearizable key-value
store but not with a sequentially consistent one. But linearizable services must
employ mechanisms that hurt their performance to satisfy linearizability's
constraints, for example, by requiring additional rounds of communication for
reads (\S\ref{sec:gryffrs}).

%% file: figures/consistency-table.tex
\begin{tabular}{| l | c c c | c c c c | c |}
\cline{2-9}
\multicolumn{1}{l|}{} & \multicolumn{3}{c|}{\textbf{Invariants}} & \multicolumn{4}{c|}{\textbf{Possible Anomalies}} & \multicolumn{1}{c|}{\textbf{Performance}} \\
\hline
\textbf{Consistency} & $\bm{\invariant{1}}$ & $\bm{\invariant{2}}$ & $\bm{\invariant{SS}}$ & $\bm{\anomaly{1}}$ & $\bm{\anomaly{2}}$ & $\bm{\anomaly{3}}$ & $\bm{\anomaly{4}}$ & \multicolumn{1}{|c|}{\textbf{Latency}} \\ \hline

Strict Serializability~\cite{papadimitriou1979serializability}
& \cmark & \cmark & \cmark
& never & never & never & always
& $\bm{\uparrow\uparrow}$ \\
\textbf{\RSSlong{}}
& \cmark & \cmark & \cmark
& never & never & temporarily & always
& $\bm{\uparrow}$ \\
Process-ordered Serializability~\cite{daudjee2004lazy,lu2016snow}
& \cmark & \xmark & \xmark
& never & always & always & always
& $\bm{-}$ \\

\hline \hline

\multicolumn{9}{|l|}{$\bm{\invariant{1\phantom{S}}:} \forall P, \forall i : i \in P.\textit{Album}.\textit{list} \implies P.\textit{Album}.\textit{photos}[i].\textit{data} \neq \textit{null}$} \\
\multicolumn{9}{|l|}{$\bm{\invariant{2\phantom{S}} :} \forall W, \forall i : \textsc{Head}(\textit{W.PhotoQ}) = i \land \textit{W.Photo}.\textit{id} = i \implies \textit{W.Photo}.\textit{data} \neq \textit{null}$} \\
\multicolumn{9}{|l|}{$\bm{\invariant{SS}:} \text{Any invariant that holds with strict serializability.}$} \\

\hline \hline

\multicolumn{9}{|l|}{$\bm{\anomaly{1}:}$ Alice adds two photos; later, only one photo
  is in her album.} \\
\multicolumn{9}{|l|}{$\bm{\anomaly{2}:}$ Alice adds a photo and calls Bob; Bob does not see the photo.} \\
\multicolumn{9}{|l|}{$\bm{\anomaly{3}:}$ Alice sees Charlie's photo and calls Bob; Bob does not see the photo.} \\
\multicolumn{9}{|l|}{$\bm{\anomaly{4}:}$ Alice tries to add a photo but never receives a response.} \\
\hline
\end{tabular}

%% file: sections/consistency.tex
\section{Regular Sequential Consistency Models}
\label{sec:consistency}

\newtext{
A consistency model's guarantees affect both application programmers and users. Stronger models place less burden on programmers (by guaranteeing more invariants) and users (by exposing fewer anomalies) but constrain service performance. In this work, we propose two new consistency models, \rsslong{} (\rss{}) and \rslong{} (\rs{}), to diminish this trade-off.
}

\newtext{
\rss{} and \rs{} are invariant-equivalent to strict serializability and linearizability, respectively. Thus, they place no additional burden on application programmers.
}

\newtext{
While they do allow more anomalies, prior work suggests anomalies with much weaker models (e.g., eventual consistency) are rare in practice (e.g., at most six anomalies per million operations~\cite{lu2015existentialConsistency}). Thus, we expect the additional burden on users to be negligible.
}

\newtext{
In this section, we define \rss{} and \rs{} and prove their invariant-equivalence to strict serializability and linearizability. We first describe our formal model of distributed applications
(\S\ref{sec:consistency:applications}) and the services they use
(\S\ref{sec:consistency:systems}). We then define \rss{} and \rs{}
(\S\ref{sec:consistency:consistency} and \S\ref{sec:consistency:definitions})
and finally prove our main result (\S\ref{sec:consistency:proof}). (We demonstrate \rss{} and \rs{} allow for services with better performance in later sections.)
}

\subsection{Applications and Executions}
\label{sec:consistency:applications}

We model a distributed \textit{application} as a collection of $n$
\textit{processes}. Processes are state machines~\cite{lynch1987ioa,lynch1996da}
that implement application logic by performing local computation, exchanging
messages, and invoking operations on services.

An application's processes define a prefix-closed set of \textit{executions},
which are sequences $s_0,\pi_1,s_1,\ldots$ of alternating \textit{states} and
\textit{actions}, starting and ending with a state. An application state
contains the state of each process---it is an $n$-length vector of process
states. As part of a process's state, we assume it has access to a local
clock, which it can use to set local timers, but the clock makes no guarantees
about its drift or skew relative to those at other processes.

Each action is a step taken by exactly one process and is one of three types:
\textit{internal}, \textit{input}, or \textit{output}. Internal actions model
local computation. Processes use input and output actions to interact with other
processes (e.g., receiving and replying to a remote procedure call) and their
environment (e.g., responding to a user gesture). As we will describe in the
following section, a subset of the input and output actions are
\textit{invocations} and \textit{responses}, respectively, which are used to
interact with services.

Processes can also exchange messages with one another via unidirectional
channels. To send a message to process $P_j$, $P_i$ uses two actions: first,
$P_i$ uses an output action $\sendto_{ij}(m)$ and later, an input action
$\sent_{ij}$ occurs, indicating $m$'s transmission on the network. Similarly, to receive a
message from $P_i$, $P_j$ first uses an output action $\request_{ij}$ and later,
an input action $\receive_{ij}(m)$ occurs, indicating the receipt of $m$.


Given an execution $\alpha$, we will often refer to an individual process's
\textit{sub-execution}, denoted $\alpha|P_i$. $\alpha|P_i$ comprises
only $P_i$'s actions and the $i$th component of each state in $\alpha$. 


\noindentparagraph{Well-formed.} An execution is \textit{well-formed} if it
satisfies the following: (1) Messages are sent before they are received; (2) A
process has at most one (total) outstanding invocation (at a service) or $\request_{ij}$ (at a channel); and (3) Processes do not take output steps while
waiting for a response from a service. We henceforth only consider well-formed executions.


\noindentparagraph{Equivalence.} Two executions $\alpha$ and $\beta$ are
\textit{equivalent} if for all $P_i$, $\alpha|P_i = \beta|P_i$. Intuitively,
equivalent executions are indistinguishable to the processes.

\subsection{Services}
\label{sec:consistency:systems}

Databases, message queues, and other back-end \textit{services} that application
processes interact with are defined by their \textit{operations} and a
\textit{specification}~\cite{herlihy1990linearizability,lynch1996da}. An
\textit{operation} comprises pairs of \textit{invocations}, specifying the
operations and their arguments, and matching \textit{responses}, containing
return values. The specification is a prefix-closed set of sequences of
invocation-response pairs defining the service's correct behavior in the absence
of concurrency. A sequence $S$ in specification $\spec$ defines a total order
over its operations, denoted $<_S$.

Several services can be composed into a composite service by combining their
specifications as the set of all
interleavings of the original services' specifications. Notably, this means a
service composed of constituent services that support transactions include those
transactional operations but \textit{does not} support transactions across its
constituent services. In the results below, we assume the processes interact
with an arbitrary (possibly composite) service.

\subsection{Consistency Models}
\label{sec:consistency:consistency}

A \textit{consistency model} specifies the possible responses a
service may return in the face of concurrent operations. Before we define our
new consistency models, we must define four preliminaries. \newtext{For ease of presentation,
two of our definitions, conflicts and reads-from, assume a key-value store interface. While these definitions could be
made general, we leave precisely defining them for other interfaces to future work.}

\noindentparagraph{Complete operations.} Given an execution $\alpha$, we say an
operation is \textit{complete} if its invocation has a matching response in
$\alpha$. We denote $\complete(\alpha)$ as the maximal subsequence of $\alpha$
comprising only complete operations~\cite{herlihy1990linearizability}.

\noindentparagraph{Conflicting operations.} \newtext{Given read-write transaction $W$,
we say a read-only transaction $R$ \textit{conflicts} with $W$ if $W$ writes a key that $R$ reads.
Given an execution $\alpha$, we denote the set of read-only transactions in $\alpha$ that conflict with $W$ as $\conflicts_\alpha(W)$. We define conflicts and $\conflicts_\alpha(w)$ analogously for non-transactional reads and writes.
}

\noindentparagraph{Real-time order.} Two actions in an execution $\alpha$ are
ordered in real
time~\cite{herlihy1990linearizability,papadimitriou1979serializability}, denoted
$\pi_1 \rt_\alpha \pi_2$, if and only if $\pi_1$ is a response, $\pi_2$ is an
invocation, and $\pi_1$ precedes $\pi_2$ in $\alpha$.

\noindentparagraph{Causal order.} Two actions are causally
related~\cite{ahamad1995causal,lloyd2011cops,lloyd2013eiger,akkoorath2016cure,lamport1978clocks,lamport1979sequential}
in an execution $\alpha$, denoted $\pi_1 \caused_\alpha \pi_2$ if any of the
following hold: (1) \textit{Process order:} $\pi_1$ precedes $\pi_2$ in a
process's sub-execution; (2) \textit{Message passing:} $\pi_1$ is a
$\sendto_{ij}(m)$ and $\pi_2$ is its corresponding $\receive_{ij}(m)$; (3)
\newtext{\textit{Reads from:} $\pi_1$ is operation $o_1$'s response, $\pi_2$ is $o_2$'s
invocation, and $o_2$ reads a value written by $o_1$;} or (4) \textit{Transitivity:} there exists
some action $\pi_3$ such that $\pi_1 \caused_\alpha \pi_3$ and $\pi_3
\caused_\alpha \pi_2$.

\subsection{\rss{} and \rs{}}
\label{sec:consistency:definitions}

We now define our new consistency models, \rsslong{} and \rslong{}. Their
definitions are nearly identical but because supporting transactions has
significant practical implications, we
distinguish between the transactional and non-transactional versions.

Intuitively, \rss{} (\rs{}) guarantees a total order of transactions
(operations) such that they respect causality. \newtext{Further, like prior ``regular'' models~\cite{lamport1986interprocess,shao2011mwregularity,viotti2016conssurvey}, reads must return a value at least as recent as the most recently completed, conflicting write.}

\noindentparagraph{\RSSlong{}.} Let $\mathcal{T}$ be the set of all transactions
and $\mathcal{W} \subseteq \mathcal{T}$ be the set of read-write transactions.
An execution $\alpha_1$ satisfies \textit{\rss{}} if it can be extended to
$\alpha_2$ by adding zero or more responses such that there exists a sequence $S
\in \spec$ where (1) $S$ is equivalent to $\complete(\alpha_2)$; (2) for all pairs
of transactions $T_1,T_2 \in \mathcal{T}$, $T_1 \caused_{\alpha_1} T_2 \implies T_1 <_S T_2$;
and (3) \newtext{for all read-write transactions $W \in \mathcal{W}$ and transactions $T
\in \conflicts_{\alpha_1}(W) \cup \mathcal{W}$, $W \rt_{\alpha_1} T \implies W <_S T$.}

\noindentparagraph{\RSlong{}.} Let $\mathcal{O}$ be the set of all operations and
$\mathcal{W} \subseteq \mathcal{O}$ be the set of writes. An execution
$\alpha_1$ satisfies \textit{\rs{}} if it can be extended to $\alpha_2$ by
adding zero or more responses such that there exists a sequence $S \in \spec$
where (1) $S$ is equivalent to $\complete(\alpha_2)$; (2) for all pairs of operations
$o_1,o_2 \in \mathcal{O}$, $o_1 \caused_{\alpha_1} o_2 \implies o_1
<_S o_2$; and (3) \newtext{for all writes $w \in \mathcal{W}$ and operations $o \in
\conflicts_{\alpha_1}(w) \cup \mathcal{W}$, $w \rt_{\alpha_1} o \implies w <_S o$.}

\subsection{\rss{} and \rs{} Maintain Application Invariants}
\label{sec:consistency:proof}


This section presents a condensed version of the proof. For brevity, we assume
here processes do not fail and all operations finish. The full proof is in
Appendix~\ref{sec:proof}.

\noindentparagraph{Preliminaries.} The results below reason about an
application's invariants, which are assertions about its states. Formally, we
say a state is \textit{reachable} in application $A$ if it is the final state of
some execution of $A$. An \textit{invariant} $\invariant{A}$ is a predicate that
is true for all of $A$'s reachable states~\cite{lynch1996da}.

In the proofs below, it will be convenient to focus on the actions within an
execution. Given an execution $\alpha$, its \textit{schedule}, $\sched(\alpha)$,
is the subsequence of just its actions.

\begin{figure}[!t]
\centering
\includegraphics[width=0.7\linewidth,page=3]{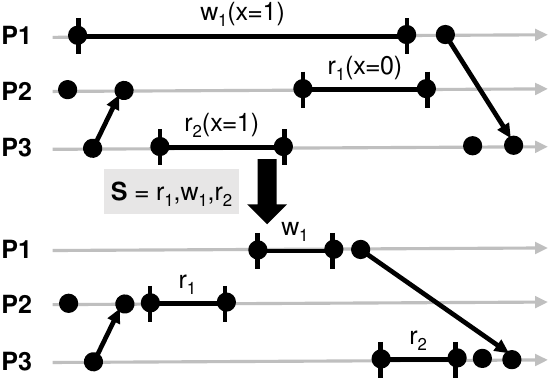}
\caption{Example transformation from an \rss{} execution to a strictly
  serializable one. Lemma~\ref{thm:executions:rss} proves such a transformation
  is possible with any \rss{} execution.}
\label{fig:transform-example}
\end{figure}

\noindentparagraph{Proof intuition.} Our main results follows from two
observations. First, Lemma~\ref{thm:executions:rss} shows we can transform an
execution in which the operations respect \rss{} into an execution in which they
respect strict serializability without reordering any actions at any of the
processes. Figure~\ref{fig:transform-example} shows an example. The key insight
is that both \rss{} and strict serializability guarantee equivalence to a
sequence in the service's specification, which by definition is strictly
serializable. Second, Theorem~\ref{thm:invariants:rss} shows that the final
states of the two executions related by Lemma~\ref{thm:executions:rss} are
identical. It follows that invariants that hold in the first execution also hold
in the second.

\begin{lem}
  \label{thm:executions:rss}
  Suppose $\alpha$ is an execution of application $A$ that satisfies \rss{}.
  Then there is an equivalent execution $\beta$ of $A$ that satisfies strict
  serializability.
\end{lem}

\begin{proof}
  The proof proceeds in two steps. First, we construct a schedule $\beta^\prime$
  from $\alpha$'s schedule $\alpha^\prime$ without reordering any
  actions at any of the processes. Second, we construct the execution $\beta$
  from $\beta^\prime$ by inserting the states.

  \textit{Step 1.} Since $\alpha$ satisfies \rss{}, there exists a sequence
  $S \in \spec$ such that $<_S$ respects $\caused_{\alpha}$ and thus
  $\caused_{\alpha^\prime}$. To get $\beta^\prime$, we reorder $\alpha^\prime$
  such that each action is ordered after the maximal (as defined by $<_S$)
  invocation or response action that causally precedes it. To do so, we define
  three relations.
  
  First, let $\pi_1 \prec \pi_2$ if there is some invocation or response $\pi_3$
  that causally precedes $\pi_2$ and that is strictly greater (by $<_S$) than all
  invocations and responses that causally precede $\pi_1$. Second, let $\pi_1
  \equiv \pi_2$ if $\pi_1 \not\prec \pi_2$ and $\pi_2 \not\prec \pi_1$. Third,
  let $<_{\alpha^\prime}$ be the total order of actions defined by
  $\alpha^\prime$. Then $\beta^\prime$ is the schedule found by ordering the
  actions such that $\pi_1 <_{\beta^\prime} \pi_2$ if and only if $\pi_1 \prec
  \pi_2$ or $\pi_1 \equiv \pi_2$ and $\pi_1 <_{\alpha^\prime} \pi_2$.

  We show $\alpha^\prime|P_i = \beta^\prime| P_i$ for all $P_i$ by
  contradiction, so assume some pair of actions $\pi_1,\pi_2$ from the same
  $P_i$ were reordered in $\beta^\prime$. Without loss of generality, assume
  $\pi_2$ is ordered before $\pi_1$ in $\alpha^\prime$ but the reverse is true
  in $\beta^\prime$. It is clear that $\pi_1 \not\equiv \pi_2$ because otherwise
  $\pi_1$ and $\pi_2$ would be ordered identically in $\alpha^\prime$ and
  $\beta^\prime$. Thus, it must be that $\pi_1 \prec \pi_2$.

  Since $\pi_1 \prec \pi_2$, there must be some invocation or response $\pi_3$
  that causally precedes $\pi_2$ and is greater than those that causally precede
  $\pi_1$. But since $\pi_1$ and $\pi_2$ are from the same process and $\pi_2
  <_{\alpha^\prime} \pi_1$ by assumption, $\pi_3 \caused_{\alpha^\prime} \pi_1$
  by the transitivity of $\caused_{\alpha^\prime}$, contradicting the strictness
  in the definition of $\pi_3$. Thus, $\alpha^\prime$ must be equivalent to
  $\beta^\prime$.

  \textit{Step 2.} To get the execution $\beta$ from $\beta^\prime$, we must
  define the processes' states. Since the order of each process's actions
  is the same in $\alpha^\prime$ and $\beta^\prime$, each process will proceed
  through the same sequence of states. \newtext{Thus, we can construct $\beta$'s states
  from the sequences of each process's states in $\alpha$.}

  To conclude, we show that $\beta$ satisfies the stated properties. Since
  $\alpha^\prime|P_i = \beta^\prime| P_i$ for all $P_i$, it is clear that $\alpha$ is equivalent to $\beta$. Further, because we only reordered the states and
  actions in $\alpha$ to get $\beta$, $\beta$ is clearly finite, and because
  $\caused_\alpha$ captures the sending and receiving of messages, $\beta$ is
  well-formed. Finally, since $S$ is a sequence of matching invocation-response
  pairs, the processes' interactions with the service in $\beta$ are sequential,
  not overlapping in real time. Thus, $\beta$ satisfies strict serializability.
\end{proof}

\begin{thm}
  \label{thm:invariants:rss}
  Suppose $\invariant{A}$ is an invariant that holds for any execution $\beta$
  of $A$ that satisfies strict serializability. Then $\invariant{A}$ also
  holds for any execution $\alpha$ of $A$ that satisfies \rss{}.
\end{thm}

\begin{proof}
  Let $\alpha$ be an arbitrary execution of $A$ that satisfies \rss{}. We must
  show that $\invariant{A}$ is true for the final state $s$ of $\alpha$.

  By Lemma~\ref{thm:executions:rss}, there is an equivalent execution $\beta$
  that satisfies strict serializability. Let $s^\prime$ be the final state of
  $\beta$. Because $\alpha|P_i = \beta|P_i$ for all $P_i$, it is easy to see
  that $s^\prime = s$. By assumption, $\invariant{A}$ is true of $s^\prime$,
  so $\invariant{A}$ is also true of $s$.
\end{proof}

We prove similar results for \rs{} and linearizability in
Appendix~\ref{sec:proof}.

\section{Practical Implications}
\label{sec:discussion}

Lemma~\ref{thm:executions:rss} shows we can transform any \rss{} execution into
an equivalent strictly serializable one. Theorem~\ref{thm:invariants:rss} shows
this is sufficient for \rss{} to maintain application invariants.

While this transformation preserves the order of each process's actions, however, the order of causally
unrelated actions, e.g., the order of Alice and Bob's Web requests handled by
different servers, may not be. In fact, this is why anomalies like
$\anomaly{2}$ and $\anomaly{3}$ are possible with \rss{} and \rs{}.

\newtext{
Further, \rss{} (\rs{}) is defined with respect to a potentially composite
service. The results above thus assume a set of distinct services can together guarantee \rss{} (\rs{}), even if processes interact with multiple services, but they do not specify how this is achieved. }

\newtext{
In the remainder of this section, we first describe how
multiple services can be composed such that their composition guarantees \rss{} (\S\ref{sec:discussion:composition}). We then discuss supporting applications whose processes interact via message passing (\S\ref{sec:discussion:mp}).  For ease of exposition, the
discussion focuses on \rss{} but applies equally to \rs{}.
}

\subsection{Composing \rss{} Services}
\label{sec:discussion:composition}

A set of \rss{} services must
always ensure the values returned by their transactions reflect a
global total order spanning all services. This is straightforward with
strictly serializable services because real-time order is
universal across services.

\newtext{
With \rss{}, however, some pairs of transactions, such as causally unrelated read-only
transactions, may be reordered with respect to real time. As a result, the states observed by processes as they interact with  multiple services can form cycles (e.g., $P_1$ reads $x=1$ then $y=0$ while $P_2$ reads $y=1$ then $x=0$), precluding a total order.
Service builders thus must implement one additional mechanism,
\textit{\rtbarriers{}}, to allow a set of \rss{} services to globally guarantee \rss{}.
}

\newtext{
A \rtbarrier{} $f_x$ at \rss{} service $x$ provides the
following guarantee: for each pair of transactions $T_1$ and $T_2$ at service $x$, if $T_1 \caused f_x$ and $f_x \rt T_2$, then $T_1 <_{S_x} T_2$, where $<_{S_x}$ is the total order of $x$'s transactions.  Every transaction that causally precedes the \barrier{} must be serialized before any transaction that follows the \barrier{} in real time. Intuitively, a process that issues a \rtbarrier{} ensures all other processes observe state that is at least as new as the state it observed. Thus, if each process issues a \barrier{} at its previous service before interacting with another, the \barriers{} prevent cycles in the states observed by multiple processes as they cross service boundaries. (We discuss the service-specific implementation of \rtbarriers{} for \spannerrss{} and \gryffrs{} in Sections~\ref{sec:spannerrss} and~\ref{sec:gryffrs}.)
}

\begin{figure}[t]
  \centering
  \input{figures/librss-interface}
  \caption{\newtext{\librss{} Interface. \librss{} helps \rss{} service builders implement composition by invoking the necessary \rtbarriers{}.}}
  \label{fig:librss-interface}
\end{figure}

\newtext{
Although the need to implement a \barrier{} for each \rss{} service places an additional burden on service builders, using \rtbarriers{} to guarantee a global total order across services
does not require changes to applications. The client libraries of the \rss{} services can insert \rtbarriers{} as necessary at run time. To this end, we implement a meta-library, \librss{}, to aid service builders with composition. Figure~\ref{fig:librss-interface} shows its interface. 
}

\newtext{
At initialization, an \rss{} service's client library registers itself with the \librss{} meta-library,
passing it a unique name and a callback that implements its \barrier{}.
The meta-library keeps an in-memory registry of all \rss{} services. During execution, the client library
must simply notify the meta-library before starting a new transaction.
}

\newtext{
With
these calls, the meta-library implements composition without intervention from application programmers. Every time an \rss{} client starts a
transaction, the meta-library checks if the transaction is at the same service as the previous one, if any. If not, \librss{} invokes the prior service's \barrier{}. In Appendix~\ref{sec:proof:composition},
we prove that if each service's \rtbarrier{} provides the guarantee described above and \librss{} follows this simple protocol, then the composition of a set of \rss{} services globally guarantees \rss{}.
}

\subsection{Capturing Causality}
\label{sec:discussion:mp}

\newtext{
A meta-library that issues \rtbarriers{} is sufficient to guarantee \rss{} for applications whose processes
interact solely through a set of \rss{} services. But for those whose processes also interact through message passing,
an \rss{} service must ensure causality is respected across these interactions.
}

\newtext{
For instance, recall our photo-sharing application and assume Alice is using her browser, which sends requests to Web servers that interact with an \rss{} key-value store. If one server reads and transmits a photo to Alice's browser and the browser subsequently reads the same photo via a second server, the key-value store must ensure causality is respected across the two transactions---the second must not return \texttt{null}. But if the store is unaware of the causal constraint between the two read-only transactions, then this may not be guaranteed.}

\newtext{
One approach is to require application processes to issue a \barrier{}
before such out-of-band interactions. For instance, the Web server must issue one before transmitting the response back to Alice's browser.
Depending on the structure of the
application, however, this may be inefficient.
}

\newtext{
A better approach is to use a
context propagation framework~\cite{mace2018baggageContexts} to pass metadata
between the interacting processes. This would ensure the second Web server has the
necessary metadata to convey causality before it interacts with the \rss{}
store. This context must also include the name of the last \rss{} service the process interacted with, so \librss{} can correctly implement composition.
}


%% file: figures/librss-interface.tex
\setlength{\tabcolsep}{0.5\tabcolsep}
\begin{tabular}{@{}l l@{}}
\toprule
\textbf{Function} & \textbf{Description}
\\ \midrule
$\textsc{RegisterService}(\text{name}, \text{\barrier{}\_f})$ & Register new service.
\\ \addlinespace[1mm]
$\textsc{UnregisterService}(\text{name})$ & Unregister service.
\\ \addlinespace[1mm]
$\textsc{StartTransaction}(\text{name})$ & Start txn at service.
\\ \bottomrule
\end{tabular}

%% file: sections/spanner.tex
\section{\spannerrss{}}
\label{sec:spannerrss}


\spanner{} is a globally distributed, transactional
database~\cite{corbett2013spanner}. It uses synchronized clocks to guarantee
strict serializability~\cite{papadimitriou1979serializability}. While
\spanner{} is designed to provide low-latency read-only (RO) transactions most
of the time, they may block, increasing tail latency significantly. \newtext{Such increases in the tail latency of low-level services can translate into increases in common-case, user-visible latency~\cite{dean2013tail}.}

Our variant of \spanner{}'s protocol, \spannerrss{}, improves tail latency
for RO transactions by relaxing the constraints on read-only transactions in accordance with \rss{}. \newtext{(We prove it provides \rss{} in Appendix~\ref{sec:correctness:spannerrss}.)}

\noindentparagraph{Spanner background.}
\spanner{} is a multi-versioned key-value
store. Keys are split across many shards, and shards are replicated using
Multi-Paxos~\cite{lamport1998paxos}. Clients atomically read and write keys at
multiple shards using transactions.

\spanner{}'s read-write (RW) transactions use two-phase locking~\cite{berstein1987ccbook} and a variant of two-phase
commit~\cite{gray1978dbos}. Each shard's Paxos leader serves as a participant or
coordinator. Further, using its TrueTime API, \spanner{} gives each
transaction a commit timestamp that is guaranteed to be between the
transaction's real start and end times.

During execution, clients acquire read locks and buffer writes. To commit, the client chooses a coordinator and sends its writes to the shards.
Each participant then does the following: (1) ensures the transaction still
holds its read locks, (2) acquires write locks, (3) chooses a prepare
timestamp, (4) replicates the prepare success, and (5) notifies the
coordinator. Assuming all participants succeed, the coordinator finishes
similarly: It checks read locks, acquires write locks, chooses the
transaction's commit timestamp, and replicates the commit success. Finally, the
coordinator releases its locks and sends the outcome to the client and
participants.

To guarantee strict serializability, each participant ensures its prepare
timestamp is greater than the timestamps of any previously committed or
prepared transactions. The coordinator chooses the commit timestamp similarly
but also ensures it is greater than the transaction's start time and greater
than or equal to all of the prepare timestamps. Combined with commit
wait~\cite{corbett2013spanner}, this ensures the transaction's commit timestamp
is between its start and end times.

\begin{figure}[!t]
\centering
\includegraphics[width=0.8\linewidth,page=4]{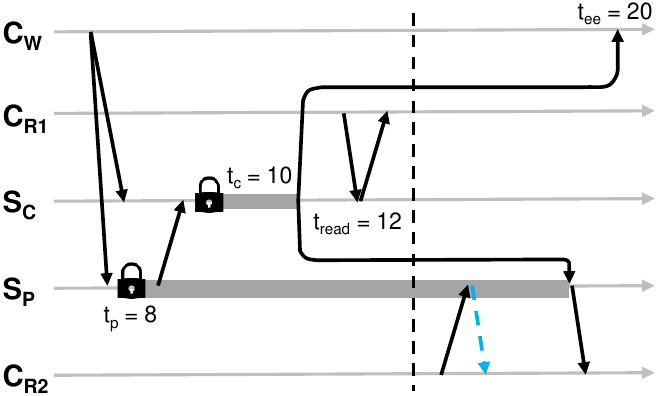}
\caption{Example execution where \spannerrss{}'s RO transaction returns before
  \spanner{}'s. (Replication is omitted.) Client $C_W$ is committing writes
  to two shards, $S_C$ and $S_P$; $S_C$ is the coordinator. $C_{R1}$ reads
  $C_W$'s writes before $C_W$'s transaction finishes. Strict
  serializability still mandates $C_{R2}$'s read also includes them. Conversely,
  with \spannerrss{}, $C_{R2}$'s read returns immediately (shown by
  blue dotted line).}
\label{fig:spanner-execution}
\end{figure}

\noindentparagraph{Strict serializability unnecessarily blocks ROs.}
Many workloads are dominated by
reads~\cite{shute2013f1,pang2019zanzibar,bronson2013tao}. Thus, \spanner{} also
includes an optimized RO transaction protocol to make the majority of
transactions as fast as possible. \spanner{}'s RO transactions are strictly
serializable but unlike RW transactions, only require one round trip between a
client and the participant shards. As a result, RO transactions have
significantly lower latency than RW transactions.


RO transactions, however, must sometimes block to ensure strict serializability.
RO transactions in \spanner{} read at a client assigned timestamp $\tread =
\textit{TT.now.latest}$, which TrueTime guarantees is after
$\tro.\textit{start}$. When a read arrives at a shard with $\tread$ greater than
the prepare timestamp of some conflicting RW transaction $\trw$, it must block
until the shard learns if $\trw$ commits at some time $\tcommit$ or aborts.
Otherwise, $\tro$ risks violating strict serializability: if $\tcommit < \tread$
because $\trw.\textit{end} < \tro.\textit{start}$ ($\tcommit < \trw.\textit{end}
< \tro.\textit{start} < \tread$), then strict serializability mandates that
$\tro$ includes $\trw$'s writes.
Because they must potentially wait while a RW transaction executes two-phase
commit, blocked RO transactions can have significantly higher latency.

One potential optimization to improve RO transaction latency would be to include
the earliest client-side end time $\elb$ for each RW transaction. Then,
RO transactions could avoid blocking if $\tread < \elb$.  Unfortunately, strict
serializability disallows this optimization because it requires $\tro$ to
observe $\trw$ even before $\elb$ if there is some other RO transaction that finishes before $\tro$ and includes any of $\trw$'s writes. 

\newtext{
Figure~\ref{fig:spanner-execution} shows an example. Because $C_{R1}$'s read observes $C_W$'s RW transaction at $S_C$, strict serializability requires all future reads at both shards to include $C_W$'s writes. Thus, $C_{R2}$'s read must block until $C_W$'s RW transaction commits.}

\newtext{
In contrast, \rss{} allows this optimization. $C_{R1}$'s transaction only imposes a constraint on reads that causally follow it, so $C_{R2}$'s read may still return an older value. 
}

\begin{algorithm}[!t]
  \caption{\spannerrss{} Client}
  \label{alg:spannerrss-client}
  \input{algorithms/spanner-client}
\end{algorithm}

\noindentparagraph{\spannerrss{}.} \spannerrss{} is our variant of 
\spanner{} that improves tail RO transaction latency by often avoiding
blocking, even when there are conflicting RW transactions. Intuitively,
a RO transaction can avoid blocking by observing a state of the database
as of some timestamp $\tsnap$ that is before its read timestamp $\tread$ if it
can infer the state satisfies \rsslong{}.

Observing this state from before $\tread$ is safe under \rss{} when three
conditions are met: (1) there are no unobserved writes from a \newtext{conflicting} RW transaction
that could have ended before $\tro$ started; (2) there are no causal constraints
that require $\tro$ to observe a write at a timestamp later than $\tsnap$; and
(3) \newtext{its results are consistent with a sequential execution of transactions.}

\begin{algorithm}[!t]
  \caption{\spannerrss{} Shard}
  \label{alg:spannerrss-shard}
  \input{algorithms/spanner-shard}
\end{algorithm}

\spannerrss{} ensures each of these conditions are met. To ensure (1), RW
transactions include a client-side earliest end time $\elb$. To ensure (2), RO
transactions include a minimum read time $\tmin$.  Finally, to ensure (3),
before completing a RO transaction, clients ensure all returned values reflect
precisely the state of the database at $\tsnap$.
Algorithms~\ref{alg:spannerrss-client} and~\ref{alg:spannerrss-shard} show the
full protocol.

\textit{Estimating, including, and enforcing $\elb$ for RW transactions.} Each
RW transaction includes an earliest client-side end time $\elb$. The client
estimates $\elb$ and includes it when it initiates two-phase commit (not
shown). The shards then store $\elb$ while the transaction is prepared but not
yet committed or aborted (Alg.~\ref{alg:spannerrss-shard}, line 1). The client later ensures $\elb$ is
less than the actual client-side end time by waiting until
$\elb < \textit{TT.now.earliest}$.



\textit{Enforcing a minimum timestamp for RO transactions.} Each RO
transaction includes a minimum read timestamp $\tmin$ to ensure it obeys any necessary
causal constraints. Each client tracks this minimum timestamp and updates it
after every transaction to include new constraints. After a RW
transaction, it is set to the transaction's commit timestamp (not shown). After
a RO transaction, it is set to be at least the transaction's snapshot time
(Alg.~\ref{alg:spannerrss-client}, line 12). $\tmin$ thus captures the causal constraints on this RO
transaction; it must observe a state that is at least as recent as its last
write and \newtext{any writes the client previously observed.}

\textit{Avoiding blocking on shards with $\elb$ and $\tmin$.} Using $\elb$ and
$\tmin$, shards can infer when it is safe for a RO transaction to skip
observing a prepared-but-not-committed RW transaction (Alg.~\ref{alg:spannerrss-shard}, line 6). It
is safe unless the prepared transaction either must be observed due to a causal
constraint ($\tprepare \leq \tmin$) or could have ended before the RO
transaction began ($\elb \leq \tread$).

\textit{Reading at $\tsnap$.} Although each shard can now infer when a RO
transaction can safely skip a prepared RW transaction, the values returned by
multiple shards may not necessarily reflect a complete, consistent snapshot at
$\tsnap$. Thus, clients and shards take four additional steps to ensure a client
always returns a consistent snapshot.

First, as in \spanner{}, a shard waits to process a RO transaction until its
Paxos safe time is greater than $\tread$ (Alg.~\ref{alg:spannerrss-shard}, line
4)~\cite{corbett2013spanner}. As a result, all future Paxos writes, and thus all
future RW prepare timestamps, will be larger than $\tread$. Thus, the shard
ensures it is returning information that is valid until at least $\tread$. (The Paxos safe time at leaders can be advanced immediately if it is within the leader's lease.) Second, shards include the commit timestamps
$\tcommit$ for the returned values (Alg.~\ref{alg:spannerrss-shard}, lines 2, 8, and 10). Third, they return
the prepare timestamp $\tprepare$ for each skipped RW transaction with $t_p \leq
\tread$ (Alg.~\ref{alg:spannerrss-shard}, lines 9-10). (Because writes use locks, there
is at most one per key.) Fourth, when a skipped RW transaction commits, a shard
sends the commit timestamp and the written values in a slow path (Alg.~\ref{alg:spannerrss-shard},
lines 13-15).

Before returning, the client examines the commit and prepare timestamps to see
if the shards returned values that are all valid at some snapshot time.
Specifically, it sets $\tsnap$ to the earliest time for which it has a value for
all keys (Alg.~\ref{alg:spannerrss-client}, lines 15-20). Then, it sees if any prepared transactions have timestamps less than $\tsnap$ (Alg.~\ref{alg:spannerrss-client},
lines 22-23). If they all prepared after $\tsnap$, the RO transaction
returns immediately (Alg.~\ref{alg:spannerrss-client}, line 13).

If some transaction prepared before $\tsnap$, however, the client must wait for slow replies from
the shards (Alg.~\ref{alg:spannerrss-client}, lines 9-10). As the client learns of commits and aborts through
the slow replies, it moves transactions out of the prepared set (Alg.~\ref{alg:spannerrss-client}, line
11), updates the values it will return (if $\tcommit \leq \tsnap$), and
potentially advances the earliest prepared timestamp (Alg.~\ref{alg:spannerrss-client}, line 22). Note
that $\tsnap$ remains the same, so the latter continues until $\tsnap <
t^\prime_\textrm{p}$, which is guaranteed by the time the final slow reply is
received.

\textit{Performance discussion.} \spannerrss{}'s RO transaction latency is never
worse and often better than \spanner{}'s. When there are no conflicting RW
transactions, RO transactions in both will return consistent results at
$\tread$. When there are conflicting transactions, however, \spannerrss{} will
often send fast replies quickly while \spanner{} blocks. Further, the fast
replies let \spannerrss{} complete the RO transaction right away unless one
of the shards returns a value with a commit timestamp that is greater than the
prepared timestamp of a skipped RW transaction. Even then, the slow replies from \spannerrss{}'s shards will be sent at the same
time \spanner{} would unblock.

\subsection{\RTBarriers{}}
\label{sec:spannerrss:barrier}

\newtext{
As described above, to ensure the order of transactions reflects causality, a client tracks and enforces a minimum read timestamp $\tmin$. Using $\tmin$, a client ensures its next transaction will be ordered after any transaction that causally precedes it by ensuring the next transaction reflects a state of the database that is at least as recent as $\tmin$.
}

\newtext{
A \rtbarrier{} must provide a slightly stronger guarantee. It must ensure that all transactions that causally precede it are serialized before any transaction that follows it in real time, regardless of the latter transaction's originating client. While this is guaranteed for future RW transactions since they already respect their real-time order, the same is not true of future RO transactions. Thus, when executed at a client with a minimum read timestamp $\tmin$, a \barrier{} in \spannerrss{} must ensure that all future RO transactions reflect a state that is at least as recent $\tmin$.
}

\newtext{
To achieve this, \spannerrss{}'s \rtbarriers{} leverage the following observation: If $L$ is the maximum difference between $\tcommit$ and $\elb$ for any RW transaction, then a RO transaction that starts after $\tcommit + L$ will reflect all writes with timestamps less than or equal to $\tcommit$. After $\tcommit+L$, a RO transaction cannot skip reading a RW transaction with commit timestamp $\tcommit$ (since $\elb \leq \tcommit + L < \tread$).
}

\newtext{
As a result, \barriers{} in \spannerrss{} are simple. To ensure all future RO transactions reflect a state that is at least as recent as $\tmin$, a \barrier{} blocks until $\tmin + L < \textit{TT.now.earliest}$.}


%% file: algorithms/spanner-client.tex
\begin{algorithmic}[1]
  \State \Global $\tmin \gets 0$
  \Function{Client::ROTransaction}{$K$}
  \State $S \gets \Call{ShardLeaders}{K}$
  \State $\tread \gets \Call{TrueTime::Now.Latest}{}$
  \State \Send \Msg{ROCommit}{$K, \tread, \tmin$} to all $s \in S$
  \State \WaitUntil receive \Msg{ROFastReply}{$P_s, V_s$} from all
  $s \in S$
  \State $P, V \gets \bigcup_{s \in S} P_s, \bigcup_{s \in S} V_s$
  \State $\tsnap \gets \Call{CalculateSnapshotTS}{K, V}$
  \While{$\Call{CheckSnapshot}{P, \tsnap} \neq \text{COMMIT}$}
  \State \WaitUntil for \Msg{ROSlowReply}{$i, d, \tcommit, V^\prime$} from $s \in S$
  \State $P, V \gets \Call{UpdatePrepared}{P, V, i, d, \tcommit, V^\prime}$
  \EndWhile
  \State $\tmin \gets \max\left(\tmin, \tsnap\right)$
  \State \Return $\Call{ReadAtTimestamp}{V, \tsnap}$
  \EndFunction
  \Statex
  \Function{Client::CalculateSnapshotTS}{$K, V$}
  \State $\tsnap \gets 0$
  \For{$k \in K$}
  \State $V^\prime \gets \{(\tcommit, k^\prime, v) \in V : k = k^\prime\}$
  \State $t_{\text{earliest}} \gets \min_{(\tcommit,k^\prime,v) \in V^\prime} \tcommit$
  \State $\tsnap \gets \max\left(\tsnap, t_{\text{earliest}}\right)$
  \EndFor
  \State \Return $\tsnap$
  \EndFunction
  \Statex
  \Function{Client::CheckSnapshot}{$P, \tsnap$}
  \State $\tprepare^\prime \gets \min_{(i, \tprepare) \in P} \tprepare$
  \IfThenElse{$\tprepare^\prime \leq \tsnap$}{$d \gets \text{WAIT}$}{$d \gets \text{COMMIT}$}
  \State \Return $d$
  \EndFunction
\end{algorithmic}


%% file: algorithms/spanner-shard.tex
\begin{algorithmic}[1]
  \State \Global $\mathcal{P} \gets \left\{\left(i, \ell, \tprepare, \elb,
      W\right),\ldots\right\}$ \Comment{Prepared txns}
  \State \Global $\mathcal{D} \gets \left\{\left(\tcommit, k, v\right),\ldots\right\}$ \Comment{Database}
  \Function{Shard::ROCommitRecv}{$c, K, \tread, \tmin$}
  \State $\WaitUntil \text{ until } \tread \leq \Call{Paxos::MaxWriteTS}{}$
  \State $P \gets \left\{(i, \ell, \tprepare, \elb, W) \in \mathcal{P} \mid R \cap W
    \neq \emptyset \land \tprepare \leq \tread \right\}$
  \State $B \gets \left\{(i, \ell, \tprepare, \elb, W)
    \in P \mid \tprepare \leq \tmin \lor \elb \leq \tread \right\}$
  \State \WaitUntil until all $p \in B$ commit or abort
  \State $V \gets \Call{ReadAtTimestamp}{\mathcal{D}, K, \tread}$
  \State $Q \gets \{(i,\tprepare) : (i, \ell, \tprepare, \elb, W) \in P \setminus B \}$
  \State \Send \Msg{ROFastReply}{$Q, V$} to $c$
  \While{$P \neq \emptyset$}
  \State \WaitUntil until some $p \in P$ commits or aborts
  \If{$p = (i, \ell, \tprepare, \elb, W)$ commits at $\tcommit$}
  \State $V \gets \Call{ReadAtTimestamp}{\mathcal{D}, K \cap W, \tcommit}$
  \State \Send \Msg{ROSlowReply}{$i, \text{COMMIT}, \tcommit, V$} to $c$
  \Else
  \State \Send \Msg{ROSlowReply}{$i, \text{ABORT}, 0, \emptyset$} to $c$
  \EndIf
  \State $P \gets P \setminus \{p\}$
  \EndWhile
  \EndFunction
\end{algorithmic}


%% file: sections/eval.tex
\begin{figure*}[!t]
  \centering
  \begin{subfigure}[t]{0.3\linewidth}
    \includegraphics[height=1.5in]{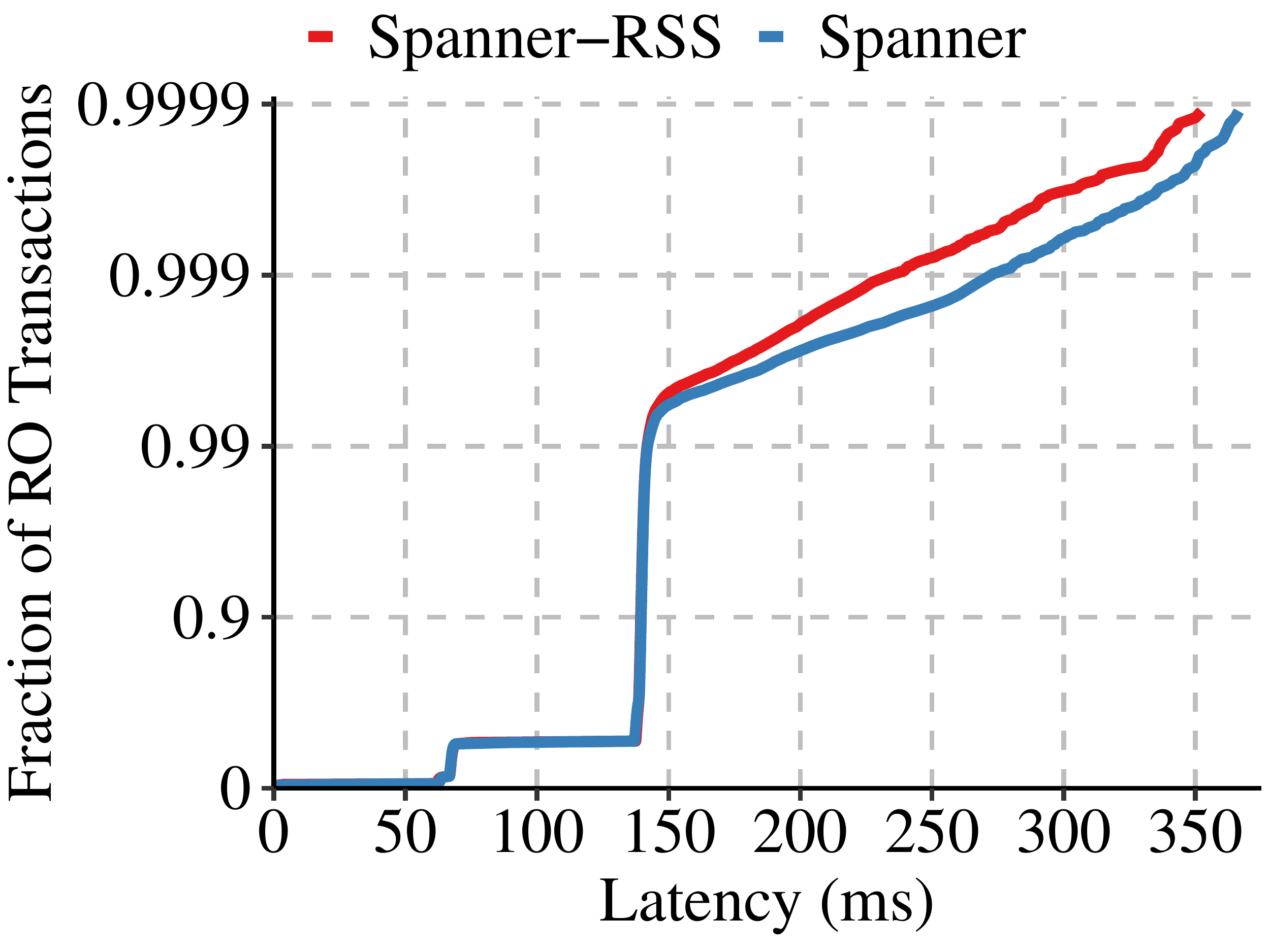}
    \caption{0.5 skew.}
    \label{fig:spanner-ro-log-lat-skew-0.5}
  \end{subfigure}
  \hspace{.2in}
  \begin{subfigure}[t]{0.3\linewidth}
    \centering
    \includegraphics[height=1.5in]{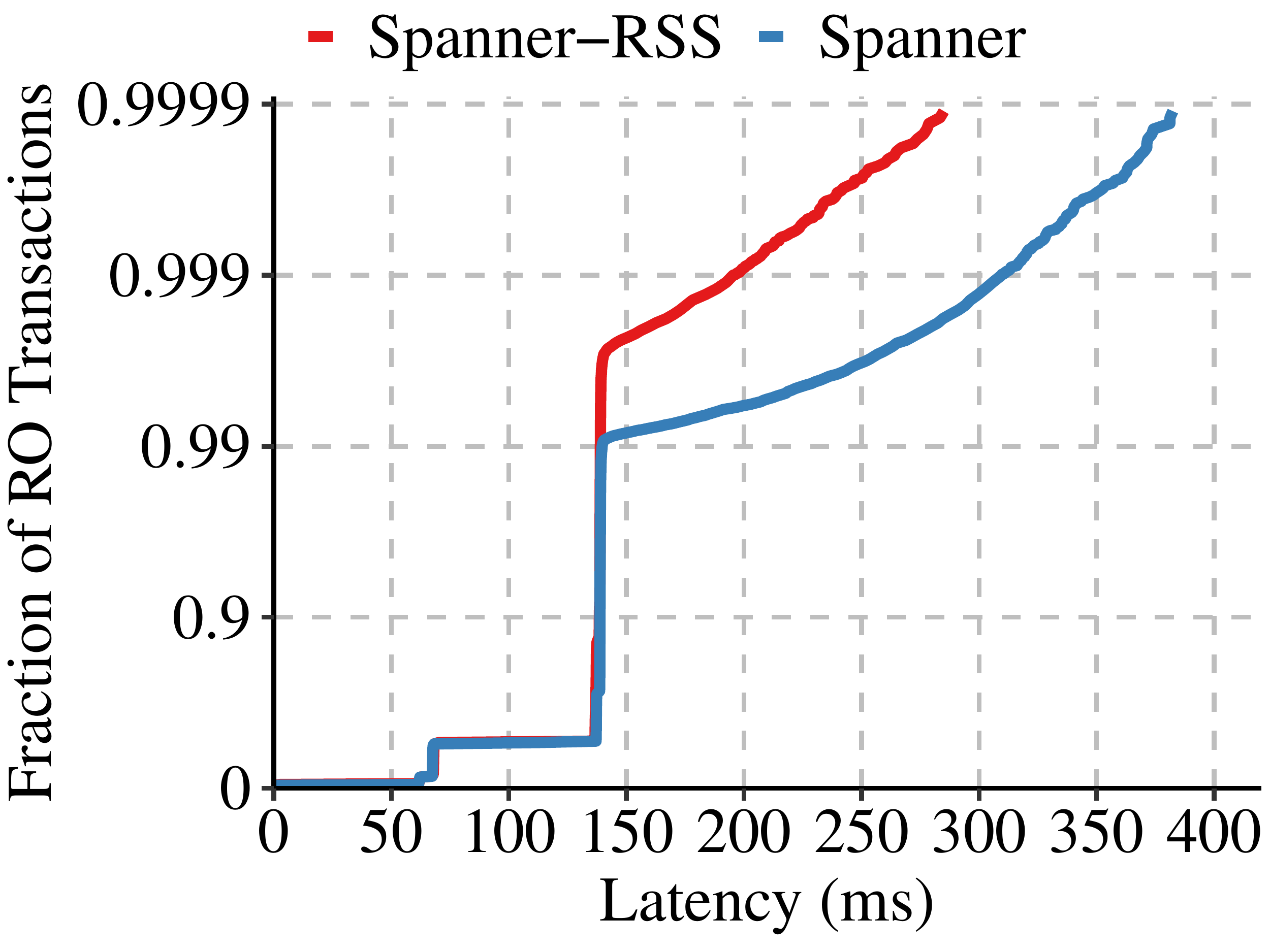}
    \caption{0.7 skew.}
    \label{fig:spanner-ro-log-lat-skew-0.7}
  \end{subfigure}
  \hspace{.2in}
  \begin{subfigure}[t]{0.3\linewidth}
    \centering
    \includegraphics[height=1.5in]{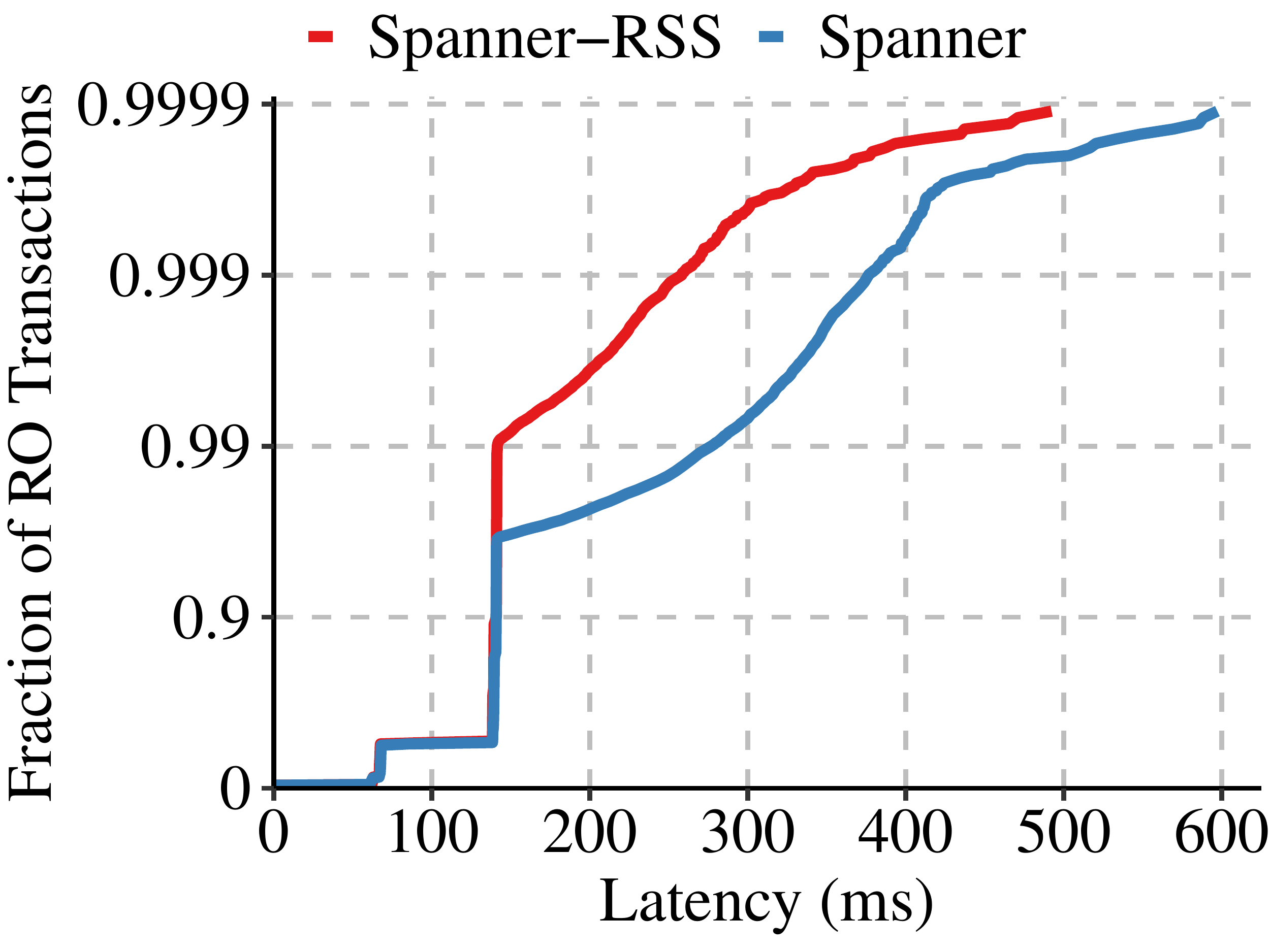}
    \caption{0.9 skew.}
    \label{fig:spanner-ro-log-lat-skew-0.9}
  \end{subfigure}
  \caption{\spannerrss{} offers better tail latency for RO transactions on
    Retwis. In contrast to \spanner{}'s, its RO transactions can often avoid
    blocking when there are concurrent, conflicting RW transactions.}
  \label{fig:spanner-ro-log-lat-skew}
\end{figure*}

\section{\spannerrss{} Evaluation}
\label{sec:eval:spanner}

Our evaluation of \spannerrss{} aims to answer two questions: Does \spannerrss{}
improve tail latency for read-only transactions
(\S\ref{sec:eval:spanner:latency}), and what performance overhead does
\spannerrss{}'s read-only transaction protocol impose
(\S\ref{sec:eval:spanner:throughput})?

We implement the \spanner{} and
\spannerrss{} protocols in C++ using TAPIR's experimental
framework~\cite{zhang2015tapir,zhang2018tapir}. Each shard is single-threaded.
The implementation reuses TAPIR's implementation of view-stamped
replication~\cite{oki1988vr} instead of Multi-Paxos~\cite{lamport1998paxos} but is otherwise faithful. Our code and experiment scripts are available online~\cite{spannerrssrepo}.

\newtext{The implementation includes two optimizations not presented in
Section~\ref{sec:spannerrss}: First, a skipped, prepared transaction's writes are
returned in the fast path instead of the slow path, allowing the client to
return faster in some cases, e.g., if it learns the transaction already committed at a different shard.}

\newtext{Second, when a transaction blocks as part of wound-wait~\cite{rosenkrantzSystemLevel1978}, it estimates how long it blocked and advances its local estimate of $\elb$ by that amount. The coordinator then aggregates the shards' $\elb$ values and returns the maximum to the client, which waits until it has passed. This reduces the chance a RO transaction will be blocked by a RW transaction whose $\elb$ has become inaccurate because of lock contention.}

\newtext{Unless otherwise specified, experiments ran on Amazon's EC2 platform~\cite{ec2}.
Each \texttt{t2.large} instance has 2 vCPUS and \SI{8}{\giga\byte} RAM.
We use three shards with three replicas each. One shard leader is in each of California, Virginia, and
Ireland, and the replicas are in the other two data centers. The round trip times are as
follows: CA-VA is \SI{62}{\milli\second}, CA-IR is \SI{136}{\milli\second}, and
VA-IR is \SI{68}{\milli\second}. Our emulated TrueTime error is
\SI{10}{\milli\second}, the p99.9 value observed in
practice~\cite{corbett2013spanner}.}

\newtext{To calculate $\elb$ for RW transactions, clients use the round-trip latencies above. In our implementation, clients use them to calculate, for each set of
participants, the coordinator choice that yields the minimum commit latency. It
stores these choices and the commit latencies, and the latter is used to
calculate $\elb$. To avoid increasing RW latency, the round-trip latencies above are the minimum observed, and clients calculate $\elb$
with respect to \textit{TT.now.earliest}.}

Each client executes the Retwis workload~\cite{retwis} over a database of ten
million keys and values. Retwis clients execute transactions in the following
proportions: 5\% add-user, 15\% follow/unfollow, 30\% post-tweet, and 50\%
load-timeline. The first three are RW transactions, and the last is RO. We
generate keys according to a Zipfian distribution~\cite{hormann1996rejinv} with skews ranging from 0.5 to 0.9. \newtext{Such read-write ratios and skews are representative of real workloads~\cite{chen2020hotring,yang2020twitter}.}

\newtext{Unless otherwise specified, we generate load with a fixed number of partly open clients~\cite{schroeder2006openclosed}. Partly open clients use three parameters to model user behavior: sessions arrive at rate $\lambda$ according to a Poisson process; after each transaction in a session, the client chooses to stay with probability $p$; and if it does, it waits for a think time $H$. The clients use a separate $\tmin$ for each session. We set $H=0$ since this yields the worst performance for \spannerrss{}. Further, we set $p=0.9$, so the average session length is ten transactions, matching measurements from real deployments~\cite{shute2013f1}.
Finally, for each workload, we set $\lambda$ such that the offered load is 70-80\% of the maximum throughput. Each data center contributes an equal fraction of the load.}

\subsection{\spannerrss{} Reduces RO Tail Latency}
\label{sec:eval:spanner:latency}

We first compare the latency distributions for RO and RW transactions with
\spanner{} and \spannerrss{} as skew varies. \spannerrss{}'s RO
transactions have lower latency than \spanner{}'s due to less blocking
during conflicting RW transactions. These improvements do not
harm RW transaction latency because \spannerrss{}'s protocol simply
requires passing around an extra timestamp with RW transactions.

Figure~\ref{fig:spanner-ro-log-lat-skew} compares the tail latency distributions
of RO transactions at three skews. (We omit the distributions for RW
transactions after verifying they are identical.) \spannerrss{} improves
RO tail latency in all cases. When contention is low
(Figure~\ref{fig:spanner-ro-log-lat-skew-0.5}), \spanner{}'s RO transactions
offer low tail latency; \newtext{up to p99, their latency is bounded only by wide-area
latency. Above this, however, it starts increasing. At p99.9,
\spannerrss{} offers a \SI{14}{\percent} (\SI{38}{\milli\second})
reduction in tail latency.}

\newtext{\spannerrss{} offers larger improvements at higher skews. In Figure~\ref{fig:spanner-ro-log-lat-skew-0.7}, latency consistently
decreases by at least \SI{76}{\milli\second} above p99.5. This is up to a
\SI{45}{\percent} reduction; at p99.9, it is a \SI{37}{\percent}
(\SI{114}{\milli\second}) reduction.}

\newtext{With a skew of 0.9, (Figure~\ref{fig:spanner-ro-log-lat-skew-0.9}), \spanner{}'s RO transaction latency starts increasing at lower percentiles. As a result, \spannerrss{} reduces p99 RO latency by \SI{49}{\percent} (\SI{135}{\milli\second}). With high contention, however, improvements farther out on the tail (e.g., above p99.95) are more inconsistent. Increased
waiting by RW transactions for wound-wait~\cite{rosenkrantzSystemLevel1978} make
the earliest end time estimates less accurate. Further, each session's $\tmin$ advances more rapidly and in turn, increases the
chance a RO transaction must block.}

\subsection{\spannerrss{} Imposes Little Overhead}
\label{sec:eval:spanner:throughput}

We now compare the two protocol under heavy load to quantify the overhead
\spannerrss{} incurs from its additional protocol complexity. Because the number
and size of its additional messages is small, \spannerrss{}'s performance should
be comparable to \spanner{}'s.

\newtext{To stress the implementations, we use a uniform
workload, set the TrueTime error to zero, and place all shards in one data center.
Since it does not depend on wide-area latencies, we ran this experiment on
CloudLab's Utah platform~\cite{duplyakin2019cloudlab}. Each \texttt{m510} machines has
8 physical cores, \SI{64}{\giga\byte} RAM, and a \SI{10}{\giga\bit} NIC.
Inter-data-center latency is less than \SI{200}{\micro\second}. We
use eight shards, so each leader has a dedicated physical CPU on one server.}

\begin{figure}[!t]
    \centering
    \includegraphics[height=1.5in]{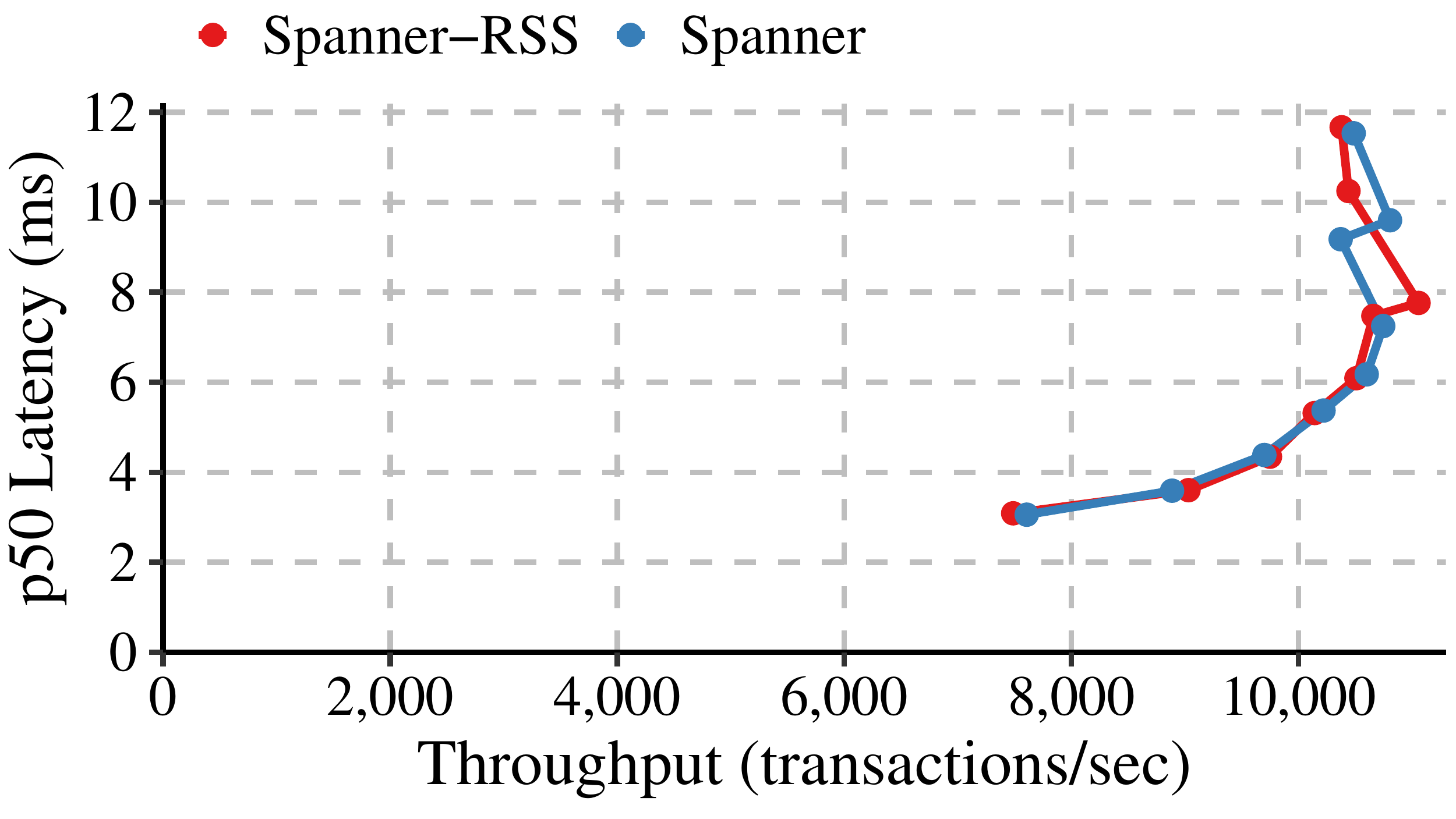}
    \caption{\spannerrss{} does not significantly impact \spanner{}'s
      performance at high load.}
    \label{fig:spanner-throughput}
\end{figure}

Figure~\ref{fig:spanner-throughput} compares the throughput and median latency
as we increase the number of \newtext{closed-loop} clients. As shown,
\spannerrss{} does not significantly impact the server's maximum throughput.
\spannerrss{}'s is within a few hundreds of transactions per second of
\spanner{}'s, and its latency is within a few milliseconds.


%% file: sections/gryff.tex
\begin{figure*}
    \centering
    \begin{subfigure}[b]{0.3\linewidth}
    \centering
    \includegraphics[height=1.5in]{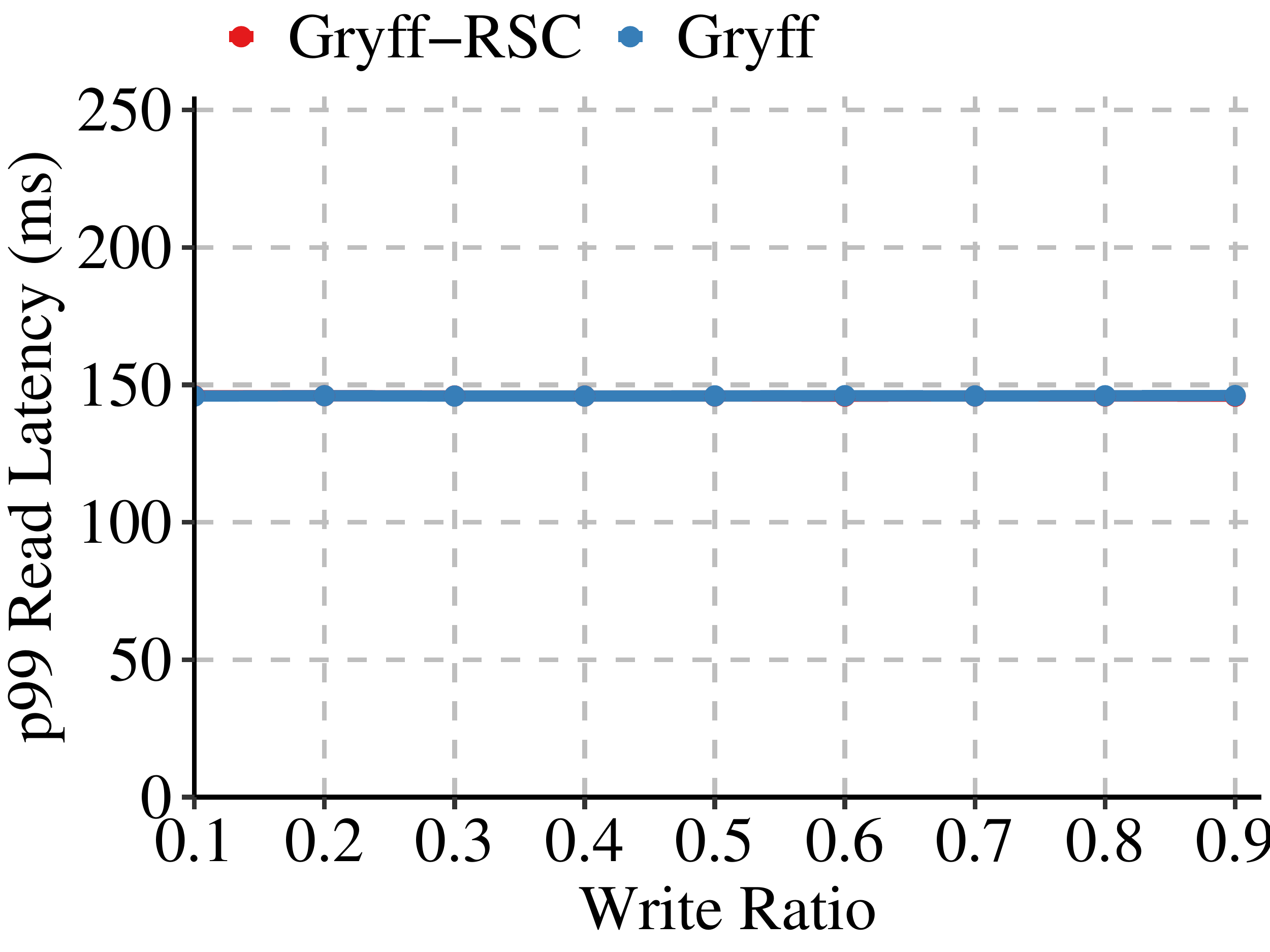}
    \caption{2\% conflicts.}
    \label{fig:gryff-p99-read-vs-write-ratio-confict-02}
    \end{subfigure}
    \hspace{.2in}
    \begin{subfigure}[b]{0.3\linewidth}
    \centering
    \includegraphics[height=1.5in]{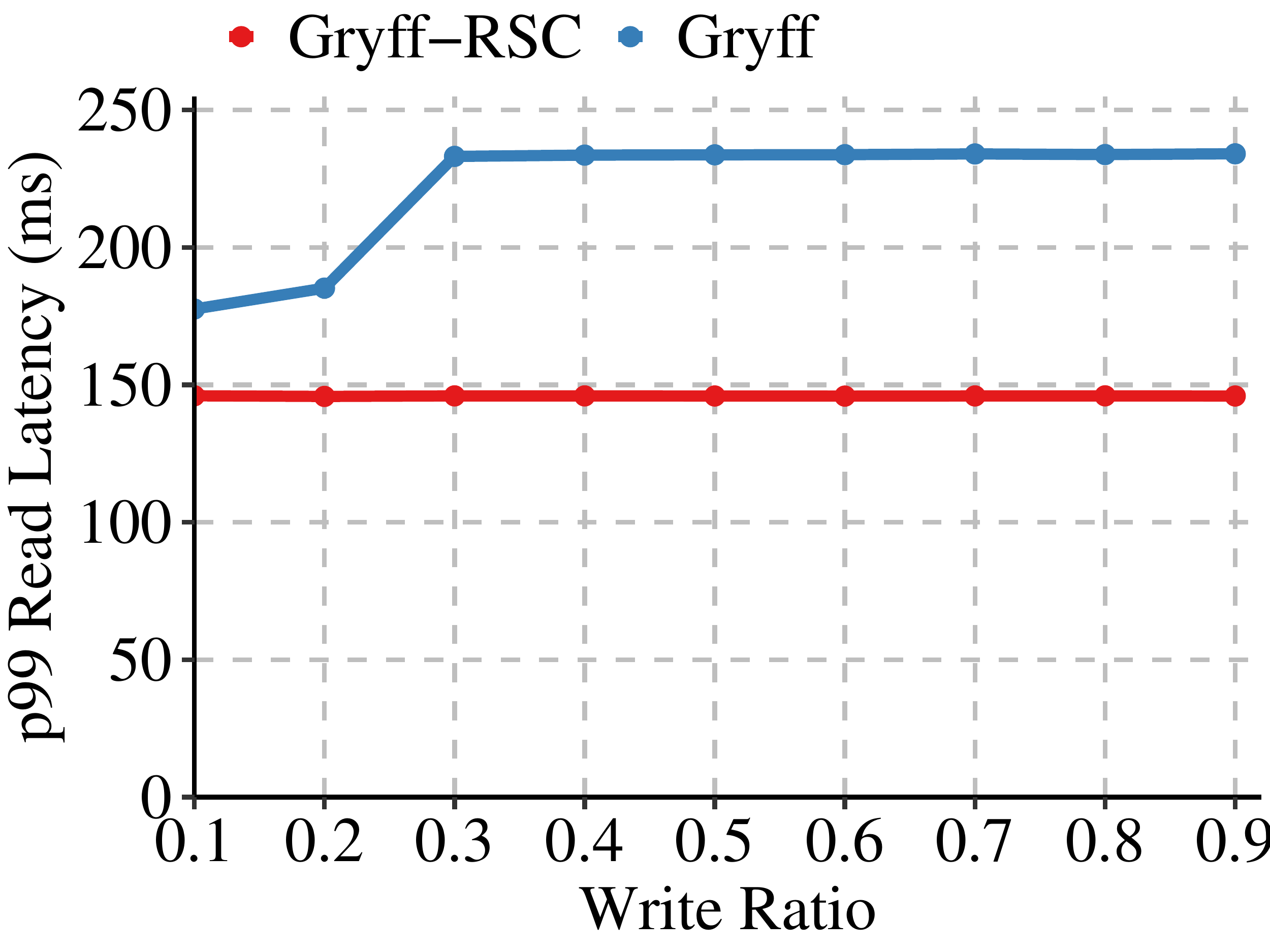}
    \caption{10\% conflicts.}
    \label{fig:gryff-p99-read-vs-write-ratio-confict-10}
    \end{subfigure}
    \hspace{.2in}
    \begin{subfigure}[b]{0.3\linewidth}
    \centering
    \includegraphics[height=1.5in]{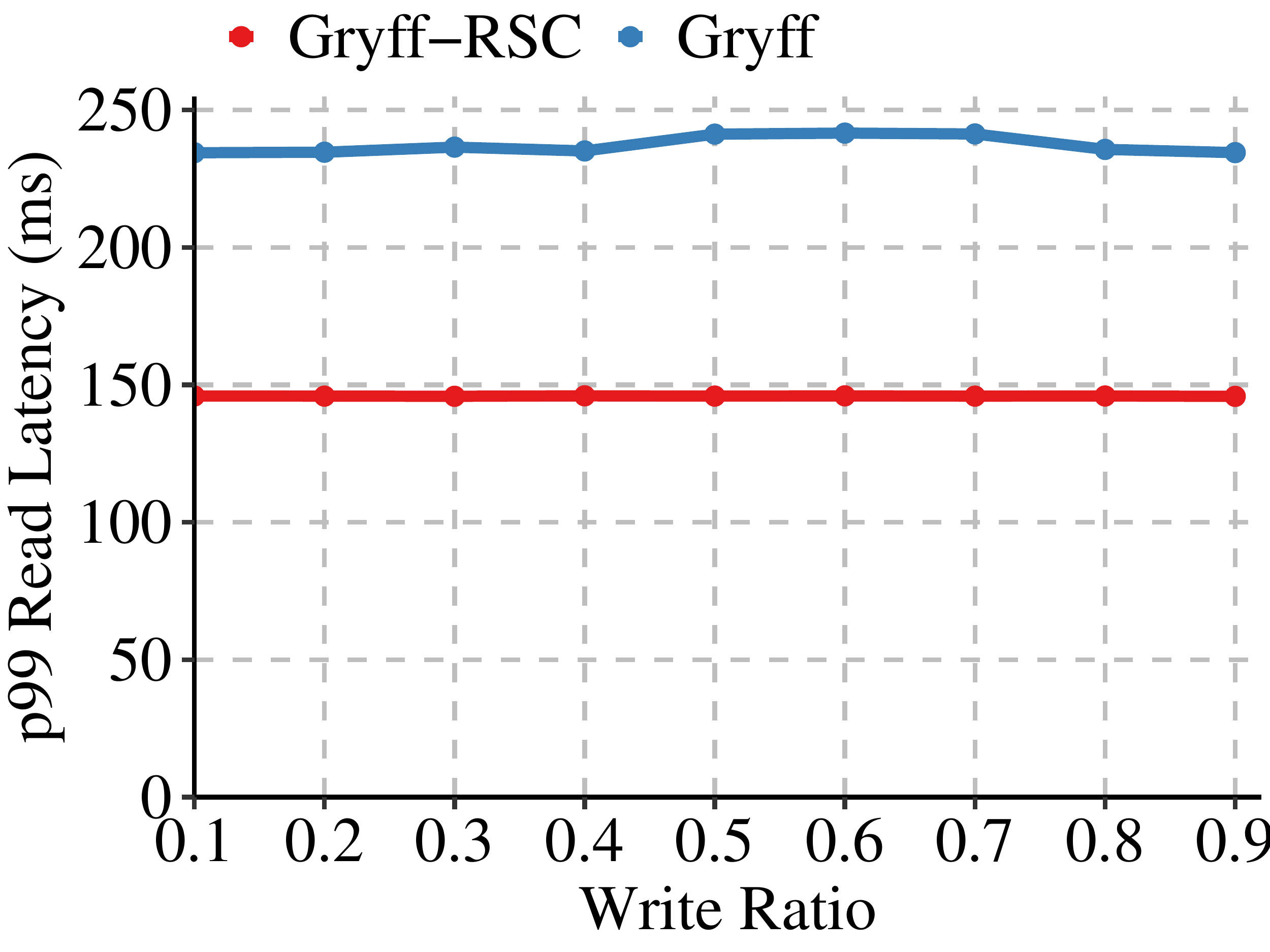}
    \caption{25\% conflicts.}
    \label{fig:gryff-p99-read-vs-write-ratio-confict-25}
    \end{subfigure}
    \caption{For moderate- and high-contention workloads, \gryffrs{} offers
      roughly a 40\% reduction in p99 read latency compared to \gryff{}. As the
      conflict ratio increases, \gryffrs{}'s benefits start at lower write
      ratios.}
    \label{fig:gryff-p99-read-vs-write-ratio-confict}
\end{figure*}

\section{\gryffrs{}}
\label{sec:gryffrs}

\gryff{} is a geo-replicated key-value store that
supports non-transactional reads, writes, and atomic
read-modify-writes (rmws)~\cite{burke2020gryff}. It provides linearizability using
a hybrid shared register and consensus protocol. Reads and writes are executed
using a shared register protocol to provide bounded tail latency whereas rmws
are executed using a consensus protocol, which is necessary for correctness.

We introduce \gryffrs{}, which provides \rslong{} and is able to reduce the
bound on tail latency from two round trips to a quorum of replicas to one
round trip. This section gives an overview of \gryffrs{}'s design and
evaluation.
The full design is described in
Appendix~\ref{sec:gryff_full_design}, and we prove it guarantees \rs{} in Appendix~\ref{sec:correctness:gryffrs}.

\subsection{\gryffrs{} Design Overview}

Read operations in \gryff{} consist of an initial read phase that contacts a
quorum of replicas to learn of the most-recent value they know of for a given
key. If the quorum returns the same values, then the read finishes.
If the quorum returns different values, however, the read continues to a second, write-back phase
that writes the most-recent value to a quorum before the read ends. This second
phase is necessitated by linearizability: once this read ends, any future reads
must observe this or a newer value.

\Rslong{} relaxes this constraint: before the write finishes, only
causally later reads must observe this or a newer value. This enables
\gryffrs{}'s reads to always complete in one phase. On a read, instead of
immediately writing the observed value back to a quorum, a \gryffrs{} client
piggybacks it onto the first phase of its next operation. Replicas write the
piggybacked value before processing the next operation. Causally later
operations by the same client are thus guaranteed to see this or a newer value.
By transitivity then, causally later operations at other clients, e.g., by the
reads-from relation, will also observe this or a newer value.

\newtext{
Piggybacking a read's second phase onto the next operation ensures a client's next operation can be serialized after all operations that causally precede it.
Similarly, a \rtbarrier{} must ensure all future operations, including those from other clients, are ordered after any operation that causally precedes it. By \rs{}, future writes and rmws are already required to respect their real-time order, but the same is not true of future reads. Thus, to execute a \rtbarrier{} in \gryffrs{}, a client writes back the key-value pair, if any, that would have been piggybacked onto its next operation. This guarantees future reads return values that are at least as recent as any operation that causally precedes the \barrier{}.
}



\subsection{\gryffrs{} Evaluation}
\label{sec:eval:gryff}

Our evaluation of \gryffrs{} aims to answer two questions: Does \gryffrs{} offer
better tail read latency on important workloads
(\S\ref{sec:eval:gryff:latency}), and what are the performance costs of
\gryffrs{}'s protocol (\S\ref{sec:eval:gryff:throughput})?

We implement \gryffrs{} in Go using the same
framework as \gryff{}~\cite{burke2020gryff}, and our code and experiment scripts are available online~\cite{gryffrsrepo}. \newtext{We keep all of \gryff{}'s optimizations enabled.}
\newtext{All experiments ran on the CloudLab~\cite{duplyakin2019cloudlab} machines described in Section~\ref{sec:eval:spanner:throughput},} and we emulate a wide-area environment. We use five replicas, one in
each emulated geographic region, \newtext{because with \gryff{}'s optimizations, reads already always finish in one round trip with three replicas.} An equal fraction of the clients are in each region. Table~\ref{tbl:wan-latency} shows the emulated round-trip times.
\begin{table}[h]
  {\small
  \begin{tabular}{c|r r r r r}
    & CA & VA & IR & OR & JP \\
    \hline
    CA & 0.2 & & & & \\
    VA & 72.0 & 0.2 & \\
    IR & 151.0 & 88.0 & 0.2 & & \\
    OR & 59.0 & 93.0 & 145.0 & 0.2 \\
    JP & 113.0 & 162.0 & 220.0 & 121.0 & 0.2 \\
  \end{tabular}
  }
  \vspace{1ex}
  \caption{Emulated round-trip latencies (in ms).}
  \vspace{-2ex}
  \label{tbl:wan-latency}
\end{table}

%
We generate load with 16 closed-loop clients. With this number, servers are
moderately loaded. Each client executes the YCSB workload~\cite{cooper2010ycsb}, which
includes just reads and writes. We vary the rate of conflicts and the read-write
ratio.

\subsection{\gryffrs{} Reduces Read Tail Latency}
\label{sec:eval:gryff:latency}

Figure~\ref{fig:gryff-p99-read-vs-write-ratio-confict} compares \gryff{} and
\gryffrs{}'s p99 read latency across a range of conflict percentages and
read-write ratios. We omit similar plots for writes because write performance
is identical in the two systems.

With few conflicts (Figure~\ref{fig:gryff-p99-read-vs-write-ratio-confict-02}), nearly
all of \gryff{}'s reads complete in one round, so \gryffrs{} cannot offer an improvement. p99 latency for both systems is \SI{145}{\milli\second}.

As Figures~\ref{fig:gryff-p99-read-vs-write-ratio-confict-10}
and~\ref{fig:gryff-p99-read-vs-write-ratio-confict-25} show, however, as the
rate of conflicts increases, more of \gryff{}'s reads must take its slow path,
incurring two wide-area round trips. This increases \gryff{}'s p99 latency by
\SI{61}{\percent} (from \SI{145}{\milli\second} to \SI{234}{\milli\second}). On
the other hand, \gryffrs{}'s reads only require one round trip, so p99 latency remains at \SI{145}{\milli\second}. At lower write ratios, the magnitude of \gryffrs{}'s improvement over \gryff{} increases with the rate of
conflicts.

Further, because reads always finish in one round, \gryffrs{} offers even
larger latency improvements farther out on the tail (not shown). For instance,
with \SI{10}{\percent} conflicts and a 0.3 write ratio, \gryffrs{} reduces p99.9
latency by \SI{49}{\percent} (from \SI{290}{\milli\second} to \SI{147}{\milli\second}).

\subsection{\gryffrs{} Imposes Negligible Overhead}
\label{sec:eval:gryff:throughput}

We also quantify the performance overhead of \gryffrs{}'s piggybacking
mechanism, but we omit the plots due to space constraints. We compare \gryff{}
and \gryffrs{}'s throughput and median latency as we increase the number of
clients. As in Section~\ref{sec:eval:spanner:throughput}, we disable wide-area emulation. With a \SI{10}{\percent}
conflict ratio, we run two workloads: 50\% reads-50\% writes and 95\% reads-5\%
writes (matching YCSB-A and YCSB-B~\cite{cooper2010ycsb}). In both cases,
\gryffrs{}'s throughput and latency are within \SI{1}{\percent} of
\gryff{}'s, suggesting the overhead from \gryffrs{}'s protocol changes are
negligible.


%% file: sections/related_work.tex
\section{Related Work}
\label{sec:related-work}

This section discusses related work on consistency models, explicitly reasoning
about invariants, equivalence results, and strictly serializable and linearizable services.

\noindentparagraph{Consistency models.} \newtext{Due to their implications for applications and services, consistency models have been studied extensively. In general, given an application, more invariants hold and fewer anomalies are possible with stronger models. But weaker models allow for better-performing services.}

\newtext{
\rss{} and \rs{} are distinct from prior works because they are the first models that are invariant-equivalent to strict serializability and linearizability, respectively. They achieve this by guaranteeing that transactions (operations) appear to execute sequentially, in an order consistent with a set of causal constraints. Prior works are not invariant-equivalent to strict serializability (linearizability) because either they do not guarantee equivalence to a sequential execution or do not capture all of the necessary causal constraints.
}

\newtext{
The discussion below generally proceeds from the strongest to the weakest consistency models. Since we have already discussed strict serializability~\cite{papadimitriou1979serializability}, process-ordered serializability~\cite{daudjee2004lazy,lu2016snow}, linearizability~\cite{herlihy1990linearizability}, and sequential consistency~\cite{lamport1979sequential} extensively, we focus here on other models.
(We also provide a technical comparison between \rss{}, \rs{}, and their proximal models in Appendix~\ref{sec:consistency-comparison}.)}

\newtext{
Like \rss{}, CockroachDB’s consistency model (CRDB)~\cite{taft2020crdb} lies between strict serializability and PO serializability. 
CRDB guarantees conflicting transactions respect their real-time order~\cite{taft2020crdb}, but it gives no such guarantee for non-conflicting transactions, which can lead to invariant violations. For instance, in a slight modification to our photo-sharing application, assume clients issue a single write to add a photo and included in this write is a logical timestamp comprising a user ID and a counter. Further, assume clients can issue a RO transaction for a user's photos. With CRDB, if Alice adds two photos and those transactions execute at different Web servers, a RO transaction that is concurrent with both may only return the second photo. If the application requires a user's photos to always appear in timestamp order, then it would be correct with a strictly serializable database but not with CRDB.
}

\newtext{
Similarly, like \rs{}, OSC(U)~\cite{levari2017osc} lies between linearizability and sequential consistency. OSC(U) guarantees writes respect their real-time order. Reads, however, may return stale values~\cite{levari2017osc}, so some pairs of reads (e.g., those invoked by different processes that also communicate via message passing) may return values inconsistent with their causal order. As a result, the non-transactional version of $\invariant{2}$ discussed in Section~\ref{sec:background:non} does not hold with OSC(U). On the other hand, OSC(U) allows services to achieve much lower read latency than what is currently achievable with \rs{}.
}

\newtext{
PO serializability and sequential consistency impose fundamental performance constraints on services~\cite{lipton1988pram}, so many weaker models, both transactional~\cite{papadimitriou1979serializability,akkoorath2016cure,
mehdi2017occult,sovran2011walter,adya1999weakcons,berenson1995ansi,
bailis2016ramp,terry1995bayou,elnikety2004gsi,pu1991esr} and
non-transactional~\cite{lipton1988pram,terry1994session,lloyd2011cops,ahamad1995causal,
cooper2008pnuts,shapiro2011crdt,balegas2015extending}, have been developed. These weaker models allow for services with much better performance than what is currently achievable with \rss{} or \rs{}. For example, a causal+
storage system can process all operations without synchronous, cross-data-center
communication~\cite{lloyd2011cops}. But application invariants break with these models because they do not guarantee equivalence to a sequential execution of either transactions or operations. Thus, they present developers with a harsh trade-off between service performance and application correctness.}

\newtext{
Based on the observation that some invariants hold with weaker consistency models, other work proposes combinations of weak and strong guarantees for different operations~\cite{ladin1992lazyrep,li2012redblue,li2018por,terry2013slas}. This allows these services to offer dramatically better performance for a subset of operations.
Maintaining application correctness while using these services, however, requires application programmers to choose the correct consistency for each operation.
}

\newtext{
Finally, three other works use causal or real-time constraints in innovative ways~\cite{baldoni1996deltacausal,yu2002tact,mahajan2011rtcausal}. First, $\Delta$-causal messaging applies real-time guarantees to a different domain where messages have limited, time-bounded relevance (e.g., video streaming)~\cite{baldoni1996deltacausal}. Second, real-time causal strengthens causal consistency by ensuring writes respect their real-time order~\cite{mahajan2011rtcausal}. But because real-time causal does not capture all necessary causal constraints, $\invariant{2}$ would not hold.
}

\newtext{
Third, TACT gives application developers fine-grained control over its consistency~\cite{yu2002tact}. For instance, an application can set a different staleness bound for each invocation to control how old (in real time) the values returned by the operation may be. (Setting zero for all operations provides strict serializability.) Compared to \rss{}, TACT's fine-grained control allows for services with better performance but requires developers to choose the correct bounds when ensuring their application's correctness.}






\noindentparagraph{Reasoning about explicit invariants.}
Several tools and techniques have been proposed for reasoning about the correctness of applications that run on services with weaker consistency~\cite{alglave2017ogre, raad2019library, najafzadeh2016cise, li2014automating, gotsman2016cause, brutschy2018static}.
For example, SIEVE~\cite{li2014automating} uses static and dynamic analysis of Java application code to determine the necessary consistency level for operations to maintain a set of explicitly written invariants.
Brutschy et al.\@~\cite{brutschy2018static} describe a static analysis tool for identifying non-serializable application behaviors that are possible when running on a causally consistent service.
Gotsman et al.\@~\cite{gotsman2016cause} introduce a proof rule (and accompanying static analysis tool~\cite{najafzadeh2016cise}) that enables
modular reasoning about the consistency level required to maintain explicit
invariants.

These tools and techniques help application programmers ensure explicit
invariants hold when using services with weaker consistency.
In contrast, \rss{} (\rs{}) services ensure the same application invariants
as strictly serializable (linearizable) services. This makes it easier to build
correct applications because programmers can write code without
stating invariants, running static analyses, or writing proofs.

\noindentparagraph{Equivalence results.} Other works have leveraged the notion
of equivalence, or indistinguishability, to prove interesting theoretical
results~\cite{goldman1993unifiedModel,lundelius1984clocksync,fischer1985flp,attiya1993seqlin}.
In fact, our results here are inspired by them. But while we leverage some of
their ideas and techniques, these works apply equivalence to different ends,
e.g., to prove bounds on clock
synchronization~\cite{lundelius1984clocksync} or show there are
fundamental differences in the performance permitted by different consistency
models~\cite{attiya1993seqlin}.

\noindentparagraph{Strictly serializable services.} \spanner{} is a globally
distributed, strictly serializable database~\cite{corbett2013spanner}. Since its
publication, other such services have been
developed~\cite{thomson2014calvin,mahmoud2013replicatedCommit,zhang2018tapir,zhang2015tapir,mu2014rococo, mu2016janus,kraska2013mdcc,ren2019slog,taft2020crdb,yan2018carousel}.
These services has largely focused on improving the
throughput~\cite{thomson2014calvin,ren2019slog} and
latency~\cite{mahmoud2013replicatedCommit,zhang2018tapir,zhang2015tapir,kraska2013mdcc,ren2019slog,yan2018carousel}
of read-write transactions, which can incur multiple inter-data-center round trips
in \spanner{}.

Because they only require one round trip between a client and the
participant shards, \spanner{}'s RO transactions continue to
perform as well or better than those of other services. These improvements
are thus orthogonal to those offered by \spannerrss{}, and in fact, weakening the
consistency of these other services to \rss{} may allow for designs that combine
their improved RW transaction performance with RO transactions
that are competitive with \spannerrss{}'s.


\noindentparagraph{Linearizable services.} \gryff{} is a recent geo-replicated
storage system that combines shared registers and
consensus~\cite{burke2020gryff}. Many other protocols have been developed to
provide replicated and linearizable storage~\cite{lamport1998paxos,
moraru2013epaxos,oki1988vr,ongaro2014raft,mao2008mencius,ailijiang2020wpaxos,
zhao2018sdpaxos,gavrielatos2020kite,katsarakis2020hermes,enes2020atlas,
enes2021efficient,ngo2020copilot}.
Weakening the consistency of these other services to \rs{} is likely to enable new variants of their designs that
improve their performance.


%% file: sections/conclusion.tex
\section{Conclusion}
\label{sec:conclusion}

Existing consistency models offer a harsh trade-off to application programmers;
they often must choose between application correctness and performance. This
paper presents two new consistency models, \rsslong{} and \rslong{}, to ease
this trade-off. \rss{} and \rs{} maintain application invariants while
permitting new designs that achieve better performance than their strictly
serializable or linearizable counterparts. To this end, we design variants of two existing systems, \spannerrss{} and \gryffrs{}, that guarantee \rss{} and \rs{}, respectively. Our
evaluation demonstrates significant (\SI{40}{\percent} to \SI{50}{\percent})
reductions in tail latency for read-only transactions and reads.


%% file: sections/consistency_comparison.tex
\section{Comparing \rss{} and \rs{} To Their Proximal Consistency Models}
\label{sec:consistency-comparison}

As discussed extensively in the main body of the paper, \rsslong{} lies between strict serializability~\cite{papadimitriou1979serializability} and process-ordered serializability~\cite{daudjee2004lazy,lu2016snow}. While other consistency and isolation definitions lie between or near these points in the consistency spectrum, \rss{} is the first consistency model that is invariant-equivalent to strict serializability. Existing models either fail to reflect a sequential execution of transactions~\cite{daudjee2006lazysi} or to respect the broad set of causal constraints necessary to maintain application invariants~\cite{taft2020crdb}. As a result, invariants like $\invariant{1}$ or $\invariant{2}$, respectively, will not hold.

Similarly, \rslong{} lies betwen linearizability~\cite{herlihy1990linearizability} and sequential consistency~\cite{lamport1979sequential}. While many existing consistency models lie near \rs{}, \rs{} is the first to be invariant-equivalent to linearizability. Existing non-transactional models again fail to provide one or both of two guarantees: They either do not reflect a sequential execution of operations~\cite{mahajan2011rtcausal,shao2003mwregularity,shao2011mwregularity} or do not respect the necessary causal constraints~\cite{lamport1979sequential,levari2017osc,viotti2016conssurvey}.

In the remainder of this section, we compare \rss{} and \rs{} to their (respective) proximal consistency models. (For brevity, we focus only on proximal models and refer the reader to other works for more comprehensive surveys of existing transactional~\cite{adya1999weakcons,crooks2017seeing} and non-transactional~\cite{viotti2016conssurvey} consistency definitions.) These comparisons primarily serve to verify that our definitions of \rss{} and \rs{} are indeed novel. They also, however, help illuminate how existing models compare to \rss{} and \rs{} with regards to application invariants and user-visible anomalies. While we also compare these existing models to each other with informal arguments, we omit formal proofs of these comparisons as they are not the focus of this work. 

\subsection{\RSSlong{}}
\label{sec:consistency-comparison:rss}

\begin{figure}[!t]
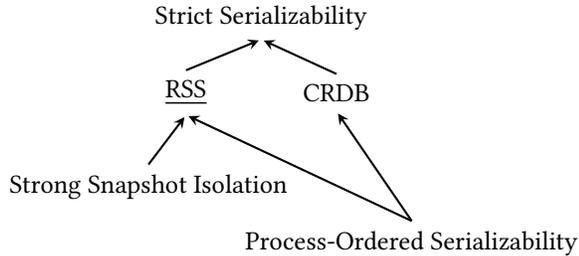

\centering
\include{figures/rss-comparison}
\caption{\rss{} compared to its proximal consistency models: strict serializability~\cite{papadimitriou1979serializability}, CRDB~\cite{taft2020crdb}, process-ordered serializability~\cite{daudjee2004lazy,lu2016snow}, and strong snapshot isolation~\cite{daudjee2006lazysi}.}
\label{fig:rss-comparison}
\end{figure}

Figure~\ref{fig:rss-comparison} compares \rss{} to its proximal consistency models.

\noindentparagraph{CRDB.} Like \rss{}, CRDB's consistency model~\cite{taft2020crdb} is stronger than process-ordered serializability; it ensures transactions appear to execute in some total order that is consistent with each client's process order. Further, it provides some real-time guarantees for writes, which prevent stale reads in many cases.

But CRDB's consistency model is incomparable to \rss{}. To show this, we exhibit two schedules: Figure~\ref{fig:rss-crdb-1} is allowed by CRDB but not by \rss{}, and the reverse is true of Figure~\ref{fig:rss-crdb-2}.

\begin{figure}[!h]
\centering
\includegraphics[width=0.7\linewidth,page=1]{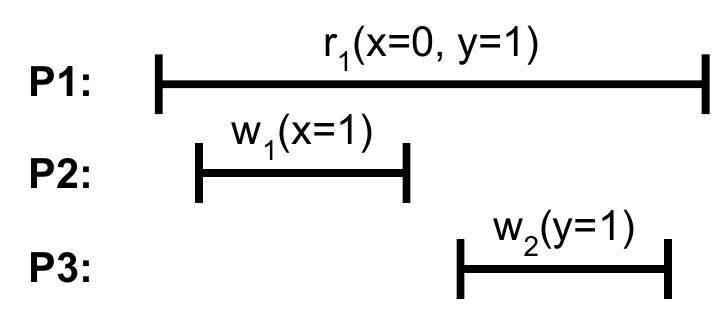}
\caption{Allowed by CRDB but disallowed by \rss{}.}
\label{fig:rss-crdb-1}
\end{figure}

\begin{figure}[!h]
\centering
\includegraphics[width=0.7\linewidth,page=2]{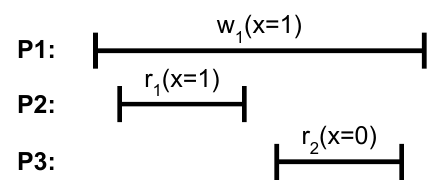}
\caption{Allowed by \rss{} but disallowed by CRDB.}
\label{fig:rss-crdb-2}
\end{figure}

Figure~\ref{fig:rss-crdb-1} shows a generic version of the example described in Section~\ref{sec:related-work}. As described, in a slightly modified version of our photo-sharing example application, the invariant that a user's photos always appear in order breaks with CRDB.

\noindentparagraph{Strong Snapshot Isolation.} Strong snapshot isolation~\cite{daudjee2006lazysi} is weaker than \rss{}. Briefly, strong snapshot isolation strengthens snapshot isolation~\cite{adya1999weakcons} by requiring that if a transaction $T_2$ follows $T_1$ in real time, then $T_2$'s start timestamp must be greater than $T_1$'s commit timestamp. This implies that $T_2$'s reads will observe $T_1$'s writes and $T_2$'s writes will be ordered after $T_1$'s. This guarantee is similar to the real-time guarantee provided by \rss{}.

Strong snapshot isolation, however, does not guarantee the database's state reflects a sequential execution of transactions. For instance, like with snapshot isolation, write skew is possible. Figure~\ref{fig:ssi-write-skew} shows an example. As a result, one can construct an invariant like $\invariant{1}$ that does not hold with strong snapshot isolation.

\begin{figure}[!h]
\centering
\includegraphics[width=0.7\linewidth,page=3]{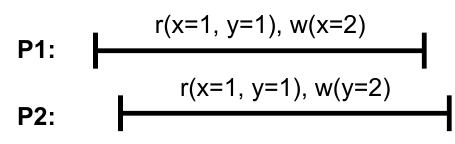}
\caption{Allowed by strong snapshot isolation.}
\label{fig:ssi-write-skew}
\end{figure}

As shown in Figure~\ref{fig:rss-comparison}, strong snapshot isolation is also incomparable to process-ordered serializability and CRDB. Write skew is impossible with process-ordered serializability and CRDB. But process-ordered serializability allows stale reads, which are not allowed by strong snapshot isolation. Further, the execution shown in Figure~\ref{fig:rss-crdb-1} is allowed by CRDB but violates strong snapshot isolation because $r_1$'s return values imply that $w_1$ is serialized after $w_2$ even though it precedes $w_2$ in real time.

\subsection{\RSlong{}}
\label{sec:consistency-comparison:rs}

\begin{figure}[!t]
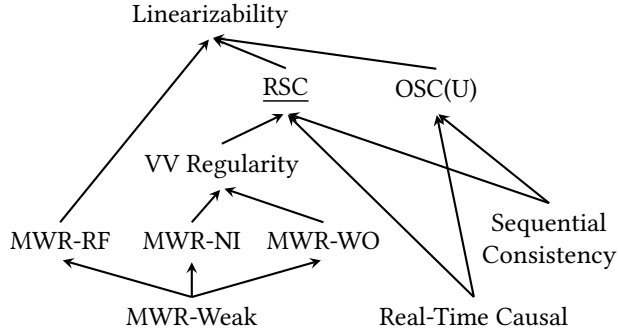

\centering
\include{figures/rs-comparison}
\caption{\rs{} compared to its proximal consistency models: linearizability~\cite{herlihy1990linearizability}, OSC(U)~\cite{levari2017osc}, sequential consistency~\cite{lamport1979sequential}, real-time causal consistency~\cite{mahajan2011rtcausal}, VV Regularity~\cite{viotti2016conssurvey}, and the four regularity definitions proposed by \citeauthor{shao2003mwregularity}~\cite{shao2003mwregularity,shao2011mwregularity}.}
\label{fig:rs-comparison}
\end{figure}

Figure~\ref{fig:rs-comparison} compares \rs{} to its proximal consistency models.

\noindentparagraph{OSC(U).} Like \rs{}, ordered sequential consistency~\cite{levari2017osc}, denoted OSC(U), is stronger than sequential consistency; it ensures operations appear to execute in some total order that is consistent with each client's process order. It also imposes some real-time constraints on writes.

But OSC(U) is incomparable to \rs{} because OSC(U) and \rs{}'s real-time guarantees differ. Unlike \rs{}, which requires that all operations \textit{following} a write are ordered \textit{after} it, OSC(U) requires that all operations \textit{preceding} a write are ordered \textit{before} it.

This difference has practical consequences. As shown in Figure~\ref{fig:rs-osc-1}, OSC(U) allows stale reads, which are not allowed by \rs{}. Thus, the non-transactional equivalent of invariant $\invariant{2}$ does not hold with OSC(U). 

\begin{figure}[!h]
\centering
\includegraphics[width=0.7\linewidth,page=1]{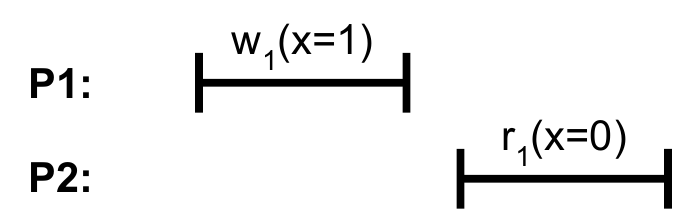}
\caption{Allowed by OSC(U) but disallowed by \rs{}.}
\label{fig:rs-osc-1}
\end{figure}

On the other hand, Figure~\ref{fig:rs-osc-2} shows an execution that \rs{} allows but OSC(U) does not. OSC(U) forbids it because $r_1$ precedes $w_1$ in real time, but P4's reads of $x=1$ and then $x=2$ imply $w_1$ precedes $w_2$ and thus $r_1$ in the total order.

\begin{figure}[!h]
\centering
\includegraphics[width=0.9\linewidth,page=2]{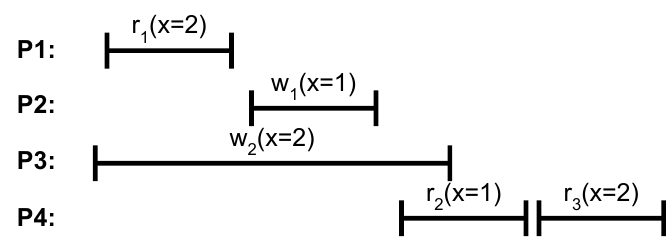}
\caption{Allowed by \rs{} but disallowed by OSC(U).}
\label{fig:rs-osc-2}
\end{figure}


\noindentparagraph{Real-time Causal.} Real-time causal consistency~\cite{mahajan2011rtcausal} strengthens causal consistency~\cite{ahamad1995causal} by also requiring that the order of causally unrelated writes respects their real-time order. Like other causal consistency models, real-time causal does not guarantee that operations appear to execute in some total order. Further, \rs{} gives the same real-time ordering guarantee for pairs of writes and captures a superset of the causal constraints required by real-time causal. \rs{} is thus stronger than real-time causal.

By similar reasoning, OSC(U) is stronger than real-time causal. They capture the same causal constraints, and OSC(U) also ensures pairs or writes are serialized in their real-time order. Further, as \citeauthor{viotti2016conssurvey} show, real-time causal is incomparable to sequential consistency and their definition of regularity (discussed below)~\cite{viotti2016conssurvey}.

\noindentparagraph{Viotti-Vukoli\'{c} Regularity.} \citeauthor{viotti2016conssurvey}~\cite{viotti2016conssurvey} give a definition of regularity that applies for multiple writers. (\citeauthor{lamport1986interprocess}'s original definition only applied to a single-writer register~\cite{lamport1986interprocess}.) Like \rs{}, their definition of regularity ensures operations appear to execute in some total order and operations that follow a write in real time must also follow it in the order.

But their consistency is weaker than \rs{} because the total order is not required to respect the causal constraints imposed by \rs{}. This can lead to violations of invariants. For instance, if in the execution shown in Figure~\ref{fig:rss-crdb-2}, there was a causal dependency between $r_1$ and $r_2$ (e.g., through message passing), $r_2$ could still return $x=0$ with VV regularity, which could lead to a violation of an invariant. The same is not true of \rs{}. 

As mentioned above, \citeauthor{viotti2016conssurvey} show their regularity defintion is incomparable to real-time causal and sequential consistency~\cite{viotti2016conssurvey}. Viotti-Vukoli\'{c} (VV) regularity is also incomparable to OSC(U): Figure~\ref{fig:rs-osc-1} is allowed by OSC(U) but disallowed by VV regularity, and Figure~\ref{fig:rs-osc-2} is allowed by VV regularity but disallowed by OSC(U).

\noindentparagraph{\citeauthor{shao2003mwregularity} Regularity} \citeauthor{shao2003mwregularity}~\cite{shao2003mwregularity,shao2011mwregularity} propose six new definitions of multi-writer regularity that form a lattice between their weakest consistency definition, MWR-Weak, and linearizability. Above MWR-Weak, MWR-Write-Order (MWR-WO), MWR-Reads-From (MWR-RF), and MWR-No-Inversion (MWR-NI) together form the next level of the lattice. The three intersections of pairs of MWR-WO, MWR-RF, and MWR-NI form the level above that. Finally, linearizability is at the top.

Briefly, MWR-Weak requires that each read returns the value of the most recently completed write or a value of some ongoing, concurrent write. Another way to state this guarantee is that an execution satisfies MWR-Weak if for each read $r$, there exists a serialization of $r$ and all writes in the execution that respects the real-time order of the read and writes. Note, however, that each read has its own serialization, so they may reflect different serializations of concurrent writes. Thus, MWR-Weak does not guarantee that operations appear to execute in some total order.

MWR-WO, MWR-RF, and MWR-NI each strengthen MWR-Weak by imposing additional, incomparable constraints on the serializations for each read~\cite{shao2003mwregularity,shao2011mwregularity}. Informally, MWR-WO requires that the serializations for each pair of reads agree on the order of the writes that are relevant to both, where a write is relevant to a read if it is concurrent or precedes the read in real time. MWR-RF requires the serialization of each read to respect a global reads-from relation in addition to the real-time order of read and writes. For instance, MWR-RF forbids the execution in Figure~\ref{fig:rs-osc-2} because the facts that $r_1$ reads from $w_2$ and $r_1$ precedes $w_1$ in real time imply $w_1$ must precede $w_2$ in $r_3$'s serialization. Thus, $r_3$ must return $x=2$. Finally, MWR-NI strengthens MWR-Weak by requiring that reads executed by the same process agree on the order of writes, although different processes may disagree on the order of concurrent writes.

Both MWR-WO and MWR-NI are weaker than VV regularity and thus \rs{}. VV regularity gives the same real-time guarantee for operations following writes. Further, the total order guaranteed by VV regularity implies that reads agree on the serialization of their mutually relevant writes (satisfying MWR-WO) and that reads from the same process agree on the order of writes (satisfying MWR-NI). 

But neither MWR-WO nor MWR-NI guarantee a total order. Figure~\ref{fig:rs-mwr-1} shows an example execution allowed by both. MWR-WO allows it because although both $w_1$ and $w_2$ are relevant to all four reads, MWR-WO does not require that reads respect the process orders. Thus, the serializations for $r_1$ and $r_2$ can be $w_1,w_2,r_1$ and $r_2,w_1,w_2$, respectively. MWR-NI allows it by similar reasoning. Thus, MWR-WO and MWR-NI are weaker than VV regularity and \rs{}.

\begin{figure}[!h]
\centering
\includegraphics[width=0.7\linewidth,page=3]{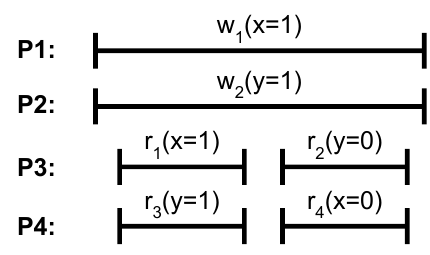}
\caption{Allowed by MWR-WO and MWR-NI but disallowed by \rs{}.}
\label{fig:rs-mwr-1}
\end{figure}

On the other hand, MWR-RF is incomparable to \rs{}. As mentioned above, MWR-RF forbids the execution in Figure~\ref{fig:rs-osc-2}, and Figure~\ref{fig:rs-mwr-2} shows an execution allowed by MWR-RF but not \rs{}. In Figure~\ref{fig:rs-mwr-2}, the reads-from relation does not introduce any additional constraints on the values returned by the reads. 

\begin{figure}[!h]
\centering
\includegraphics[width=0.7\linewidth,page=4]{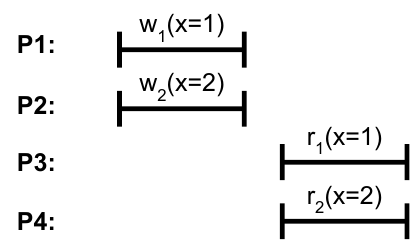}
\caption{Allowed by MWR-RF and MWR-NI but disallowed by \rs{}.}
\label{fig:rs-mwr-2}
\end{figure}

Finally, we consider the remaining three consistency models found by intersecting pairs of MWR-WO, MWR-RF, and MWR-NI. Since both MWR-WO and MWR-NI allow Figure~\ref{fig:rs-mwr-1}, their intersection does, too. Further, since VV regularity (and thus \rs{}) is stronger than both, it is also stronger than their intersection. Similarly, the execution shown in Figure~\ref{fig:rs-mwr-2} is also allowed by MWR-NI since each process has at most one read. Thus, the intersection of MWR-RF and MWR-NI is also incomparable to \rs{}. Finally, \citeauthor{shao2011mwregularity} show that the intersection of MWR-WO and MWR-RF is equivalent to linearizability~\cite{shao2011mwregularity}.

\citeauthor{shao2003mwregularity} show their regularity definitions are
incomparable to sequential consistency because they do not guarantee that
operations appear to execute in some total order consistent with each
client's process order~\cite{shao2011mwregularity}. By similar reasoning, their
definitions are incomparable to OSC(U). Finally, the authors also show
their consistency definitions are incomparable to causal consistency. By similar
reasoning, they are incomparable to real-time causal.


%% file: figures/rss-comparison.tex
\begin{tikzpicture}[
            > = stealth, 
            shorten > = 1pt, 
            thick 
        ]

        \tikzstyle{every state}=[
            draw = none,
            rectangle,
            thick,
            fill = none,
            minimum size = 0
        ]

        \node[state] (ss) at (0,0) {Strict Serializability};
        \node[state] (crdb) at (1,-1) {CRDB};
        \node[state] (rss) at (-1,-1) {\underline{\rss{}}};
        \node[state] (ssi) at (-1.5,-2.25) {Strong Snapshot Isolation};
        \node[state] (pos) at (2,-3) {Process-Ordered Serializability};
        
        \draw[->] (crdb.north) -- (ss.south);
        \draw[->] (rss.north) -- (ss.south);
        \draw[->] (ssi.north) -- (rss.south);
        \draw[->] (pos.north) -- (crdb.south);
        \draw[->] (pos.north) -- (rss.south);
    \end{tikzpicture}

%% file: figures/rs-comparison.tex
\begin{tikzpicture}[
            > = stealth, 
            shorten > = 1pt, 
            thick 
        ]

        \tikzstyle{every state}=[
            draw = none,
            rectangle,
            thick,
            fill = none,
            minimum size = 0
        ]

        \node[state] (l) at (-3,0) {Linearizability};
        \node[state] (osc) at (0,-1) {OSC(U)};
        \node[state] (rs) at (-2,-1) {\underline{\rs{}}};
        \node[state,text width=3cm,align=center] (sc) at (1.5,-3) {Sequential Consistency};
        \node[state] (vvreg) at (-2.85,-2) {VV Regularity};
        \node[state] (rtc) at (0.5,-4) {Real-Time Causal};
        
        \node[state] (mwrwo) at (-1.5,-3) {MWR-WO};
        \node[state] (mwrni) at (-3.25,-3) {MWR-NI};
        \node[state] (mwrrf) at (-5,-3) {MWR-RF};
        \node[state] (mwrweak) at (-3.25,-4) {MWR-Weak};

        \draw[->] (osc.north) -- (l.south);
        \draw[->] (rs.north) -- (l.south);
        \draw[->] (vvreg.north) -- (rs.south);
        
        \draw[->] (sc.north) -- (osc.south);
        \draw[->] (sc.north) -- (rs.south);
        
        \draw[->] (rtc.north) -- (osc.south);
        \draw[->] (rtc.north) -- (rs.south);
        
        \draw[->] (mwrweak.north) -- (mwrwo.south);
        \draw[->] (mwrweak.north) -- (mwrni.south);
        \draw[->] (mwrweak.north) -- (mwrrf.south);
        
        \draw[->] (mwrwo.north) -- (vvreg.south);
        \draw[->] (mwrni.north) -- (vvreg.south);
        \draw[->] (mwrrf.north) -- (l.south);
\end{tikzpicture}

%% file: sections/gryff_full_design.tex
\section{Full \gryffrs{} Design}
\label{sec:gryff_full_design}

Gryff is a geo-replicated key-value store that
supports non-transactional operations, namely, reads, writes, and atomic
read-modify-writes (rmws) of single objects. It provides linearizability using
a hybrid shared register and consensus protocol. Reads and writes are executed
using a shared register protocol to provide bounded tail latency whereas rmws
are executed using a consensus protocol, which is necessary for correctness. We
show Gryff can be modified to provide \rslong{} to improve tail read latency.

\noindentparagraph{Gryff background.} Each replica in Gryff maintains a mapping from
keys to values and auxiliary state for its underlying consensus protocol. In
addition, the value of each of key is associated with a consensus-after-register
timestamp (carstamp) that denotes the position in the linearizable total order
of operations of the last write or rmw to the key.

The read, write, and rmw protocols each follow a common structure that ensures
carstamps are observed and updated in an order consistent with the real-time
order of operations. For a coordinator executing an operation $o$
that accesses key $k$, the structure is composed of a Read Phase and
a subsequent Write Phase:
\begin{enumerate}[leftmargin=*]
  \item \textbf{Read Phase.} The coordinator gathers the carstamp $\textit{cs}_r$
    associated with $k$ from each server $s$ in a quorum $Q$. It then chooses
    the \emph{read phase carstamp} $\textit{cs}_r$ to be the maximum over all $s \in Q$
    of $\textit{cs}_s$.
  \item \textbf{Write Phase.} The coordinator determines the operation's
    carstamp $\textit{cs}$ and write value $v$ based on $o$'s type. For
    reads, $\textit{cs} = \textit{cs}_r$ and $v$ is the value that is associated
    with $\textit{cs}$. For writes, $\textit{cs}$ is chosen to be
    larger than $\textit{cs}_r$ and $v$ is the new value being written. For rmws,
    $\textit{cs}$ is also chosen to be larger than $\textit{cs}_r$ and $v$ is
    some user-defined function of the value that is associated with $\textit{cs}$.
    The coordinator then propagates $\textit{cs}$ and $v$ to each server
    in a quorum $Q'$.
\end{enumerate}
The key to correctness is that quorums are required to have non-empty
intersection. This implies that once an operation $o_1$ completes
its Write Phase, any subsequent operation $o_2$ will observe the
carstamp for $o_1$ on at least one replica during its Read Phase.
Further, since the carstamp of $o_2$
is chosen to be at least as large as the largest observed carstamp during the
Read Phase, the carstamp for $o_2$ will be at least as large as the
carstamp for $o_1$. 

This manner of ordering operations enables a performance optimization for reads
because they never modify the value of the objects they
access. A coordinator of a read can omit the Write Phase entirely while still
maintaining linearizability if all of the carstamps it observes in the
Read Phase are the same. In this case, the carstamp of the read is
already propagated to a quorum, so the Write Phase is not needed to ensure
subsequent operations observe the read's carstamp.

\begin{algorithm}[!tb]
  \caption{\gryffrs{} Client}
  \label{alg:gryff-rsc-client}
  \input{algorithms/abd-client}
\end{algorithm}

\begin{algorithm}[!tb]
  \caption{\gryffrs{} Server Read/Write}
  \label{alg:gryff-rsc-server}
  \input{algorithms/abd-server}
\end{algorithm}

\begin{algorithm}[!tb]
  \caption{\gryffrs{} Server RMW}
  \label{alg:gryff-rsc-server-rmw}
  \input{algorithms/abd-server-rmw}
\end{algorithm}

\noindentparagraph{\gryffrs{}.} Relaxing the consistency model from linearizability to
\rslong{} allows us to further optimize Gryff's read protocol. The Write Phase
of the read protocol is only necessary to ensure that subsequent reads observe
the same or newer values as previously completed reads, which is required for
linearizability. \Rslong{}, however, only requires this when the reads are causally related.

To take advantage of this weaker requirement, \gryffrs{} always omits the Write
Phase for reads and instead tracks a small amount of causal metadata to ensure
that causally related reads are ordered properly.
Algorithms~\ref{alg:gryff-rsc-client}, \ref{alg:gryff-rsc-server}, and
\ref{alg:gryff-rsc-server-rmw} show how this metadata is tracked. It is a
single tuple $d$ maintained by each client process. The tuple
comprises the key $d.k$, carstamp $d.\textit{cs}$, and value
$d.v$ of the most recent read the client completed that has not
yet been propagated to a quorum.

The metadata is populated with the carstamp and value of a read when the read
completes at the client and it does not have enough information to know
that the observed value already exists on a quorum. When the client executes
its next operation $o$, it piggybacks $d$ in the Read Phase of
$o$. For reads and writes, the client directly performs the Read Phase,
so $d$ is directly attached to \textit{Read1} and \textit{Write1} messages,
respectively. For rmws, the client forwards $d$ to the server that coordinates
the operation, and the server attachs $d$ to \textit{PreAccept} messages.

Replicas receiving the Read Phase messages first update their key-value stores
with the information contained in $d$, overwriting their local
carstamp and value for $d.k$ if $d.\textit{cs}$ is larger than
their current carstamp. Then the servers process the Read Phase messages as
normal in \gryff{}. The client clears $d$ as soon as it receives confirmation
that it has been propagated to a quorum, either at the end of the Read Phase for
reads and writes or when it receives notification that the operation is complete
for rmws.

In Algorithm~\ref{alg:gryff-rsc-server-rmw}, we omit the coordination of a rmw
beyond the PreAccept phase because the rest of the processing of \textit{PreAccept}
messages, the processing of \textit{Accept}
and \textit{Commit} messages, the recovery procedure, and the execution procedure are unchanged from \gryff{}. We refer the
reader to the complete description of \gryff{}~\cite{burke2020gryff} for more details.




%% file: algorithms/abd-client.tex
\begin{algorithmic}[1]
    \State \Global $c \gets \text{unique client ID}$
    \State \Global $d \gets \bot$ \Comment{Dependency}
    \Function{Client::Read}{$k$}
        \State \Send \Msg{Read}{$k, d$} to all $s \in S$
        \State \WaitUntil receive \Msg{ReadReply}{$v_s, \textit{cs}_s$} from all
            $s \in Q \in \mathcal{Q}$
        \State $\textit{cs} \gets \max_{s \in Q} \textit{cs}_s$
        \State $v \gets v_s : \textit{cs}_s = \textit{cs}$
        \If{$\exists s \in Q : \textit{cs}_s \neq \textit{cs}$}
            \State $d \gets (k, v, \textit{cs})$
        \EndIf
        \State \Return $v$
    \EndFunction
    \Statex
    \Function{Client::Write}{$k,v$}
        \State \Send \Msg{Write1}{$k, d$} to all $s \in S$
        \State \WaitUntil receive \Msg{Write1Reply}{$\textit{cs}_s$} from all $s
            \in Q \in \mathcal{Q}$
        \State $d \gets \bot$
        \State $\textit{cs} \gets \max_{s \in Q} \textit{cs}_s$
        \State \Send \Msg{Write2}{$k, v, (\pi_0(\textit{cs}) + 1, c)$} to all
            $s \in S$
        \State \WaitUntil receive \Msg{Write2Reply}{} from all $s \in Q^\prime \in \mathcal{Q}$
    \EndFunction
    \Statex
    \Function{Client::RMW}{$k,f(\cdot)$}
        \State \Send \Msg{RMW}{$k, f(\cdot), d$} to one $s \in S$
        \State \WaitUntil receive \Msg{RMWReply}{} from $s$
        \State $d \gets \bot$
    \EndFunction
\end{algorithmic}

%% file: algorithms/abd-server.tex
\begin{algorithmic}[1]
    \State \Global $V \gets [\bot,\ldots,\bot]$ \Comment{Values}
    \State \Global $\textit{CS} \gets [(0, 0, 0),\ldots,(0, 0, 0)]$ \Comment{Carstamps}
    \Function{Server::ReadRecv}{$c, k, d$}
        \If{$d \neq \bot$}
            \State \Call{Apply}{$d.k, d.v, d.\textit{cs}$}
        \EndIf
        \State \Send \Msg{ReadReply}{$V[k], \textit{CS}[k]$} to $c$
    \EndFunction
    \Statex
    \Function{Server::Write1Recv}{$c, k, d$}
        \If{$d \neq \bot$}
            \State \Call{Apply}{$d.k, d.v, d.\textit{cs}$}
        \EndIf
        \State \Send \Msg{Write1Reply}{$\textit{CS}[k]$} to $c$
    \EndFunction
    \Statex
    \Function{Server::Write2Recv}{$c, k, v, \textit{cs}$}
        \State \Call{Apply}{$k, v, \textit{cs}$}
        \State \Send \Msg{Write2Reply}{} to $c$
    \EndFunction
    \Statex
    \Function{Server::Apply}{$k, v, \textit{cs}$}
    \If{$\textit{cs} > \textit{CS}[k]$}
        \State $V[k] \gets v$
        \State $\textit{CS}[k] \gets \textit{cs}$
    \EndIf
    \EndFunction
\end{algorithmic}

%% file: algorithms/abd-server-rmw.tex
\begin{algorithmic}[1]
  \State \Global $s \gets \text{unique server ID}$
    \State \Global $\textit{prev} \gets [(\bot, (0, 0, 0)),\ldots]$ \Comment{Result of previous rmw for key}
    \State \Global $i \gets 0$ \Comment{Next unused instance number}
    \State \Global $\textit{cmds} \gets [[\bot,\ldots],\ldots]$ \Comment{Instances:}
      \State \hskip1.0em $\textit{cmd}$ - command to be executed
      \State \hskip1.0em $\textit{deps}$ - commands that must execute
        before this one
      \State \hskip1.0em $\textit{seq}$ - sequence \#, breaks cycles in dependency graph
      \State \hskip1.0em $\textit{base}$ - possible base update for rmw
      \State \hskip1.0em $\textit{status}$ - status of instance
    \Statex
    \Function{Server::RMWRecv}{$c, k, f(\cdot), d$}
      \State $i \gets i + 1$ \Comment{PreAccept Phase}
      \State $\textit{cmd} \gets (k,f(\cdot))$
      \State $\textit{seq} \gets 1 + \max(\{\textit{cmds}[j][\ell].\textit{seq} | (j,\ell) \in I_\textit{cmd}\} \cup \{0\})$
      \State $\textit{deps} \gets I_\textit{cmd}$
      \State $\textit{base} \gets (V[k],	\textit{CS}[k])$
      \State $\textit{cmds}[s][i] \gets (\textit{cmd},\textit{seq}, \textit{deps},
        \textit{base}, \textbf{pre-accepted})$
      \State \Send \Msg{PreAccept}{$\textit{cmd},\textit{seq},\textit{deps},
          \textit{base},s,i,d$} to all $s' \in F \setminus \{s\}$ where $F \in \mathcal{F}$
      \State \WaitUntil receive \Msg{PreAcceptOK}{$\textit{seq}_{s'}',\textit{deps}_{s'}',
        \textit{base}_{s'}'$} from all $s' \in F \setminus \{s\}$
      \Statex $\ldots$ \Comment{Rest of RMW coordinate unchanged}
    \EndFunction
    \Statex
    \Function{Server::PreAcceptRecv}{$\textit{cmd},\textit{seq},\textit{deps}, \textit{base},s',i,d$}
        \If{$d \neq \bot$}
            \State \Call{Apply}{$d.k, d.v, d.\textit{cs}$}
        \EndIf
        \State $\textit{seq}' \gets \max(\{\textit{seq}\} \cup \{1 +
            \textit{cmds}[j][\ell].\textit{seq} | (j,\ell) \in I_\textit{cmd}\}$
        \State $\textit{deps}' \gets \textit{deps} \cup I_\textit{cmd}$
        \If{$\textit{cs} > \textit{base}. \textit{cs}$}
          \State $\textit{base}' \gets (V[\textit{cmd}.k], \textit{CS}[\textit{cmd}.k])$
        \Else
          \State $\textit{base}' \gets  \textit{base}$
        \EndIf
        \State $\textit{cmds}[s'][i] \gets (\textit{cmd}, \textit{seq}',
            \textit{deps}', \textit{base}', \textbf{pre-accepted})$
        \State \Send \Msg{PreAcceptOK}{$\textit{seq}',\textit{deps}',
            \textit{base}'$} to $s'$\;
    \EndFunction
    \Statex $\ldots$ \Comment{Other message handlers unchanged}
\end{algorithmic}

%% file: sections/proof.tex
\renewcommand{\invariant}[1]{\mathcal{I}_{#1}}

\section{Full Proof}
\label{sec:proof}

\subsection{Preliminaries}
\label{sec:proof:preliminaries}

\subsubsection{I/O Automata}
\label{sec:proof:preliminaries:ioa}

We model each component in our system as an I/O automaton
(IOA)~\cite{lynch1987ioa,lynch1996da}, a type of state machine. Each
transition of an IOA is labeled with an \textit{action}, which can be an
\textit{input}, \textit{output}, or \textit{internal} action. Input and output
actions allow the automaton to interact with other IOA and the environment. We
assume input actions are not controlled by an IOA---they may arrive at any time.
Conversely, output and internal actions are \textit{locally controlled}---an IOA defines
when they can be performed.

To specify an IOA, we must first specify its \textit{signature}. A signature $S$ is
a tuple comprising three disjoint sets of actions: input actions $\actin(S)$,
output actions $\actout(S)$, and internal actions $\actint(S)$. We also define
$\localacts(S) = \actout(S) \cup \actint(S)$ as the set of locally controlled actions, $\extacts(S) = \actin(S) \cup \actout(S)$ as the set of external actions, and $\acts(S) = \actin(S) \cup \actout(S) \cup \actint(S)$ as the set of all actions.

Formally, an I/O automaton $A$ comprises four items:
\begin{enumerate}
\item a signature $\sig(A)$,
\item a (possibly infinite) set of states $\states(A)$,
\item a set of start states $\start(A) \subseteq \states(A)$, and
\item a transition relation $\trans(A) \subseteq \states(A) \times
  \acts(\sig(A)) \times \states(A)$.
\end{enumerate}
Since inputs may arrive at any time, we assume that for every state $s$ and
input action $\pi$, there is some $(s, \pi, s^\prime) \in \trans(A)$.

An \textit{execution} of an I/O automaton $A$ is a finite or infinite
sequence of alternating states and actions $s_0,\pi_1,s_1,\ldots$ such that for
each $i \geq 0$, $(s_i, \pi_{i+1}, s_{i+1}) \in \trans(A)$ and $s_0 \in \start(A)$. Finite executions always end with a state.

Given an execution $\alpha$, we can also define
its \textit{schedule} $\sched(\alpha)$, which is the subsequence of just
the actions in $\alpha$. Similarly, an execution's \textit{trace}
$\trace(\alpha)$ is the subsequence of just the external actions $\pi \in
\extacts(A)$.

\subsubsection{Composition and Projection}
\label{sec:proof:preliminaries:composition}

To compose
two IOA, they must be \textit{compatible}. Formally, a finite set of
signatures $\{S_i\}_{i \in I}$ is compatible if for all $i,j \in I$ such that $i \neq j$:
\begin{enumerate}
\item $\actint(S_i) \cap \acts(S_j) = \emptyset$, and
\item $\actout(S_i) \cap \actout(S_j) = \emptyset$.
\end{enumerate}
A finite set of automata are compatible if their signatures are compatible.

Given a set of compatible signatures, we define their \textit{composition} $S =
\prod_{i \in I} S_i$ as the signature with $\actin(S) = \bigcup_{i \in I}
\actin(S_i) - \bigcup_{i \in I} \actout(S_i)$, $\actout(S) = \bigcup_{i \in I}
\actout(S_i)$, and $\actint(S) = \bigcup_{i \in I} \actint(S_i)$.

The composition of a set of compatible IOA yields the automaton $A =
\prod_{i \in I} A_i$ defined as follows:
\begin{enumerate}
\item $\sig(A) = \prod_{i \in I} \sig(A_i)$,
\item $\states(A) = \prod_{i \in I} \states(A_i)$,
\item $\start(A) = \prod_{i \in I} \start(A_i)$, and
\item $\trans(A)$ contains all $(s,\pi,s^\prime)$ such that for all $i \in I$, if $\pi \in \acts(\sig(A_i))$, then $(s_i,\pi,s_i^\prime) \in \trans(A_i)$ and otherwise, $s_i = s_i^\prime$.
\end{enumerate}
The states of the composite automaton $A$ are vectors of the states of the composed
automata. When an action occurs in
$A$, all of the component automata with that action each take a step simultaneously, as defined
by their individual transition relations. The resulting state differs in each of
the components corresponding to these automata, and the other components are
unchanged. We denote the composition of a small number of
automata using an infix operator, e.g., $A \times B$.

Given the execution $\alpha$ of a composed automaton $A = \prod_{i \in I} A_i$,
we can \textit{project} the execution onto one of the component automata $A_i$.
The execution $\alpha|A_i$ is found by removing all actions from
$\alpha$ that are not actions of $A_i$. The states of $\alpha|A_i$ are
given by the $i$th component of the corresponding state in $\alpha$. The
projection of a trace is defined similarly.

Further, we can write the projection of a state $s$ of $A$ on $A_i$ as $s|A_i$.  Finally, we can also project $\trace(\alpha)$ onto a set of actions $\Pi$ where $\trace(\alpha)|\Pi$ yields the subsequence of $\trace(\alpha)$ containing only actions in $\Pi$.

\subsubsection{Invariants}
\label{sec:proof:preliminaries:invariants}

Application programmers reason about their applications by reasoning about the
invariants that hold during all executions of their application. To formalize
this notion in the IOA model, we say a state is \textit{reachable} in automaton
$A$ if it is the final state of some finite execution of $A$. An
\textit{invariant} $\invariant{A}$ is a predicate on the states of $A$ that is
true for all reachable states of $A$~\cite{lynch1996da}. Similarly, if $A_i$ is
a component of some automaton $A$, then $\invariant{A_i}$ is an invariant of
$A$ if $\invariant{A_i}$ is true of $s|A_i$ for all reachable states $s$ of
$A$.

\subsubsection{Channels}
\label{sec:proof:preliminaries:channels}

\begin{table}[!t]
  \centering
  \begin{tabular}{l l l}
    \multicolumn{3}{l}{\textbf{Signature:}} \\ \hline
    Inputs: & & Outputs: \\
    $\sendto(m)_{ij}, m \in M$ & & $\sent_{ij}$ \\
    $\recvfrom_{ij}$ & & $\receive(m)_{ij}, m \in M$ \\ \hline
    \\
    \multicolumn{3}{l}{\textbf{States:}} \\ \hline
    \multicolumn{3}{l}{$Q$, a FIFO queue of elements of $M$, initially empty} \\
    \multicolumn{3}{l}{$e$, a Boolean, initially \textit{false}} \\
    \multicolumn{3}{l}{$r$, a Boolean, initially \textit{false}} \\ \hline
    \\
    \multicolumn{3}{l}{\textbf{Actions:}} \\ \hline
    \multicolumn{2}{l|}{$\sendto_{ij}(m)$:} & $\sent_{ij}$: \\
    \multicolumn{2}{l|}{Precondition: \textit{true}} & Precondition: $e$ \\
    \multicolumn{2}{l|}{Effect: $\textsc{Push}(Q,m); e \gets \textit{true}$} & Effect: $e \gets \textit{false}$ \\
    \hline \hline
    $\recvfrom_{ij}$: & \multicolumn{2}{|l}{$\receive_{ij}(m)$:} \\
    Precondition: \textit{true} & \multicolumn{2}{|l}{Precondition: $r \land m = \text{Head}(Q)$} \\
    Effect: $r \gets \textit{true}$ & \multicolumn{2}{|l}{Effect: $\textsc{Pop}(Q); r \gets \textit{false}$} \\ \hline
  \end{tabular}
  \captionof{figure}{Buffered Channel I/O Automaton}
  \label{fig:buffered-channel-ioa}
\end{table}

Each pair of processes in our system communicates via a pair of FIFO
channels that are asynchronous, reliable, and buffered. Each
channel's signature, states, and actions are specified in
Figure~\ref{fig:buffered-channel-ioa}.

We denote the channel that process $i$ uses to send messages to process $j$ as
$C_{ij}$. $C_{ij}$ has two sets of input actions, $\sendto_{ij}(m)$ and
$\recvfrom_{ij}$, and two sets of output actions, $\sent_{ij}$ and
$\receive_{ij}(m)$, for all $m$ in some space of messages $M$. Process $i$ has
corresponding output actions, $\sendto_{ij}(m)$ and $\recvfrom_{ji}$, and input
actions, $\sent_{ij}$ and $\receive_{ji}(m)$, for all other processes $j$ and
messages $m$. To send a message to process $j$, process $i$ takes a
$\sendto_{ij}$ step, and $C_{ij}$ subsequently takes a $\sent_{ij}$ step.
Similarly, to receive a message from $C_{ij}$, process $j$ takes a
$\recvfrom_{ij}$ step, and $C_{ij}$ subsequently takes a $\receive_{ij}(m)$ step.

The modeling of buffering in the channels differs from past
work~\cite{lynch1996da}. There, receive-from actions are omitted, and received
actions are modeled as output actions of channels and corresponding input
actions of processes. This implies processes cannot control when they
change their state in response to a message.

But in real applications, this is unrealistic. Although the network stack of a
machine may accept and process a packet at any time, application code controls
when it processes the contained message. For instance, the packet's contents remain in a kernel buffer until the application performs a read system call
on a network socket. This control is essential to our proof as it ensures
application processes do not receive messages while waiting for responses from
services.

We say an execution $\alpha$ of a channel $C_{ij}$ is \textit{well-formed} if
(1) $\trace(\alpha) | \{\sendto_{ij}(m)\}_{m \in M} \cup \{\sent_{ij}\}$ is a
sequence of alternating send-to and sent actions, starting with a send-to; and
similarly (2) $\trace(\alpha) | \{\recvfrom_{ij}\} \cup \{\receive_{ij}(m)\}_{m
  \in M}$ is a sequence of alternating receive-from and received actions,
starting with a receive-from. The following four lemmas show that adjacent pairs
of actions are commutative---reordering them in an execution always yields
another execution. Lemmas~\ref{lem:channels-commute1} shows this for adjacent
send-to and receive-from actions, Lemma~\ref{lem:channels-commute2} for adjacent
send-to and received actions, Lemma~\ref{lem:channels-commute3} for adjacent
sent and receive-from action, and finally, Lemma~\ref{lem:channels-commute4} for
adjacent sent and received actions.

\begin{lemma}
  \label{lem:channels-commute1}
  Let $\alpha$ be a well-formed, finite execution of $C_{ij}$, and let
  $\alpha^\prime = \trace(\alpha)$. If there exists some $\pi_s = \sendto_{ij}(m)$
  and $\pi_r = \recvfrom_{ij}$ that are adjacent in $\alpha^\prime$, then there
  exists a well-formed, finite execution $\beta$ of $C_{ij}$ with trace
  $\beta^\prime$ such that $\beta^\prime$ is identical to $\alpha^\prime$ but
  with the order of $\pi_s$ and $\pi_r$ reversed.
\end{lemma}

\begin{proof}
  Let $\alpha$ be a well-formed, finite execution of $C_{ij}$ and let
  $\alpha^\prime = \trace(\alpha)$. Assume that $\pi_s = \sendto_{ij}(m)$ and
  $\pi_r = \recvfrom_{ij}$ are adjacent in $\alpha^\prime$. We proceed by cases, so to start,
  assume $\pi_s$ is before $\pi_r$ in $\alpha^\prime$.
  
  We construct a sequence of alternating states and actions $\beta$ and show
  that $\beta$ is a well-formed, finite execution of $C_{ij}$.

  Define $k \geq 1$ such that $\pi_s$ is the $k$th action in $\alpha$. The
  sequence $\beta$ is identical to $\alpha$ in all states and actions except for
  the $k$th action, the $k$th state, the $(k+1)$th action, and the $(k+1)$th
  state. The $k$th action is $\pi_r$ and the $(k+1)$th action is $\pi_s$.
  Let $s_k$ be the $k$th state in $\beta$. The $k$th state $s_{k}$
  is defined such that $(s_{k-1},\pi_r,s_{k}) \in \trans(C_{ij})$.
  Similarly, the $(k+1)$th state $s_{k+1}$ is defined such that
  $(s_{k},\pi_s,s_{k+1}) \in \trans(C_{ij})$. By the definition of the actions
  of $C_{ij}$ in Figure~\ref{fig:buffered-channel-ioa}, these transitions must exist
  because the preconditions of $\pi_r$ and $\pi_s$ are vacuously true.

  We claim that $\beta$ is an execution of $C_{ij}$. Let $s_i$ be the $i$th
  state and $\pi_i$ be the $i$th action of $\beta$. Clearly the
  zeroth state $s_0$ in $\beta$ is in $\start(C_{ij})$ because $s_0$ is
  identical to the zeroth state of $\alpha$ and $\alpha$ is an execution of
  $C_{ij}$. Moreover, for $0 \leq i \leq k$, $(s_i,\pi_{i+1},s_{i+1}) \in
  \trans(C_{ij})$ because the first $k$ states and $k-1$ actions of $\beta$ are
  identical to the corresponding states and actions in $\alpha$ and $\alpha$ is
  an execution of $C_{ij}$. By the definition of $\pi_{k}$, $s_{k}$,
  $\pi_{k+1}$, and $s_{k+1}$, both $(s_{k-1},\pi_{k},s_{k})$ and
  $(s_{k},\pi_{k+1},s_{k+1})$ are in $\trans(C_{ij})$.

  Now consider the state $s_{k+1}$. Let $s_i^\alpha$ be the $i$th state and
  $\pi_i^\alpha$ be the $i$th action of $\alpha$. Since $\pi_s$ only modifies
  the $Q$ and $e$ variables of $C_{ij}$ and $\pi_r$ only modifies the $r$
  variable of $C_{ij}$, as defined in Figure~\ref{fig:buffered-channel-ioa}, the
  state after executing $\pi_s$ then $\pi_r$ is the same as the state after
  executing $\pi_r$ then $\pi_s$ when starting from the same state. By the
  previous fact, the fact that $s_{k+1}^\alpha$ is the state after executing
  $\pi_s$ then $\pi_r$ from $s_k^\alpha$, the fact that $s_{k+1}$ is the state
  after executing $\pi_r$ then $\pi_s$ from $s_k$, and the definition of $s_k =
  s_k^\alpha$, $s_{k+1}$ is identical to $s_{k+1}^\alpha$.

  By the previous fact and the definition of $\beta$, $\beta$ is identical to
  $\alpha$ for all states and actions after and including $s_{k+1}$. Hence, for
  all $i \geq 0$, $(s_i,\pi_i,s_{i+1}) \in \trans(C_{ij})$. This implies that
  $\beta$ is an execution of $C_{ij}$.

  Because $\beta$ is identical to $\alpha$ except for the order of $\pi_s$ and
  $\pi_r$ and the intervening states and because $\alpha$ is finite, $\beta$ is
  also finite. Furthermore, $\beta^\prime = \trace(\beta)$ is identical to
  $\alpha^\prime$ but with the order of $\pi_s$ and $\pi_r$ reversed.

  Lastly, $\alpha^\prime|\{\sendto_{ij}(m)\}_{m \in M} \cup \{\sent_{ij}\}$ is a
  sequence of alternating send-to and sent actions and similarly
  $\alpha^\prime|\{\recvfrom_{ij}\} \cup \{\receive_{ij}(m)\}_{m \in M}$ is a
  sequence of alternating receive-from and received actions because $\alpha$ is
  well-formed. Since $\beta$ is identical to $\alpha$ except for the order of a
  single pair of receive-from and send actions, $\beta$ is thus also
  well-formed.
  
  The case where $\pi_r$ precedes $\pi_s$ in $\alpha^\prime$
  can be shown using nearly identical reasoning.
\end{proof}

The remaining three proofs employ similar logic as above. For each, we
define $\alpha$, $\alpha^\prime$, $\beta$, and $\beta^\prime$ as above and to
start, assume $\pi_s$ precedes $\pi_r$ in $\alpha^\prime$.

Similarly, we define $k$ as above, so $\pi_k$ is $\pi_s$ in $\alpha$ but
$\pi_r$ in $\beta$. To conclude each proof, we then simply show that
$(s_{k-1},\pi_r,s_k) \in \trans(C_{ij})$, $(s_{k},\pi_s,s_{k+1}) \in
\trans(C_{ij})$, and $s_{k+1}$ is identical in $\alpha$ and $\beta$. The
remaining reasoning is identical to that above.

\begin{lemma}
  \label{lem:channels-commute2}
  Let $\alpha$ be a well-formed, finite execution of $C_{ij}$, and let
  $\alpha^\prime = \trace(\alpha)$. If there exists some $\pi_s = \sendto_{ij}(m)$
  and $\pi_r = \receive_{ij}(m^\prime)$ that are adjacent in $\alpha^\prime$
  with $m \neq m^\prime$, then there exists a well-formed, finite execution of
  $\beta$ of $C_{ij}$ with trace $\beta^\prime$ such that $\beta^\prime$ is
  identical to $\alpha^\prime$ but with the order of $\pi_s$ and $\pi_r$
  reversed.
\end{lemma}

\begin{proof}
  First, by the definitions of $C_{ij}$'s actions shown in
  Figure~\ref{fig:buffered-channel-ioa}, $\pi_s$ does not modify $r$. Combining
  this fact with the assumption that $m \neq m^\prime$, $\pi_r$'s precondition
  must hold in $s_{k-1}$. Thus, $(s_{k-1},\pi_r,s_k) \in \trans(C_{ij})$.
  Further, $(s_{k},\pi_s,s_{k+1}) \in \trans(C_{ij})$ because $\pi_s$'s
  precondition is vacuously true.
  
  Now consider the state $s_{k+1}$. Because $m \neq m^\prime$, $Q$ must not have
  been empty in $s_{k-1}$. Thus, the value of $Q$ resulting from the enqueue of
  $m$ by $\pi_s$ and the dequeue of $m^\prime$ by $\pi_r$ is identical to the
  value resulting from performing the two operations in the reverse order. Using
  this fact, since only $\pi_s$ sets $e$ and only $\pi_r$ sets $r$, $\pi_{k+1}$
  thus must be identical in both $\alpha$ and $\beta$.
  
  As above, the case where $\pi_r$ precedes $\pi_s$ can be shown using nearly
  identical reasoning.
\end{proof}

\begin{lemma}
  \label{lem:channels-commute3}
  Let $\alpha$ be a well-formed, finite execution of $C_{ij}$, and let
  $\alpha^\prime = \trace(\alpha)$. If there exists some $\pi_s = \sent_{ij}$
  and $\pi_r = \recvfrom_{ij}$ that are adjacent in $\alpha^\prime$, then there
  exists a well-formed, finite execution $\beta$ of $C_{ij}$ with trace
  $\beta^\prime$ such that $\beta^\prime$ is identical to $\alpha^\prime$ but
  with the order of $\pi_s$ and $\pi_r$ reversed.
\end{lemma}

\begin{proof}
  First, by the definitions of $C_{ij}$'s actions shown in
  Figure~\ref{fig:buffered-channel-ioa}, $\pi_r$'s precondition is vacuously
  true. Thus, $(s_{k-1},\pi_r,s_k) \in \trans(C_{ij})$. Further, because $\pi_r$
  does not modify $e$ and $s_{k-1}$ is identical in $\alpha$ and $\beta$,
  $\pi_s$'s precondition must hold in $s_{k}$. Thus, $(s_{k},\pi_s,s_{k+1}) \in
  \trans(C_{ij})$.

  Finally, consider the state $s_{k+1}$. Because $\pi_s$ only modifies $e$ and
  $\pi_r$ only modifies $r$, $\pi_{k+1}$ thus must be identical in both $\alpha$
  and $\beta$.

  As above, the case where $\pi_r$ precedes $\pi_s$ can be shown using nearly
  identical reasoning.
\end{proof}

\begin{lemma}
  \label{lem:channels-commute4}
  Let $\alpha$ be a well-formed, finite execution of $C_{ij}$, and let
  $\alpha^\prime = \trace(\alpha)$. If there exists some $\pi_s = \sent_{ij}$
  and $\pi_r = \receive_{ij}(m)$ that are adjacent in $\alpha^\prime$, then
  there exists a well-formed, finite execution $\beta$ of $C_{ij}$ with trace
  $\beta^\prime$ such that $\beta^\prime$ is identical to $\alpha^\prime$ but
  with the order of $\pi_s$ and $\pi_r$ reversed.
\end{lemma}

\begin{proof}
  First, by the definitions of $C_{ij}$'s actions shown in
  Figure~\ref{fig:buffered-channel-ioa}, $\pi_s$ does not modify $r$ or $Q$.
  Thus, since $s_{k-1}$ is identical in $\alpha$ and $\beta$, $\pi_r$'s
  precondition must hold in $s_{k-1}$. Thus, $(s_{k-1},\pi_r,s_k) \in
  \trans(C_{ij})$. Further, because $\pi_r$ does not modify $e$, $\pi_s$'s
  precondition must hold in $s_{k}$, so $(s_{k},\pi_s,s_{k+1}) \in
  \trans(C_{ij})$.

  Finally, consider the state $s_{k+1}$. Because $\pi_s$ only modifies $e$ and
  $\pi_r$ only modifies $r$ and $Q$, $\pi_{k+1}$ thus must be identical in both
  $\alpha$ and $\beta$.

  As above, the case where $\pi_r$ precedes $\pi_s$ can be shown using nearly
  identical reasoning.
\end{proof}

\subsubsection{Types and Services}
\label{sec:proof:preliminaries:services}

Processes in our system interact with \textit{services}, each with a specified
\textit{type}~\cite{herlihy1990linearizability,lynch1996da}. A service's type
$\type$ defines its set of possible \textit{values} $\vals(\type)$, an
\textit{initial value} $v_0 \in \vals(\type)$, and the \textit{operations}
$\ops(\type)$ that can be invoked on the service. Each operation $o \in
\ops(\type)$ is defined by a pair of sets of actions: \textit{invocations} $\invs(o)$ and
\textit{responses} $\resps(o)$. Each contains subscripts denoting a unique service name and a process index. An invocation and response \textit{match} if their subscripts are equal.
Finally, each service has a \textit{sequential specification} $\spec$, a
prefix-closed set of sequences of matching invocation-response pairs.

For example, consider a read/write register $x$ that supports a set of $n$
processes and whose values is the set of integers. The read operation would then
be defined with invocations $\{\textsc{read}_{i,x}\}$ and responses
$\{\textsc{ret}_{i,x}(j)\}$ for all $0 \leq i \leq n$ and $j \in \mathcal{N}$.
Similarly, the write operation would have invocations
$\{\textsc{write}_{i,x}(j)\}$ and responses $\{\textsc{ack}_{i,x}\}$. Finally,
its sequential specification would be the set of all sequences of reads and
writes such that reads return the value written by the most recent write or the
initial value if none exists.

We can also compose types $\{\type_x\}_{x \in X}$ and sequential
specifications $\{\spec_x\}_{x \in X}$. Formally, the values of $\type = \prod_{x
  \in X} \type_x$ are vectors of the values of the composed types. $\ops(\type)$
is the union of those of the composed types. Finally, the composite sequential
specification $\spec = \prod_{x \in X} \spec_x$ is the set of all interleavings of the the invocation-response pairs of the component
specifications $\spec_x$.

\subsubsection{System Model}
\label{sec:proof:preliminaries:model}

We model a distributed application as the composition of two finite sets of I/O
automata: processes and channels. $n$ denotes the number of processes, so there are $n^2$ channels. The processes execute the application's code by performing local computation, exchanging messages via channels, and performing invocations on and receiving responses from services.

In the results below, we are interested in reasoning about which process invariants hold while assuming various correctness conditions of the services they interact with. Thus, we do not model services as IOA. Instead, we assume the processes interact with a (possibly composite) service with an arbitrary type $\type$ and $\spec$ defined for $n$ processes.

For each operation $o \in \ops(\type)$, process $P_i$ is then assumed to have an output action for every invocation action in $\invs(o)$ with process index $i$. Similarly, $P_i$ has an input action for every response action in $\resps(o)$ with process index $i$. We refer to these input and output actions as a process's \textit{system-facing} actions $\system(P_i)$.

To model stop failures, we assume each process $P_i$ has an input
action $\textit{stop}_i$ such that after receiving it, $P_i$ ceases taking
steps. If $\textit{stop}_i$ occurs while $P_i$ is waiting for a response from a service, then we assume the service does not return a response. But the operation may still cause a service's state to change, and this change may be visible to operations by other processes.

Finally, to allow the distributed application to receive input from and return values to its environment (e.g., users), we assume each $P_i$ has a set of \textit{user-facing} actions $\user(P_i)$. Similar to a process's interactions with services, a user's interaction with a process is modeled through input-output pairs of user-facing actions.

We make two final assumptions about the processes: First, we assume that while
each process has access to a local clock, which is part of its state, and may
set local timers, which are internal actions, the process makes no assumptions
about the drift or skew of its clock relative to others. Second, processes only
invoke an operation on a service when they have no outstanding send-to or
receive-from actions at any channels.

Let $P = \prod_{i \in I} P_i$ be the composition of the $n$ processes and $C = \prod_{1 \leq i \leq n} \prod_{1 \leq j \leq n} C_{ij}$ be the composition of $n^2$ channels. Let $\alpha$ be an execution of the distributed application $P \times C$. $\alpha$ is \textit{well-formed} if it satisfies three criteria:
First, for all $P_i$, $\trace(\alpha) | \system(P_i)$ must be a sequence of alternating invocation and matching response actions, starting with an invocation. Second, for all $C_{ij}$, $\alpha | C_{ij}$ is well-formed. Third, for all $P_i$, $P_i$ does not take an output step while waiting for a response from some service. 

\subsubsection{Real-Time Precedence}
\label{sec:proof:preliminaries:rt-precedence}

Given an execution $\alpha$, let $\complete(\alpha)$ be the maximal subsequence of
$\alpha$ comprising only matching system-facing invocations and responses~\cite{herlihy1990linearizability}. Further, let $S$ be a sequence of invocation and response actions. We define an irreflexive partial order $\rt_S$ as the real-time order induced by $S$; $\pi_1 \rt_S \pi_2$ if and only if $\pi_1$ is a response action, $\pi_2$ is an invocation action, and $\pi_1 <_S \pi_2$ where $<_S$ is the total order defined by $S$.

Execution $\alpha$ of $P \times C$ thus induces the irreflexive partial order over the invocation and response actions it contains defined by $\rt_{\trace(\alpha)|\system(P)}$. For simplicity, we denote this as simply $\rt_\alpha$. A well-formed execution $\alpha_1$ of $P \times C$ satisfies \textit{real-time precedence} if $\alpha_1$ can be extended to $\alpha_2$ adding zero or more response actions such that there exists a sequence $S \in \spec$ where (1)
for all processes $P_i$, $\complete(\alpha_2) | P_i = S | P_i$, and
(2) $\rt_{\alpha_1} \subseteq <_S$.

\subsubsection{Potential Causality}
\label{sec:proof:preliminaries:causality}

To define the next notion of precedence, we must first define causality. An execution $\alpha$ induces an irreflexive partial order $\caused_\alpha$ on its actions,
reflecting the notion of \textit{potential causality}
\cite{lamport1978clocks,ahamad1995causal,lloyd2011cops}. $\pi_1 \caused_\alpha
\pi_2$ if one of the following is true:
\begin{enumerate}
\item $\pi_1$ precedes $\pi_2$ in some process's local execution $\alpha|P_i$;
\item $\pi_1$ is a $\sendto_{ij}(m)$ action and $\pi_2$ is its corresponding $\receive_{ij}(m)$ action;
\item $\pi_1$ is a response action of operation $o_1$ and $\pi_2$ is an invocation action of operation $o_2$ such that $o_2$'s return value includes the effect of $o_1$; or
\item there exists some action $\pi_3$ such that $\pi_1 \caused_\alpha \pi_3$ and $\pi_3 \caused_\alpha \pi_2$.
\end{enumerate}

The meaning of item three depends on the type and specification of the service that the processes interact with through their system-facing actions. For example, for a shared register, $\pi_1 \caused_\alpha \pi_2$ if $\pi_1$ is the response of a write and $\pi_2$ is the invocation of a read that returns the written value. Similarly, for a FIFO queue, $\pi_1 \caused_\alpha \pi_2$ if $\pi_1$ is the response of an enqueue operation and $\pi_2$ is the invocation of a dequeue operation that returns the enqueued value.

This definition of
potential causality subsumes prior definitions, which either only consider
messages passed between processes~\cite{lamport1978clocks} or only consider
causality between operations on a shared data
store~\cite{ahamad1995causal,lloyd2011cops}.

\subsubsection{Causal Precedence}
\label{sec:proof:preliminaries:causal-precedence}

A well-formed execution $\alpha_1$ of $P \times C$ satisfies \textit{causal precedence} if $\alpha_1$ can be extended to $\alpha_2$ adding zero or more
response actions such that there exists a sequence $S \in \spec$ where (1)
for all processes $P_i$, $\complete(\alpha_2) | P_i = S | P_i$, and
(2) for all pairs of system-facing actions $\pi_1$ and $\pi_2$, $\pi_1 \caused_{\alpha_1} \pi_2 \implies \pi_1 <_S \pi_2$.

\subsection{Causal Precedence Maintains Invariants}
\label{sec:proof:causal-precedence}

We now show that for every well-formed, finite execution $\alpha$ of $P \times
C$ that satisfies causal precedence, there is a corresponding well-formed,
finite execution $\beta$ of $P \times C$ that satisfies real-time precedence
such that each process proceeds through the same sequence of actions and states.

\begin{lemma}
  \label{theorem:executions}
  Suppose $\alpha$ is a finite execution of $P \times C$ that satisfies causal precedence. Then there exists a finite execution $\beta$ of $P \times C$ such that $\beta$ satisfies real-time precedence and for all processes $P_i$, $\alpha|P_i = \beta|P_i$.
\end{lemma}

\begin{proof}
  We show how to construct $\beta$ from an arbitrary, finite
  $\alpha$. We start by focusing on the actions in $\alpha$, so let $\alpha^\prime =
  \sched(\alpha)$. We first construct a schedule $\beta^\prime$
  from $\alpha^\prime$ and then replace the states to get $\beta$.

  Since $\alpha$ satisfies causal precedence, there exists a sequence $S \in
  \spec$ such that $<_S$ respects $\caused_{\alpha}$ and thus
  $\caused_{\alpha^\prime}$. As defined in
  Sections~\ref{sec:proof:preliminaries:rt-precedence}
  and~\ref{sec:proof:preliminaries:causal-precedence}, $S$ includes zero or more
  response actions that are not necessarily in $\alpha^\prime$. Let $R$ denote
  this set of responses. $<_S$ is thus defined over the set of complete
  operations in $\alpha^\prime$ as well as the additional set of invocations
  completed by responses in $R$.

  As allowed by the definition in
  Section~\ref{sec:proof:preliminaries:causal-precedence}, however, there may be
  some invocation actions in $\alpha^\prime$ that do not have matching responses
  in $R$, so $<_S$ does not order them. Thus, we extend $<_S$ to $<_S^\prime$ by
  placing these removed invocation actions at the end of $<_S$ in an arbitrary
  order.

  Let $O$ be the totally ordered set defined by the set of all invocation and
  response actions in $\alpha^\prime$ and $<_S^\prime$. To get $\beta^\prime$,
  we reorder all of the actions in $\alpha^\prime$ by ordering each action
  after the maximal element of $O$ that causally precedes it.
  
  To do so, we first define two relations: Given two actions $\pi_1$ and
  $\pi_2$, let $\pi_1 \prec \pi_2 \iff \exists \pi_3 \in O, \forall \pi_4 \in O
  : (\pi_3 \caused_{\alpha^\prime} \pi_2 \land \pi_4 \caused_{\alpha^\prime}
  \pi_1) \implies \pi_4 <_S^\prime \pi_3$. Finally, let $\pi_1 \equiv \pi_2$ if
  $\pi_1 \not\prec \pi_2$ and $\pi_2 \not\prec \pi_1$. It is clear that $\prec$
  is an irreflexive partial order over the actions in $\alpha^\prime$.
  
  Let $<_{\alpha^\prime}$ be the total order of actions given by their order in $\alpha^\prime$.
  To define $<_{\beta^\prime}$, we use $<_{\alpha^\prime}$ to extend $\prec$ to a total order. In particular,
  let $\pi_1 <_{\beta^\prime} \pi_2$ if $\pi_1 \prec \pi_2$ or $\pi_1 \equiv
  \pi_1$ and $\pi_1 <_{\alpha^\prime} \pi_2$. Thus,
  $<_{\beta^\prime}$ is a total order over the actions in $\alpha^\prime$. Let $\beta^\prime$ be the schedule defined by $<_{\beta^\prime}$.

  We show that this does not reorder any actions at any of the processes, and
  thus $\alpha^\prime|P_i = \beta^\prime| P_i$ for all $P_i$. Assume to
  contradict that there exists some pair of actions $\pi_1,\pi_2$ from the same
  $P_i$ that have been reordered in $\beta^\prime$. Without loss of generality,
  assume $\pi_2 <_{\alpha^\prime} \pi_1$ but $\pi_1 <_{\beta^\prime} \pi_2$.

  It is clear that $\pi_1 \not\equiv \pi_2$ because otherwise by the definition
  of $<_{\beta^\prime}$, $\pi_1$ and $\pi_2$ would be ordered identically in
  $\alpha^\prime$ and $\beta^\prime$. Thus, by the definition of
  $<_{\beta^\prime}$ and the assumption that $\pi_1 <_{\beta^\prime} \pi_2$, it
  must be that $\pi_1 \prec \pi_2$.

  Let $\pi_3 \in O$ be as defined by $\pi_1 \prec \pi_2$, so by definition,
  $\pi_3 \caused_{\alpha^\prime} \pi_2$. Observe that since $\pi_2
  <_{\alpha^\prime} \pi_1$ and $\pi_1$ and $\pi_2$ are from the same process,
  $\pi_2 \caused_{\alpha^\prime} \pi_1$. Thus, $\pi_3 \caused_{\alpha^\prime}
  \pi_1$ by the transitivity of $\caused_{\alpha^\prime}$. But since
  $<_S^\prime$ is an irreflexive total order, $\pi_3 \not<_S^\prime \pi_3$,
  contradicting the definition of $\pi_3$. Thus, $\beta^\prime$ is a schedule
  such that $\alpha^\prime|P_i = \beta^\prime|P_i$ for all $P_i$.
  
  To get $\beta$ from $\beta^\prime$, we now just need to replace
  the states of $P \times C$. Since we did not reorder any of the actions at any
  of the processes, we can simply use the same process states from $\alpha$ for
  the process components of the states in $\beta$.

  We just need to fill in the states of each channel $C_{ij}$. Since
  $\alpha^\prime|P_i = \beta^\prime| P_i$ for all $P_i$ and by the definitions
  of $\caused_{\alpha}$ and causal precedence, only some pairs of actions may
  have been reordered when transforming $\alpha^\prime$ to $\beta^\prime$.
  Specifically, the order of a $\sendto_{ij}(m)$ action may be reordered with
  respect to a $\recvfrom_{ij}$ action or a $\receive_{ij}(m^\prime)$ action for
  $m \neq m^\prime$. Similarly, the order of a $\sent_{ij}$ action may be
  reordered with a $\recvfrom_{ij}$ or a $\receive_{ij}(m)$ action.

  Let $\kappa_1 = \alpha|C_{ij}$, $\kappa_1^\prime = \trace(\kappa_1)$, and
  $\kappa_2^\prime = \beta^\prime|C_{ij}$. Since the only possible differences
  between $\kappa_1^\prime$ and $\kappa_2^\prime$ are the re-orderings described
  above, it is possible, using repeated applications of
  Lemmas~\ref{lem:channels-commute1} through~\ref{lem:channels-commute4} if
  necessary, to find an execution $\kappa_2$ from $\kappa_1$ such that
  $\trace(\kappa_2) = \kappa_2^\prime$. Thus, to fill in the states of each
  $C_{ij}$, we first find $\kappa_2$ and then use the states from this execution
  to fill in $C_{ij}$'s states in $\beta$.
  
  Since we previously showed that $\alpha^\prime|P_i = \beta^\prime| P_i$, it is
  clear that $\alpha|P_i = \beta|P_i$ for all $P_i$ and $\beta$ is well-formed.
  Further, because we did not add any states or actions at any of the processes
  or channels, it is clear that $\beta$ is finite.
  
  Since $S \in \spec$ was a sequence of matching invocation-response pairs, by
  the definition of $<_{\beta^\prime}$, $\complete(\beta)$ is sequential. But
  $\complete(\beta) | P_i$ may not equal $S|P_i$ if $\textit{stop}_i$ occurred
  while $P_i$ was waiting for a service response and the operation's effects
  affected the responses of operations by other processes. To show that $\beta$
  satisfies real-time precedence, we thus extend $\beta$ to $\gamma$ by adding
  those response actions that are in $S$ but are not in $\beta$. These are
  exactly the response actions in $R$ (defined above) that were originally added
  to satisfy the definition of causal precedence for $\alpha$. After this
  addition, it is clear that $\gamma|P_i = S|P_i$ for all $P_i$ and since
  $\complete(\beta)$ is sequential, it respects the real-time precedence of the
  operations. Thus, $\beta$ thus satisfies real-time precedence.
\end{proof}

\begin{theorem}
  \label{theorem:invariants}
  Suppose $\invariant{P}$ is an invariant that holds for any execution $\beta$ of $P \times C$ that satisfies real-time precedence. Then $\invariant{P}$ also holds for any execution $\alpha$ of $P \times C$ that satisfies causal precedence.
\end{theorem}

\begin{proof}
  Let $\alpha$ be an arbitrary, finite execution of $P \times C$ that satisfies causal precedence. We must show that $\invariant{P}$ is true for the final state $s$ of $\alpha|P$.

  By Lemma~\ref{theorem:executions}, there exists a finite execution $\beta$ of $P
  \times C$ that satisfies real-time precedence and for all processes $P_i$, $\alpha|P_i = \beta|P_i$.
  
  Let $s^\prime$ be the final state of $\beta|P$. Because $\alpha|P_i =
  \beta|P_i$ for all $P_i$, it is easy to see that $s^\prime = s$. By assumption, $\invariant{P}$ is true of $s^\prime$, so
  $\invariant{P}$ is also true of $s$.
\end{proof}

\subsection{\rs{} and \rss{} Maintain Application Invariants}
\label{sec:proof:rs}

In this section, we leverage the results above to show that our new consistency models, \rslong{} and \rsslong{}, maintain application invariants that hold with linearizability~\cite{herlihy1990linearizability} and strict serializability~\cite{papadimitriou1979serializability}, respectively. We start by defining the two new consistency models and then present the results.

\subsubsection{\RSlong{}}
\label{sec:proof:rs:rs}

\textit{\Rslong{}} (\rs{}) guarantees casual precedence and also requires that writes respect their real-time order. \newtext{In particular, any conflicting operation that follows a write in real-time must reflect the state change of that write in its result.}

Let $\mathcal{W} \subseteq \mathcal{O}$ be the subset of the (possibly composite) service's operations that mutate its value.
\newtext{Further, given an execution $\alpha$ and write $w \in \mathcal{W}$, define
$\conflicts_\alpha(w)$ as the set of non-mutating operations in $\alpha$ that
conflict with $w$.}
A well-formed execution $\alpha_1$ of $P \times C$ satisfies \rs{} if $\alpha_1$
can be extended to $\alpha_2$ by adding zero or more response actions such that
there exists a sequence $S \in \spec$ where (1) for all processes $P_i$,
$\complete(\alpha_2) | P_i = S | P_i$; (2) for all pairs of invocation and
response actions $\pi_1$ and $\pi_2$, $\pi_1 \caused_{\alpha_1} \pi_2 \implies
\pi_1 <_S \pi_2$; and (3)
\newtext{for all response actions $\pi_1$ of $w \in \mathcal{W}$ and invocation actions $\pi_2$
of $o \in \conflicts_{\alpha_1}(w) \cup \mathcal{W}$ , $\pi_1
\rt_{\alpha_1} \pi_2 \implies \pi_1 <_S \pi_2$.}

\subsubsection{\RSSlong{}}
\label{sec:proof:rs:transactions}

Before defining \rsslong{}, we must first discuss how transactions can be defined within the formal framework presented above. Fortunately, the formalism of types, services, and sequential specifications presented in Section~\ref{sec:proof:preliminaries:services} is sufficiently general and can be easily adapted to transactions. 

\noindentparagraph{Transactional services.} We refer to services that support transactions as \textit{transactional services}. For example, consider a transactional key-value store $\mathcal{D}$ that stores a mapping from a set of keys $\mathcal{K}$ and values $\mathcal{V}$ including some initial value $\bot$. The read-only transaction operation is defined with invocations $\{\textsc{ro}_{i,\mathcal{D}}(K)\}$ and responses
$\{\textsc{ret}_{i,\mathcal{D}}(V)\}$ for all $0 \leq i \leq n$, $K \in 2^\mathcal{K}$, and $V \in 2^\mathcal{V}$. Let $f : 2^\mathcal{K} \times 2^\mathcal{V} \times K \to V$ be a function that takes as input a set of keys, their corresponding values, and a single key (that may or may not be in first set) and returns a value. The read-write transaction operation then has invocations $\{\textsc{rw}_{i,\mathcal{D}}(R, W, f)\}$ and responses $\{\textsc{ack}_{i,\mathcal{D}}\}$ for all $0 \leq i \leq n$ and $R, W \in 2^\mathcal{K}$.

The transactional key-value store's sequential specification is the set of all sequences of read-only and read-write transactions satisfying the following: (1) Reads of a key $k$ in both read-only and read-write transactions return the most recently written value for $k$ or $\bot$ if none exists. (2) For each $k \in W$ of a read-write transaction, the transaction writes value $f(R,V,k)$ where $V$ is the set of read values.

\noindentparagraph{Composition.} Unlike some prior work~\cite{herlihy1990linearizability}, when composing the types and sequential specifications of transactional services, we do not assume that transactions are extended across multiple services in the composition. For instance, the composition of two transactional key-value stores does not yield a single transactional key-value store whose operations are transactions that possibly span both key sets. As a result, the prior definitions can be used without modification for transactional services.

\noindentparagraph{\RSSlong{}.} Thanks to the generality of our definition, \textit{\rsslong{}} (\rss{}) is simply \rslong{} applied to a transactional service, such as the transactional key-value store described above. The set of writes $\mathcal{W} \subseteq \mathcal{O}$ is simply the set of read-write transactions, \newtext{and given a read-write transaction $w \in \mathcal{W}$, the set of non-mutating conflicts $\conflicts_\alpha(w)$ in an execution $\alpha$ is simply the set of read-only transactions that read a key written by $w$.}

\subsubsection{Proof Results}

Given the definitions above, we are now ready to show that \rs{} and \rss{} maintain application invariants. In fact, these results follow as corollaries of Theorem~\ref{theorem:invariants}.

\begin{corollary}
  Suppose $\invariant{P}$ is an invariant that holds for any execution $\beta$ of $P \times C$ that satisfies linearizability. Then $\invariant{P}$ also holds for any execution $\alpha$ of $P \times C$ that satisfies \rs{}.
\end{corollary}

\begin{proof}
  By their definitions, linearizability~\cite{herlihy1990linearizability} guarantees real-time precedence and \rs{} guarantees causal precedence for a set of non-transactional services, respectively. Thus, the corollary follows immediately from Theorem~\ref{theorem:invariants}.
\end{proof}

\begin{corollary}
  Suppose $\invariant{P}$ is an invariant that holds for any execution $\beta$ of $P \times C$ that satisfies strict serializability. Then $\invariant{P}$ also holds for any execution $\alpha$ of $P \times C$ that satisfies \rss{}.
\end{corollary}

\begin{proof}
  By their definitions, strict serializability~\cite{papadimitriou1979serializability} guarantees real-time precedence and \rss{} guarantees causal precedence for a set of transactional services, respectively. Thus, the corollary follows immediately from Theorem~\ref{theorem:invariants}.
\end{proof}

\subsection{\rs{} Composition Using \RTBarriers{}}
\label{sec:proof:composition}

The definitions and proofs in this section mirror very similar results proved for ordered sequential consistency (OSC)~\cite{levari2017osc}. The differences in the proofs primarily result from differences in the definitions of OSC and \rs{} and differences in our notation. The techniques and proof steps are nearly identical, but we include them for completeness.

The main result in this section shows that a special mechanism, a \textit{\rtbarrier{}}, can be used to compose a set of \rs{} services and ensure their composition satisfies \rs{}. As a result, the results in the previous section regarding application invariants will hold. 

\subsubsection{Definitions and Assumptions}

We focus here on composition, so we need to distinguish between a sequence $S \in \spec$ of a composite \rs{} service, as used above, and the corresponding serializations of each service $x \in X$. We denote such serializations as $S_x \in \spec_x$.

In some of the definitions and results below, instead of assuming that an execution $\alpha$ satisfies \rs{}, we assume each service $x \in X$ individually satisfies \rs{}. In other words, we do not assume there is a sequence $S \in \spec$ of the composite service that satisfies the definition in Section~\ref{sec:proof:rs:rs}, and instead assume each $S_x \in \spec_x$ satisfies the definition. In the remainder of the section, we make clear which we assume.

\subsubsection{\RTBarriers{}}

\textit{\Rtbarriers{}} are special operations, one per service $x \in X$, that help compose a set of \rs{} services. Each \barrier{} $f_x$ has exactly one invocation and response action, which we denote $i_{f_x}$ and $r_{f_x}$, and gives the following guarantees: Let $\alpha$ be a well-formed execution and $f_x$ be a real-time \barrier{} on service $x$. Then for all system-facing actions $\pi \in S_x$, (1) if $\pi \caused_\alpha i_{f_x}$, then $\pi <_{S_x} i_{f_x}$; and (2) if $r_{f_x} \rt_\alpha \pi$, then $r_{f_x} <_{S_x} \pi$. As a result, any system-facing actions that causally precede the \barrier{} are serialized in $S$ before any that follow it in real time.
 
Given an execution $\alpha$ of services $X$ that individually satisfy \rs{}, we define a \barrier{} $f_x$'s \textit{past set}, denoted $\pastset_\alpha(f_x)$, as the set of actions $\pi \in S_x$ such that $\pi \leq_{S_x} i_{f_x}$ where $\leq_{S_x}$ extends $<_{S_x}$ in the natural way. Further, define a \barrier{}'s \textit{last invocation} $\lastinv_\alpha(f_x)$ as the latest invocation in $\pastset_\alpha(f_x)$. Note $\lastinv_\alpha(f_x)$ can be $i_{f_x}$.

\subsubsection{Proof}

Before we can prove our main result, we first introduce several lemmas. First, we show that we can define a total order over the set of \barriers{} in an execution, even if those \barriers{} were issued at different \rs{} services. Next, we lift this to a total order over all system-facing actions at all services. Finally, we leverage this total order to show that if processes follow a simple protocol, then the composition of a set of \rs{} services also satisfies \rs{}. 

\begin{lemma}
\label{theorem:composition:lem-1}
Given an execution $\alpha$ of services $X$ that individually satisfy \rs{}, for all \barriers{} $f_x$, $\lastinv_\alpha(f_x) <_\alpha r_{f_x}$, where $<_\alpha$ is the strict total order of actions defined by $\alpha$.
\end{lemma}

\begin{proof}
We prove by contradiction, so suppose $r_{f_x} <_\alpha \lastinv_\alpha(f_x)$. Then clearly, $\lastinv_\alpha(f_x) \neq i_{f_x}$. Since $r_{f_x} <_\alpha \lastinv_\alpha(f_x)$, $r_{f_x} \rt_\alpha \lastinv_\alpha(f_x)$ by the definition of $\rt_\alpha$. But then by the definition of $f_x$, $r_{f_x} <_{S_x} \lastinv_\alpha(f_x)$, contradicting the definition of $\lastinv_\alpha(f_x)$.
\end{proof}

\begin{lemma}
\label{theorem:composition:lem-2}
Given an execution $\alpha$ of services $X$ that individually satisfy \rs{}, let $f_x$ and $f_x^\prime$ be two \barriers{} in $S_x$. If $r_{f_x} <_{S_x} i_{f_x^\prime}$, then $\lastinv_\alpha(f_x) \leq_\alpha \lastinv_\alpha(f_x^\prime)$, where $\leq_\alpha$ extends $<_\alpha$ in the natural way.
\end{lemma}

\begin{proof}
Since $r_{f_x} <_{S_x} i_{f_x^\prime}$, $\pastset_\alpha(f_x) \subset \pastset_\alpha(f_x^\prime)$ by the definition of $\pastset_\alpha$. Then since $\lastinv_\alpha(f_x) \in \pastset_\alpha(f_x)$ and $\pastset_\alpha(f_x) \subset \pastset_\alpha(f_x^\prime)$, $\lastinv_\alpha(f_x) \in \pastset_\alpha(f_x^\prime)$. Thus, either $\lastinv_\alpha(f_x) = \lastinv_\alpha(f_x^\prime)$ or there is some later invocation action in $\pastset(f_x^\prime)$ with a later invocation. As a result, $\lastinv_\alpha(f_x) \leq_\alpha \lastinv_\alpha(f_x^\prime)$.
\end{proof}

Given an execution $\alpha$ of services $X$ that individually satisfy \rs{}, we use these lemmas to define a strict total order over all of the \barriers{} in $\alpha$. For each pair of \barriers{} $f_x$ and $f_y$, we define $f_x \triangleleft f_y$ as follows: If $x = y$, then $f_x \triangleleft f_y$ if and only if $r_{f_x} <_{S_x} i_{f_y}$; otherwise, $x \neq y$, and $f_x \triangleleft f_y$ if and only if $\lastinv_\alpha(f_x) <_\alpha \lastinv_\alpha(f_y)$.

\begin{lemma}
\label{theorem:composition:lem-3}
$\triangleleft$ is a strict total order.
\end{lemma}

\begin{proof}
We must prove $\triangleleft$ is irreflexive, transitive, and total. Irreflexivity and totality follow from the definitions of $\triangleleft$ and $<_{S_x}$. We now show $\triangleleft$ is transitive.

Let $f_x$, $f_y$, and $f_z$ be \barriers{} such that $f_x \triangleleft f_y$ and $f_y \triangleleft f_z$. We must show $f_x \triangleleft f_z$. There are four cases:

If $x = y = z$, then the transitivity of $<_{S_x}$ implies $r_{f_x} <_{S_x} i_{f_z}$, so $f_x \triangleleft f_z$. If $x = y \neq z$, then since $x=y$ and $f_x \triangleleft f_y$, $\lastinv_\alpha(f_x) \leq_\alpha \lastinv_\alpha(f_y)$ by the Lemma~\ref{theorem:composition:lem-2}. By the definition of $\triangleleft$, since $y \neq z$, $\lastinv_\alpha(f_y) <_\alpha \lastinv_\alpha(f_z)$, so $\lastinv_\alpha(f_x) <_\alpha \lastinv_\alpha(f_z)$ and $f_x \triangleleft f_z$. Similar reasoning applies to the case where $x \neq y = z$. Finally, if $x \neq y$ and $y \neq z$, then by the definition of $\triangleleft$ and the transitivity of $<_\alpha$, $\lastinv_\alpha(f_x) <_\alpha \lastinv(f_z)$. Clearly if $x \neq z$, then $f_x \triangleleft f_z$. Further, if $x = z$, then by the contrapositive of Lemma~\ref{theorem:composition:lem-2}, $i_{f_x} <_{S_x} r_{f_z}$. Then since $S_x$ is a sequence of invocation-response pairs, $r_{f_x} <_{S_x} i_{f_z}$, so $f_x \triangleleft f_z$.
\end{proof}

$\triangleleft$ defines a strict total order over the \barriers{} in executions involving multiple \rs{} services. To extend this to a total order over all system-facing actions, we first define a system-facing action's \textit{next \barrier{}}, denoted $\nextbar_\alpha(\pi)$. Specifically, given an execution $\alpha$ of services $X$ that individually satisfy \rs{}, a \rs{} service $x \in X$, and a system-facing action $\pi_x \in S_x$, define $\nextbar_\alpha(\pi_x)$ as the earliest \barrier{} $f_x$ such that $\pi_x \leq_{S_x} r_{f_x}$. To ensure $\nextbar_\alpha(\pi)$ is defined for all $\pi$, we assume $\alpha$ is augmented with a sequence of \barrier{} invocation-response pairs $i_{\top_x},r_{\top_x}$, one for each $x \in X$, that are added to the end of $\alpha$.

We use next \barriers{} to lift $\triangleleft$ to all system-facing actions. Let $\alpha$ be an execution of services $X$ that satisfy \rs{} individually; $x,y \in X$ be two \rs{} services; and $\pi_x,\pi_y$ be two system-facing actions on $x$ and $y$, respectively. Define $\pi_x \prec \pi_y$ as follows: If $\nextbar_\alpha(\pi_x) \neq \nextbar_\alpha(\pi_y)$, then $\pi_x \prec \pi_y$ if and only if $\nextbar_\alpha(\pi_x) \triangleleft \nextbar_\alpha(\pi_y)$; otherwise, $\nextbar_\alpha(\pi_x) = \nextbar_\alpha(\pi_y)$, so $x = y$ and $\pi_x \prec \pi_y$ if and only if $\pi_x <_{S_x} \pi_y$.

We prove two facts about $\prec$: First, it is a strict total order. Second, it generalizes the $S_x$ of each service $x \in X$.

\begin{lemma}
\label{theorem:composition:lem-4}
$\prec$ is a strict total order.
\end{lemma}

\begin{proof}
We must prove $\prec$ is irreflexive, transitive, and total. Irreflexivity and totality follow from the definitions of $\triangleleft$ and $<_{S_x}$. We now show $\prec$ is transitive.

Let $\pi_x$, $\pi_y$, and $\pi_z$ be system-facing actions on services $x$, $y$, and $z$, respectively, such that $\pi_x \prec \pi_y$ and $\pi_y \prec \pi_z$. We must show $\pi_x \prec \pi_z$.

By the definition of $\prec$, since $\pi_x \prec \pi_y$, $\nextbar_\alpha(\pi_x) \trianglelefteq \nextbar_\alpha(\pi_y)$, where $\trianglelefteq$ extends $\triangleleft$ in the natural way. By similar reasoning, $\nextbar_\alpha(\pi_y) \trianglelefteq \nextbar_\alpha(\pi_z)$, so $\nextbar_\alpha(\pi_x) \trianglelefteq \nextbar_\alpha(\pi_z)$ by the transitivity of $\triangleleft$.

If $\nextbar_\alpha(\pi_x) \triangleleft \nextbar_\alpha(\pi_z)$, then by the definition of $\prec$, $\pi_x \prec \pi_z$, so suppose $\nextbar_\alpha(\pi_x) = \nextbar_\alpha(\pi_z)$. Then since $\nextbar_\alpha(\pi_x) \trianglelefteq \nextbar_\alpha(\pi_y) \trianglelefteq \nextbar_\alpha(\pi_z)$, it must be the case that $x = y = z$, so by the definition of $\prec$, $\pi_x <_{S_x} \pi_y$ and $\pi_y <_{S_x} \pi_z$. Thus, $\pi_x <_{S_x} \pi_z$, and by the definition of $\prec$, $\pi_x \prec \pi_z$.
\end{proof}

\begin{lemma}
\label{theorem:composition:lem-5}
Given an execution $\alpha$ of services $x \in X$ that individually satisfy \rs{}, a \rs{} service $x \in X$, and two system-facing actions $\pi_1,\pi_2 \in S_x$, if $\pi_1 <_{S_x} \pi_2$, then $\pi_1 \prec \pi_2$.
\end{lemma}

\begin{proof}
By the definition of $\nextbar_\alpha$, since $\pi_1 <_{S_x} \pi_2$, $r_{\nextbar_\alpha(\pi_1)} \leq_{S_x} r_{\nextbar_\alpha(\pi_2)}$. If $r_{\nextbar_\alpha(\pi_1)} = r_{\nextbar_\alpha(\pi_2)}$, then clearly $\pi_1 \prec \pi_2$, so assume $r_{\nextbar_\alpha(\pi_1)} <_{S_x} r_{\nextbar_\alpha(\pi_2)}$. Since $S_x$ is a sequence of alternating invocation-response pairs, $\nextbar_\alpha(\pi_1) \triangleleft \nextbar_\alpha(\pi_2)$ by the definition of $\triangleleft$, so $\pi_1 \prec \pi_2$.
\end{proof}

We are now ready to prove our main result of the section. We show that if processes follow a simple protocol, then the composition of a set of \rs{} services satisfies \rs{}. More specifically, a process interacting with service $x$ must issue a \barrier{} to $x$ before interacting with another service $y$. This applies to sets of processes interacting through message passing, too. If $P_1$ issues an operation at $x$, sends a message to $P_2$, and $P_2$ receives it, then $P_2$ must issue a \barrier{} to $x$ before interacting with a different service $y$.

\begin{theorem}
\label{theorem:composition:main-theorem}
  Let $\alpha$ be an execution of service $X$ that individually satisfy \rs{}, and assume processes issue real-time \barriers{} between interactions with different services, as described above. Then $\alpha$ satisfies \rs{}.
\end{theorem}

\begin{proof}
We prove $\alpha$ satisfies \rs{} by construction. Specifically, let $S$ be the sequence of system-facing actions in $\alpha$ defined by $\prec$. We must prove $S$ satisfies \rs{}.

By Lemma~\ref{theorem:composition:lem-5}, since $\prec$ and thus $S$ generalizes the sequences of each $S_x \in \spec_x$, $S \in \spec$. Similarly, Lemma~\ref{theorem:composition:lem-5} implies the third requirement of \rs{} is satisfied.

We now prove that $S$ respects causality, so let $\pi_1$ and $\pi_2$ be system-facing actions such that $\pi_1 \caused_\alpha \pi_2$. To start, we only consider the case where there is no system-facing action $\pi_3$ such that $\pi_1 \caused_\alpha \pi_3$ and $\pi_3 \caused_\alpha \pi_2$.

If both actions are on the same service $x$, then since $\pi_1 \caused_\alpha \pi_2$ and $S_x$ satisfies \rs{}, $\pi_1 <_{S_x} \pi_2$. By Lemma~\ref{theorem:composition:lem-5}, $\pi_1 <_S \pi_2$.

Now suppose $\pi_1$ is on service $x$ and $\pi_2$ is on service $y$ with $x \neq y$. Since processes issue \barriers{} between interactions with different services, it must be the case that $\pi_1 = r_{f_x}$ for some \barrier{} $f_x$ and $\pi_2$ is an invocation action on $y$. By the definition of $\nextbar_\alpha$, $\nextbar_\alpha(\pi_1) = \pi_1$. Further, by Lemma~\ref{theorem:composition:lem-1}, $\lastinv_\alpha(f_x) <_\alpha r_{f_x} = \pi_1$.

Since $\pi_1$ and $\pi_2$ are on different services, $\pi_2$ cannot be part of an operation whose return value includes the effects of $\pi_1$'s operation (i.e., $f_x$). Combining this with the assumption that $\pi_1 \caused_\alpha \pi_2$, we get $\pi_1 <_\alpha \pi_2$.

Let $f_y = \nextbar_\alpha(\pi_2)$. By the definitions of $\lastinv_\alpha$ and $\nextbar_\alpha$, it must be the case that $\pi_2 \leq_\alpha \lastinv_\alpha(f_y)$; either $\pi_2 = \lastinv_\alpha(f_y)$, or there is some later last invocation. Combining this with the facts that $\lastinv_\alpha(f_x) <_\alpha \pi_1$ and $\pi_1 <_\alpha \pi_2$, we see that $\lastinv_\alpha(f_x) <_\alpha \lastinv_\alpha(f_y)$. By the definitions of $\prec$ and $\triangleleft$, $\pi_1 \prec \pi_2$.

We now consider transitivity, so suppose there is some system-facing action $\pi_3$ such that $\pi_1 \caused_\alpha \pi_3$ and $\pi_3 \caused_\alpha \pi_2$. By the reasoning above, $\pi_1 \prec \pi_3$ and $\pi_3 \prec \pi_2$. Then by Lemma~\ref{theorem:composition:lem-4}, which shows $\prec$ is transitive, we conclude that $\pi_1 \prec \pi_2$. Thus, $S$ satisfies the second requirement of \rs{}.

Finally, since $S$ respects causality and by Lemma~\ref{theorem:composition:lem-5}, generalizes each $S_x$, it is clear that $S$ respects the order of system-facing invocations and responses at each process. Thus, $S$ satisfies the first requirement of $\rs{}$.
\end{proof}


%% file: sections/correctness_proofs.tex
\section{Proofs of Correctness}
\label{sec:correctness}

In Sections~\ref{sec:spannerrss} and~\ref{sec:gryffrs}, we presented two new
protocols: a variant of \spanner{} that relaxes its consistency from strict
serializability~\cite{papadimitriou1979serializability} to \rss{} and a variant
of \gryff{} that relaxes its consistency from
linearizability~\cite{herlihy1990linearizability} to \rs{}. The designs are
agnostic to the structure of the applications using them, but we make some basic
assumptions about the application's runtime depending on the structure of the
application.

If a set of clients (e.g., mobile phones) use the services directly and do not
communicate outside of the service via message passing, then it is sufficient for
each client to simply use the client libraries to communicate with
the services. If clients do communicate via message passing (e.g., a mobile
phone proxying its requests through multiple Web servers), then as discussed in
Section~\ref{sec:discussion:mp}, some metadata must be propagated
between processes executing on different machines to ensure that the services
return values reflecting all causal constraints. Fortunately, existing
frameworks, such as Baggage Contexts~\cite{mace2018baggageContexts}, can
automatically propagate this metadata between processes.

In the proofs below, we assume this metadata propagation, if necessary, is
implemented correctly within the application's runtime. For \spannerrss{}, this
metadata is the minimum read timestamp $\tmin$, and for \gryffrs{}, it is the
dependency tuple $\textit{dep}$ that is piggybacked on the next interaction with
\gryffrs{}.

\subsection{\spannerrss{}}
\label{sec:correctness:spannerrss}

We begin with three observations about \spanner{}'s protocol:

\noindentparagraph{Observation 1.} If a read-write (RW) transaction has committed
at a shard with key $k$ and timestamp $\tcommit$, then there cannot be a current
or future prepared transaction that writes $k$ and has a prepare timestamp less
than $\tcommit$. This follows from \spanner{}'s use of strict two-phase locking
and the fact that each RW transaction chooses its prepare timestamp to be
greater than all previously committed writes at each participant
shard~\cite{corbett2013spanner}.

\noindentparagraph{Observation 2.} If two RW transactions conflict, then
their commit timestamps cannot be equal. This follows from \spanner{}'s use of
strict two-phase locking and the fact that each RW transaction chooses its
commit timestamp to be greater than the prepare timestamp from each participant
shard.

\noindentparagraph{Observation 3.} The commit timestamp of each RW transaction is
guaranteed to be between its real start and end times. This is shown in the
original paper~\cite{corbett2013spanner}
and is what makes \spanner{} strictly serializable.

We now prove several supporting lemmas and then use them to prove the
correctness of \spannerrss{}. We define a transaction's timestamp $t$
as $\tcommit$ if it is a RW transaction and $\tsnap$ if it is a read-only (RO)
transaction. We denote a transaction $T_1$'s timestamp as $t_1$. We use the
transactions' timestamps to construct a total order.
Lemmas~\ref{theorem:spannerrss:causal-timestamps}
and~\ref{theorem:spannerrss:rt-timestamps} prove properties about
the timestamps of transactions related by causality or real time, and we use
them to show the constructed total order satisfies \rss{}.
Lemmas~\ref{theorem:spannerrss:rw-valid} and~\ref{theorem:spannerrss:ro-valid} are used to show
the constructed total order is in \spannerrss{}'s sequential
specification (i.e., that the order is consistent with the values
returned by each transaction's reads). 

\begin{lemma}
    \label{theorem:spannerrss:causal-timestamps}
    If $T_1$ and $T_2$ are transactions such that $T_1 \caused T_2$, then $t_1 \leq t_2$. Further, if $T_1$ and $T_2$ are both RW transactions, then $t_1 < t_2$.
\end{lemma}

\begin{proof}
  We first consider the four pairs of transactions. For each case, we consider the three direct causal relationships: process order, message passing, and reads-from. We then consider transitivity.

  $\text{RO}_1 \caused \text{RO}_2$. Observe that because the first transaction is RO, it is not possible for the second to read from the first. If the two RO transactions are causally related by process order or message passing, then the second RO transaction's $\tmin$ will be greater than or equal to the first's by the assumption that applications propagate the necessary metadata. As a result, line 6 of Algorithm~\ref{alg:spannerrss-shard} guarantees the second RO transaction will include any writes with $\tcommit \leq \tmin$, so $\tmin \leq t_2$. Thus, $t_1 \leq t_2$.
  
  $\text{RO} \caused \text{RW}$. As above, the RW transaction cannot read from the RO transaction. If the two transactions are causally related by process order or message passing, then the RO transaction must precede the RW transaction in real time. Because RW transactions perform commit wait, if a write is returned in a RO transaction, then that write's commit timestamp is guaranteed to be in the past before the RO transaction ends. As a result, $t_1$ is guaranteed to be less the RO's end time. Combined with the fact that a RW transaction's commit timestamp is guaranteed to be after its start time, this implies $t_1 < t_2$.
  
  $\text{RW} \caused \text{RO}$. If the RO transaction reads from the RW transaction, then clearly $t_1 \leq t_2$ by the way $\tsnap$ is calculated (Alg.~\ref{alg:spannerrss-client}, lines 15-20). If instead the transactions are causally related by process order or message passing, then because a process sets its $\tmin$ to be at least $\tcommit$ after a RW transaction finishes, line 6 of Algorithm~\ref{alg:spannerrss-shard} guarantees the RO transaction includes any writes with $\tcommit^\prime \leq \tcommit$, and thus $t_1 \leq t_2$.
  
  $\text{RW}_1 \caused \text{RW}_2$. If the two RW transactions are causally related by process order or by message passing, then $\text{RW}_1$ must precede $\text{RW}_2$ in real time. By observation 3 above about \spanner{}'s RW transactions (and thus \spannerrss{}'s), it must be that $t_1 < t_2$. If $\text{RW}_2$ reads from $\text{RW}_1$, then the two transactions conflict and \spannerrss{}'s use of strict two-phase locking guarantees $t_1 < t_2$.
  
  We now consider transitivity. Clearly, if $t_1 \leq t_2$ holds for each pair of causally related transactions, then $t_1 \leq t_2$ applies for pairs of transactions causally related through transitivity. Further, because in the cases two and four above we have shown $t_1 < t_2$, it must be that $t_1 < t_2$ for pairs of RW transactions causally related through transitivity.
\end{proof}

\begin{lemma}
    \label{theorem:spannerrss:rt-timestamps}
    If $T_1$ is a RW transaction and $T_2$ is a conflicting RO transaction such that $T_1 \rt T_2$, then $t_1 \leq t_2$.
\end{lemma}

\begin{proof}
  Since $T_1$ ends before $T_2$ starts, $t_1$ must be less than $T_2$'s start time. Further, line 4 of Algorithm~\ref{alg:spannerrss-client} guarantees the RO transaction's $\tread$ is greater than its start time, so $t_1 < \tread$.
  
  As a result, since $T_1$ has committed and its earliest end time $\elb$ has passed, when $T_2$ executes at any shard with conflicting keys, lines 6-8 of Algorithm~\ref{alg:spannerrss-shard} ensure it will read $T_1$'s write or one with a greater timestamp. Thus, $t_1 \leq t_2$. 
\end{proof}

\begin{lemma}
    \label{theorem:spannerrss:rw-valid}
    Suppose $T_2$ is a RW transaction that commits with timestamp $t_2$. Then for
    each key, $T_2$'s reads return the values written by the RW transaction with
    the greatest commit timestamp $t_1$ such that $t_1 < t_2$.
\end{lemma}

\begin{proof}
  Since \spannerrss{}'s RW transaction protocol is nearly identical to \spanner{}'s, this follows from the correctness argument for \spanner{}, which follows from the correctness of strict two-phase locking~\cite{berstein1987ccbook} and \spanner{}'s timestamp assignment~\cite{corbett2013spanner}.
\end{proof}

\begin{lemma}
    \label{theorem:spannerrss:ro-valid}
    Suppose $T_2$ is a RO transaction with a snapshot time of $t_2$. Then for each
    key, $T_2$ returns the values written by the RW transaction with the greatest
    commit timestamp $t_1$ such that $t_1 \leq t_2$.
\end{lemma}

\begin{proof}
  Let $T_1$ be an arbitrary RW transaction that conflicts with $T_2$ at keys $K$. Fix a $k \in K$, and assume $t_1$ is the greatest timestamp for a write of $k$ such that $t_1 \leq t_2$. For ease of exposition, assume each key resides on a different shard.
  
  We say a RO transaction begins executing at a shard once it finishes waiting for $\tread$ to be less than the Multi-Paxos maximum write timestamp (i.e., it reaches line 5 of Algorithm~\ref{alg:spannerrss-shard}). There are three cases. In each case, we either derive a contradiction or show that $T_2$ returns $T_1$'s write of $k$.
  
  First, suppose there is at least one $k^\prime \in K$ such that $T_1$ has not prepared at $k^\prime$'s shard when $T_2$ begins executing there. Then combining the facts that $\tsnap \leq \tread$, $T_2$ waits until $\tread \leq \textsc{Paxos::MaxWriteTS}$, and $T_1$'s prepare timestamp at each shard is chosen to be strictly greater than the shard's $\textsc{Paxos::MaxWriteTS}$, $T_1$'s prepare timestamp $\tprepare$ at the shard must be strictly greater than $\tread$. Since $\tsnap \leq \tread$ and $\tprepare \leq t_1$, this contradicts the assumption that $t_1 \leq t_2$.
  
  Next, for each $k \in K$, suppose $T_1$ has prepared but not committed at $k$'s shard when $T_2$ begins executing there. Because $\tprepare \leq t_1$ and $t_1 \leq t_2$, $\tprepare \leq t_2$. Further since $t_2 \leq \tread$, $\tprepare \leq \tread$, so lines 5, 9, and 10 of Algorithm~\ref{alg:spannerrss-shard} ensure $T_1$'s prepare timestamp is returned to $T_2$'s client.
  
  Since $\tprepare \leq t_1 \leq \tsnap = t_2$ by assumption, lines 22 and 23 of Algorithm~\ref{alg:spannerrss-client} ensure the client waits until $T_1$ commits. Once $T_1$ commits, lines 13-15 of Algorithm~\ref{alg:spannerrss-shard} and lines 10-11 of Algorithm~\ref{alg:spannerrss-client} transmit $T_1$'s values to $T_2$'s client. Since $t_1$ was assumed to be the greatest timestamp for key $k$ such that $t_1 \leq t_2$, line 13 of Algorithm~\ref{alg:spannerrss-client} returns $T_1$'s write of $k$.
  
  Finally, suppose $T_1$ has prepared at all shards containing keys $K$ and further, there is at least one $k^\prime \in K$ such that $T_1$ has committed at $k^\prime$'s shard when $T_2$ begins executing there. There are two sub-cases: If $T_1$ has not committed at the shard containing $k$ when $T_2$ begins executing there, then by similar reasoning as in the previous case, $T_1$'s write of $k$ will ultimately be sent to $T_2$'s client and returned. Now suppose $T_1$ has committed at $k$'s shard when $T_2$ begins executing there. We argue that line 8 of Algorithm~\ref{alg:spannerrss-shard} must return $T_1$'s write.
  
  Since line 8 returns the latest write with timestamp less than $\tread$, the only other possibility is that it returns some write with a timestamp $t_3 > t_1$. In this case, since $T_1$ already committed, its write would never be returned to $T_2$'s client. As a result, the value of $t_{\text{earliest}}$ calculated for key $k$ (Alg.~\ref{alg:spannerrss-client}, line 18) will be at least $t_3$, so by line 19 of Algorithm~\ref{alg:spannerrss-client}, $\tsnap$ would ultimately be at least $t_3$, contradicting the assumption that $t_1$ is the write of $k$ with the greatest timestamp less than or equal to $\tsnap = t_2$. Thus, line 8 of Algorithm~\ref{alg:spannerrss-shard} must return $T_1$'s write of $k$. By observation 1, there are no prepared transactions that write $k$ with timestamps less than $t_1$, so line 13 of Algorithm~\ref{alg:spannerrss-client} ultimately returns $T_1$'s write.
\end{proof}

\begin{theorem}
  \label{theorem:correctness:spanner}
  \spannerrss{} guarantees \rss{}.
\end{theorem}

\begin{proof}
  Let $\alpha_1$ be a well-formed execution of \spannerrss{}. We first construct a sequence $S$ of transaction invocations and responses from $\alpha_1$. We then use $S$ to extend $\alpha_1$ to $\alpha_2$ such that $S$ is equivalent to $\complete(\alpha_2)$ and finally, show $S$ satisfies \rss{}.

  To start, define a RW transaction as complete if it has committed at its coordinator and a RO transaction as complete if it has returned to its client.
  Further, using Observation 2 and the second part of
  Lemma~\ref{theorem:spannerrss:causal-timestamps}, observe that the set of
  transactions with a given timestamp $t$ comprises a set of non-conflicting,
  causally unrelated RW transactions and for each RW transaction, a set of
  casually related RO transactions. Thus, the set of transactions with a given
  timestamp $t$ can be arranged into a set of directed acyclic graphs (DAGs).
  Each DAG's vertices are a RW transaction and its causally related RO
  transactions, and each DAG's edges are defined by $\caused$ between the
  transaction vertices. The directed graphs are acyclic because $\caused$ is
  acyclic.
  
  Using this observation, we define a strict total order $\prec$ over pairs of
  complete transactions $T_1,T_2$ in $\alpha_1$. two steps: First, order the sets of transactions
  according to their timestamps $t$. Second, for each set of transactions,
  choose an arbitrary order for its DAGs and then within each DAG,
  topologically sort the transactions.
  
  To show $\prec$ is a strict total order, we must show it is irreflexive, total,
  and transitive. Irreflexivity follows from the irreflexivity of
  $\caused$ and the fact that a transaction can only belong to one set (since it
  only has one timestamp). Totality follows from the fact that timestamps are
  totally ordered, the fact that the arbitrary order of DAGs is chosen to be
  total, and the fact that the directed graphs of transactions are acyclic.
  We now show $\prec$ is transitive.
  
  Let $T_1$, $T_2$, and $T_3$ be three transactions with timestamps $t_1$,
  $t_2$, and $t_3$ such that $T_1 \prec T_2$ and $T_2 \prec T_3$. We show
  $T_1 \prec T_3$.

  There are four cases. 
  (1) If $t_1 < t_2$ and $t_2 < t_3$, then clearly $t_1 < t_3$ and $T_1 \prec T_3$.
  Similarly, (2) if $t_1 < t_2$ and $t_2 = t_3$ or (3) if $t_1 = t_2$ and $t_2 <
  t_3$, $T_1 \prec T_3$. For the case (4) where
  $t_1 = t_2 = t_3$, there are four sub-cases:
  
  \begin{enumerate}[label=(\alph*)]
    \item If $T_1 \caused T_2$ and $T_2 \caused T_3$, then $T_1$,
    $T_2$, and $T_3$ are in the same DAG of causally related transactions with
    the same timestamp. By the transitivity of $\caused$, $T_1 \caused T_3$, so
    $T_1$ will be topologically sorted before $T_3$.
    \item Now suppose $T_1 \caused T_2$ and $T_2 \not\caused T_3$.
    $T_1$ and $T_2$ are in the same DAG, but $T_3$ is in a different DAG.
    Furthermore, because $T_2 \prec T_3$ by assumption, $T_2$'s (and $T_1$'s) DAG
    is ordered before $T_3$'s. Thus, $T_1 \prec T_3$.
    \item Similar reasoning applies to the case where $T_1 \not\caused T_2$
    and $T_2 \caused T_3$.
    \item Finally, if $T_1 \not\caused T_2$ and $T_2 \not\caused T_3$, all
    three transactions are in different DAGs. Since $T_1 \prec T_2$ and $T_2
    \prec T_3$, $T_3$'s DAG must follow $T_2$'s and thus $T_1$'s, so $T_1 \prec
    T_3$.
  \end{enumerate}
  Thus, $\prec$ is a strict total order.
  
  Let $S$ be the sequence of transaction invocations and responses defined by $\prec$. $\alpha_1$, however, may not contain some responses that are in $S$, in particular, those of committed RW transactions whose response did not yet reach the client.
  
  We thus construct $\alpha_2$ by extending $\alpha_1$ with responses for these RW transactions. For each key $k$ read by a RW transaction whose response must be added, let the returned value be that of the most recent write of $k$ that precedes it in $S$. Then let $\alpha_2$ be the extension of $\alpha_1$ with any necessary response actions with these return values.
  
  To conclude the proof, we first show that $S$ is in
  \spannerrss{}'s sequential specification (i.e., the order is consistent with
  the values returned by each transaction's reads) and then that it satisfies
  \rss{}.
  
  To show $S$ is in \spannerrss{}'s sequential
  specification, we use Lemmas~\ref{theorem:spannerrss:rw-valid}
  and~\ref{theorem:spannerrss:ro-valid}. First, observe that since a read in \spannerrss{} can only return a transaction's
  write after the transaction has committed, any transaction whose writes have
  been observed will be complete, have a commit timestamp, and thus be ordered by $\prec$.
  
  Let $T_2$ be a RW transaction.
  Since $\prec$ orders the transactions according to their timestamps $t$, by
  Lemma~\ref{theorem:spannerrss:rw-valid}, $T_2$'s reads include the writes of all
  transactions $T_1$ such that $T_1 \prec T_2$.

  
  Now let $T_2$ be a RO transaction. By Lemma~\ref{theorem:spannerrss:ro-valid},
  $T_2$'s reads reflect all the writes of all RW transactions $T_1$ such that
  $t_1 < t_2$. Further, by Observation 2 about conflicting RW transactions and
  Lemma~\ref{theorem:spannerrss:ro-valid}, if $T_1$ is a conflicting RW transaction
  such that $t_1 = t_2$, then $T_2$ reads from $T_1$, so $T_1 \caused T_2$. Then
  by the definition of $\prec$, $T_1 \prec T_2$. Thus, the sequence $S$ is in \spannerrss{}'s sequential
  specification. We now show it satisfies the three requirements of \rss{}:

  (1) By construction, $S$ contains the same invocations and responses as $\alpha_2$, which extends $\alpha_1$ with zero or more response actions. Further, as we show below, $S$ respects causality, which subsumes the clients' process orders. Thus, $S$ is equivalent to $\complete(\alpha_2)$.
  
  (2) Consider two transaction $T_1$ and $T_2$. Assume that $T_1 \caused T_2$.
  Lemma~\ref{theorem:spannerrss:causal-timestamps} implies that $t_1 \leq t_2$. If
  $t_1 < t_2$, then $T_1 \prec T_2$ because $\prec$ is first defined on the
  order of the timestamps of transactions. Otherwise, if $t_1 = t_2$, then $T_1$
  and $T_2$ are in the same DAG of causally related transactions. The topological
  sort of the DAG ensures that $T_1 \prec T_2$.
  
  (3) By Observation 3, it is clear that if $T_1$ and $T_2$ are RW transactions and
  $T_1 \rt T_2$, $t_1 < t_2$, so $T_1 \prec T_2$. Further, by
  Lemmas~\ref{theorem:spannerrss:rt-timestamps} and~\ref{theorem:spannerrss:ro-valid},
  if $T_1$ is a RW transaction and $T_2$ is a conflicting RO transaction, then
  $T_2$'s conflicting reads will return $T_1$'s writes or newer versions. As a
  result, either $t_1 = t_2$ and $T_1 \caused T_2$ or $t_1 < t_2$. In either
  case, $T_1 \prec T_2$.
  
\end{proof}

\input{sections/proofs-gryff}

%% file: sections/proofs-gryff.tex
\subsection{\gryffrs{}}
\label{sec:correctness:gryffrs}

Unless stated otherwise, we consider an arbitrary, well-formed execution $\alpha$ of a set of application processes
interacting with a \gryffrs{} service. In a slight abuse of notation, we define $o_1 \caused_\alpha o_2$
to mean that $o_1$'s response causally precedes $o_2$'s invocation, and
we define $o_1 \rt_\alpha o_2$ similarly.
For simplicity, we assume real-time values are unique.
To reason about the order of operations in \gryffrs{}, we first introduce several
definitions:

Given an operation $o$, we define its \emph{decision point}, denoted $\fq(o)$, as the time at which the last replica in its first-round quorum
chooses a \timestamp{}. If $o$ is a write or rmw, then $o$'s \emph{visibility point}, denoted $\vp(o)$, is the earliest time at which its write is applied to one replica, and $o$'s
\emph{propagation point}, denoted $\pp(o)$, is the earliest time that $o$'s write is applied to a quorum of replicas. The latter can occur either as part of $o$'s protocol or through dependency propagation.

Given an execution $\alpha$, we define an operation $o$ as \textit{complete} as follows:
If $o$ is a write or rmw, at least one replica has applied its key-value-\timestamp{}
tuple (while processing a \textit{Write2} message or executing a rmw command,
respectively).
Note that by definition, all complete operations have a decision point, and all complete writes and rmws have a visibility point. However, not all complete writes and rmws have a propagation point. Unless specified otherwise, we henceforth only consider complete operations.

Recall that each write or read-modify-write (rmw) in \gryffrs{} has a unique \timestamp{}. Further, a read's \timestamp{} is equal to the \timestamp{} of the write or rmw it reads from. Using these observations, we can define a total order over the operations to a single object $x \in X$.  

Let $o_1$ and $o_2$ be complete operations on object $x \in X$ with \timestamp{}s $\mts_1$ and $\mts_2$. We define a strict total order $<_x$ as follows: if $\mts_1 \neq \mts_2$, then $o_1 <_x o_2$ if and only if $\mts_1 < \mts_2$; otherwise, $\mts_1 = \mts_2$, and $o_1 <_x o_2$ if and only if $\fq(o_1) < \fq(o_2)$.

\begin{lemma}
  $<_x$ is a strict total order.
  \begin{proof}
    We must show $<_x$ is irreflexive, total, and transitive. Irreflexivity follows from the fact that each operation has one \timestamp{} and the irreflexivity of $<$. Totality follows from the fact that \timestamp{}s and real-time values are totally ordered. We now show $<_x$ is transitive.

    Let $o_1$, $o_2$, and $o_3$ be three complete operations with \timestamp{}s $\mts_1$, $\mts_2$, and $\mts_3$ such that $o_1 <_x o_2$ and $o_2 <_x o_3$. We must show $o_1 <_x o_3$.

    If $o_1 <_x o_2$ because $\mts_1 < \mts_2$ and $o_2 <_x o_3$ because $\mts_2 < \mts_3$, then clearly $o_1 <_x o_3$. Similarly, if $\mts_1 < \mts_2$ and $\mts_2 = \mts_3$ or if $\mts_1 = \mts_2$ and $\mts_2 < \mts_3$, then $o_1 <_x o_3$. Finally, if $\mts_1 = \mts_2 = \mts_3$, then $o_1 <_x o_2$ implies $\fq(o_1) < \fq(o_2)$ and $o_2 <_x o_3$ implies $\fq(o_2) < \fq(o_3)$. Thus, $\fq(o_1) < \fq(o_3)$. This implies $o_1 <_x o_3$.
  \end{proof}
\end{lemma}

\begin{definition}
  The \emph{sequential specification} $\spec_x$ of an object $x \in X$ is the
  set of all sequences of invocation-response pairs of reads, writes, and rmws to $x \in X$ such that each read or
  rmw returns the value written by the most recent write or rmw (or the initial
  value if no such write or rmw exists).
\end{definition}

\begin{lemma}
  \label{lemma:correctness:gryff:legality:single-object}
  The sequence $S_x$ defined by $<_x$ over the invocation-response pairs of complete operations to object $x \in X$ is in the sequential specification $\spec_x$.
  \begin{proof}
    Consider a read $r$ and let $w$ be the write or rmw that $r$ reads from.
    Then by \gryffrs{}'s protocol, $\mts_r = \mts_w$.
    
    Since $r$ reads
    from $w$, it must the case that $\vp(w) < \fq(r)$. Otherwise, $w$ would not
    have been applied at any replica when $r$ read from the replica. By the
    definition of $\fq(w)$, $\fq(w) < \vp(w)$, so $\fq(w) < \fq(r)$. Since
    $\mts_r = \mts_w$ and $\fq(w) < \fq(r)$, $w <_x r$.
    Further, since writes and rmws have unique \timestamp{}s, there does not exist
    any other write or rmw $w'$ with $\mts_w = \mts_{w'}$. This implies that for
    all other writes or rmws $w'$, either $w' <_x w <_x r$ or $w <_x r <_x w'$.
    Since $S_x$ is the sequence of invocation-response pairs
    defined by $<_x$, the same holds for $r$ in $S_x$.

    Now consider a rmw $\mathit{rmw}$. Let $w$ be the write or rmw that
    $\mathit{rmw}$ reads from. We proceed by contradiction.

    Assume for a contradiction that there exists a write or rmw $w'$ such that
    $w <_x w' <_x \mathit{rmw}$.
    Since $w <_x w' <_x \mathit{rmw}$, the definiton of $<_x$ implies that
    $\mts_w < \mts_{w'} < \mts_\mathit{rmw}$.
    Cases 3.2.2, 3.2.3, and 3.2.4 in the proof of Lemma~B.10 from the \gryff{}
    proof of correctness~\cite{burke2020gryff} show that such an ordering of
    \timestamp{}s is impossible when \timestamp{}s are assigned to writes and
    rmws as in \gryff{}. Since \gryffrs{} uses the same process for \timestamp{}
    assignment, this impossibility is a contradiction resulting from the
    earlier assumption.
    Since $S_x$ is the sequence of invocation-response pairs defined by $<_x$, there is thus also 
    no $w'$ such that $w <_{S_x} w' <_{S_x} \mathit{rmw}$.
  \end{proof}
\end{lemma}

Given an execution $\alpha$, we define $\caused'_\alpha \subseteq \caused_\alpha$ as the relation that omits the reads-from case. $\caused'_\alpha$ thus also omits any pairs derived transitively using one or more reads-from pairs. We prove three useful lemmas involving $\caused'_\alpha$. 

\begin{lemma}
  \label{lemma:decision-point-causality}
  Given operations $o_1$ and $o_2$ such that $o_1 \caused'_\alpha o_2$, $\fq(o_1) < \fq(o_2)$.
  \begin{proof}
    Since $\caused'_\alpha$ omits the reads-from case, $o_2$ must causally follow
    $o_1$ through some sequence of one or more actions related by process order
    or message passing. As a result, for any pair of adjacent operations
    $o_i,o_{i+1}$ in this sequence, $o_i \rt_\alpha o_{i+1}$. The transitivity
    of $\rt_\alpha$ thus implies $o_1 \rt_\alpha o_2$.
    
    By the definition of $\fq$, for any operation $o$, $\fq(o)$ is between $o$'s
    invocation and response. Since $o_1 \rt_\alpha o_2$, $\fq(o_1) < \resp(o_1)
    < \inv(o_2) < \fq(o_2)$.
  \end{proof}
\end{lemma}

\begin{lemma}
  \label{lemma:propagation-point-causality}
  Let $o_1$ be a write or rmw, $o_2$ be a read that reads from $o_1$, and $o_3$ be an operation such that $o_2 \caused'_\alpha o_3$.
  Then $\pp(o_1) \leq \fq(o_3)$. 
  \begin{proof}
    There are two cases.

    \noindentparagraph{Case 1.} Suppose $o_1$'s write is returned by a quorum of
    replicas in $o_2$. Then clearly $\pp(o_1) < \fq(o_2)$. By the definition of $\caused'_\alpha$, since $o_2 \caused'_\alpha o_3$, $o_2 \rt_\alpha o_3$, so $\resp(o_2) < \inv(o_3)$. Further, by the definition of $\fq$, $\fq(o_2) < \resp(o_2)$ and $\inv(o_3) < \fq(o_3)$. Together, these inequalities imply $\pp(o_1) < \fq(o_3)$.

    \noindentparagraph{Case 2.} Now suppose $o_1$'s write is not returned by a
    quorum of replicas in $o_2$.
    Then by \gryffrs{}'s read protocol, $o_2$'s client will store $o_1$'s write
    as its dependency $d$. Since $o_2 \caused'_\alpha o_3$, there must exist
    some sequence of operations that begins with an operation $o$ such that
    $o_2 \caused'_\alpha o
    \caused'_\alpha \ldots \caused'_\alpha o_3$. By the definition of
    $\caused'_\alpha$, $o_2 \rt_\alpha o$. There are two sub-cases.

    Assume $o = o_3$. If $o_1$ finishes applying its
    write to a quorum before $\fq(o_3)$, then we are done, so suppose not. Since
    $o_2 \caused'_\alpha o_3$ and $o_2$ stored $o_1$'s write as a dependency,
    \gryffrs{}'s dependency propagation ensures that $o_1$'s write is
    propagated as a dependency to the process invoking $o_3$. And since \gryffrs{}'s
    dependency propagation includes $o_1$'s write as part of the first round
    round messages for $o_3$, $\pp(o_1) = \fq(o_3)$.

    Now assume $o \neq o_3$. By similar reasoning about \gryffrs{}'s dependency
    propagation, $\pp(o_1) \leq \fq(o)$. By the definition of $\caused'_\alpha$
    and the transitivity of $\rt_\alpha$,
    $o \rt_\alpha o_3$. Thus,  $\pp(o_1) \leq \fq(o) < \resp(o) <
    \inv(o_3) < \fq(o_3)$.
  \end{proof}
\end{lemma}

\begin{lemma}
  \label{lemma:decision-point-object-and-causal}
  Let $o_1$, $o_2$, and $o_3$ be operations such that $o_1 <_x o_2$ for
  some $x \in X$ and $o_2 \caused'_\alpha o_3$. Then $\fq(o_1) < \fq(o_3)$.
  \begin{proof}
    To start, observe that by the definition of $\caused'_\alpha$, since $o_2
    \caused'_\alpha o_3$, the process executing $o_2$ must have either executed
    another operation or sent a message after $o_2$. This implies $o_2$ executed
    its entire protocol. Further, by the reasoning in
    Lemma~\ref{lemma:decision-point-causality}, $\fq(o_2) < \resp(o_2) <
    \inv(o_3) < \fq(o_3)$. There are two cases.

    \noindentparagraph{Case 1.} Assume $o_2$ is a write or rmw. Since
    $o_1 <_{x_1} o_2$ and $o_2$ is a
    write or rmw, $\mts_{o_1} < \mts_{o_2}$. Further, since $o_2$ executes its
    entire protocol, it must be the case that $\pp(o_2) < \resp(o_2)$. Finally,
    $\fq(o_1) < \pp(o_2)$ because otherwise $o_1$ would read $o_2$'s \timestamp{}
    at at least one replica in its first-round quorum. By \gryffrs{}'s protocol,
    this would force $\mts_{o_2} \leq \mts_{o_1}$, contradicting the fact that
    $\mts_{o_1} < \mts_{o_2}$. Together, these inequalities imply $\fq(o_1) <
    \pp(o_2) < \resp(o_2) < \inv(o_3) < \fq(o_3)$.

    \noindentparagraph{Case 2.} Assume $o_2$ is a read. There are two sub-cases.

    (2a) Assume $o_1$ is also a read. Suppose $o_1$ and $o_2$ read from the same
    write or rmw. Then $\mts_{o_1} = \mts_{o_2}$, and by the definition of
    $<_{x_1}$, $\fq(o_1) < \fq(o_2)$. Combined with the inequality above,
    $\fq(o_1) < \fq(o_2) < \fq(o_3)$.
    
    Now suppose $o_1$ reads from $w_1$ and $o_2$ reads from $w_2 \neq w_1$.
    Since $o_1 <_{x_1} o_2$, $\mts_{o_1} = \mts_{w_1} < \mts_{w_2} = \mts_{o_2}$.
    Assume to contradict that $\fq(o_3) < \fq(o_1)$. By
    Lemma~\ref{lemma:propagation-point-causality}, because $o_2$ reads from
    $w_2$ and $o_2 \caused'_\alpha o_3$, $\pp(w_2) \leq \fq(o_3)$, so $\pp(w_2)
    < \fq(o_1)$. But then $w_2$ would be applied at at least one replica before
    that replica chooses a \timestamp{} for $o_1$ in $o_1$'s first-round quorum.
    This implies $\mts_{o_2} = \mts_{w_2} \leq \mts_{o_1}$, which contradicts
    the fact that $\mts_{o_1} < \mts_{o_2}$.   

    (2b) Assume $o_1$ is a write or rmw. Suppose $o_2$ reads from $o_1$. By the
    definition of visibility point, $\fq(o_1) < \vp(o_1)$. Since $o_2$ read from
    $o_1$, clearly $\vp(o_1) < \fq(o_2)$. Combined with the reasoning above,
    these inequalities imply $\fq(o_1) < \vp(o_1) < \fq(o_2) < \fq(o_3)$.

    Now suppose $o_2$ reads from some $w_2 \neq o_1$. Since $o_1 <_{x_1} o_2$,
    $\mts_{o_1} < \mts_{w_2} = \mts_{o_2}$. Assume to contradict that $\fq(o_3)
    < \fq(o_1)$. By Lemma~\ref{lemma:propagation-point-causality}, because $o_2$
    read from $w_2$ and $o_2 \caused'_\alpha o_3$, $\pp(w_2) \leq \fq(o_3)$, so
    $\pp(w_2) < \fq(o_1)$. But then $w_2$ would be applied at at least one
    replica before that replica chooses a \timestamp{} for $o_1$ in $o_1$'s
    first-round quorum. As above, this contradicts the fact that $\mts_{o_1} <
    \mts_{o_2}$.
  \end{proof}
\end{lemma}

\begin{lemma}{(Zig-Zag Lemma)}
  \label{lemma:zig-zag}
  Consider a sequence of $m \geq 2$ operations $o_1,...,o_m$ such that the following hold: (1) successive
  pairs alternate between belonging to $\caused'_\alpha$ and
  $<_{x_k}$ for some $x_k \in X$; and (2) the last pair $(o_{m-1},o_m)$ belongs to $\caused'_\alpha$. Then $\fq(o_1) < \fq(o_m)$.
  \begin{proof}
    By the definition of the sequence, there are two cases. We prove the first
    and then use the result to prove the second.
    
    \noindentparagraph{Case 1.} Assume $o_1 <_{x_1} o_2$.
    The sequence ends with $o_{m-1} \caused'_\alpha o_m$, so by the
    assumption of the case, $m \geq 3$ and $m$ must be odd. The sequence thus may be
    written as $o_{2k-1} <_{x_{2k-1}} o_{2k} \caused'_\alpha o_{2k+1}$ for all
    $1 \leq k \leq (m-1)/2$.
    
    We prove the case by induction on $k$, so first consider the base case $o_1
    <_{x_1} o_2 \caused'_\alpha o_3$. Lemma~\ref{lemma:decision-point-object-and-causal}
    implies that $\fq(o_1) < \fq(o_3)$.

    \emph{Inductive hypothesis}: Assume $\fq(o_1) <
    \fq(o_{2(j-1)+1})$ for some $j < (m-1)/2$.

    \emph{Inductive step}: By the definition of the sequence, $o_{2j-1} <_{x_{j-1}}
    o_{2j} \caused'_\alpha o_{2j+1}$. Lemma~\ref{lemma:decision-point-object-and-causal}
    implies that $\fq(o_{2j-1}) < \fq(o_{2j+1})$. Since
    $o_{2(j-1)+1} = o_{2j-1}$, the inductive hypothesis implies
    $\fq(o_1) < \fq(o_{2(j-1)+1}) = \fq(o_{2j-1}) <
    \fq(o_{2j+1})$.

    \noindentparagraph{Case 2.} Assume $o_1 \caused'_\alpha o_2$. By
    Lemma~\ref{lemma:decision-point-causality}, $\fq(o_1) < \fq(o_2)$.
    If $m=2$, we are done. Otherwise, the sequence of operations
    $o_2,...,o_m$ comprises successive pairs alternating
    between $<_{x_k}$ and $\caused'_\alpha$, starting with $<_{x_2}$ and ending with $\caused'_\alpha$.
    As a result, the first case shows that $\fq(o_2) < \fq(o_m)$, and thus, $\fq(o_1) < \fq(o_2) < \fq(o_m)$.
  \end{proof}
\end{lemma}

Let $<_\psi$ be a partial order defined over pairs of complete operations $o_1,o_2$ as follows:
\begin{itemize}
  \item $o_1 <_x o_2 \implies o_1 <_\psi o_2$,
  \item $o_1 \in \mathcal{W} \land o_2 \in \conflicts_\alpha(o_1) \cup \mathcal{W} \land o_1 \rt_\alpha o_2 \implies
    o_1 <_\psi o_2$,
  \item $o_1 \caused'_\alpha o_2 \implies o_1 <_\psi o_2$, and
  \item $o_1 <_\psi o_2 \land o_2 <_\psi o_3 \implies o_1 <_\psi o_3$.
\end{itemize}

\begin{lemma}
  The partial order $<_\psi$ is acyclic.
  \begin{proof}
    We prove the lemma by contradiction, so assume there exists a
    cycle. Consider a shortest such cycle of $m$ operations $o_1,o_2,...,o_m$.
    Observe that $<_\psi$ is irreflexive by definition, so $m \geq 2$.
    First, we prove three useful properties of the cycle.

    \noindentparagraph{Property 1.} \emph{There cannot be two consecutive $\caused'_\alpha$ edges in
        the cycle.} Assume to contradict that there are
        two consecutive $\caused'_\alpha$ edges $o_j \caused'_\alpha o_{j+1}
        \caused'_\alpha o_{j+2}$. By the transitivity of $\caused'_\alpha$,
        there must exist an edge $o_j \caused'_\alpha o_{j+2}$, which forms a
        shorter cycle $...,o_j,o_{j+2},...,o_j$. This contradicts our choice of
        a shortest cycle.

    \noindentparagraph{Property 2.} \emph{There cannot be two consecutive $<_x$ edges in the cycle.} By similar reasoning as above, $o_j <_x o_{j+1} <_x o_{j+2}$ implies
    the existence of a shorter cycle using the edge $o_j <_x o_{j+2}$. This contradicts our choice of a shortest cycle.

    \noindentparagraph{Property 3.} \emph{There is at most one $\rt_\alpha$ edge in the cycle.}
        Assume to contradict that there are two $\rt_\alpha$ edges
        in the cycle. Let them be $o_i \rt_\alpha o_j$ and $o_k
        \rt_\alpha o_\ell$, re-indexing the cycle if necessary such that $i < j
        \leq k < \ell$. By the definition of $\rt_\alpha$, $\resp(o_i) < \inv(o_j)$ and $\resp(o_k) < \inv(o_\ell)$. It must be the case that
        $\inv(o_\ell) < \resp(o_i)$; otherwise the edge $o_i \rt_\alpha
        o_\ell$ would exist, which allows for a shorter cycle $...,o_i, o_\ell, ..., o_i$. It also must be the case that $\inv(o_j) < \resp(o_k)$, since the contrary would similarly imply the existence
        of a shorter cycle. Together, however, these inequalities imply that $\resp(o_i) <
        \inv(o_j) < \resp(o_k) < \inv(o_\ell) < \resp(o_i)$, contradicting the irreflexivity of $<$.

    We are now ready to prove $<_\psi$ is acyclic. Recall that $<_\psi$ has three types of edges: $<_x$,
    $\caused'_\alpha$, and $\rt_\alpha$. By Property 3, there are two cases.

    \noindentparagraph{Case 1.} Assume there are zero $\rt_\alpha$ edges in the cycle.
    Since $m \geq 2$, there
    are at least two edges in the cycle, and by Property 1, at
    least one is a $<_x$ edge. Without loss of generality,
    re-index the cycle $o_1,o_2,...,o_m,o_1$ such that $o_1 <_{x_1} o_2$.
    By Properties 1 and 2, the sequence $o_1,...,o_m,o_1$ must
    alternate between $<_{x_i}$ and $\caused'_\alpha$, starting with $<_{x_1}$
    and ending with $\caused'_\alpha$.
    Lemma~\ref{lemma:zig-zag} thus implies $\fq(o_1) < \fq(o_1)$, contradicting the irreflexivity of $<$.
    
    \noindentparagraph{Case 2.} Assume there is one $\rt_\alpha$ edge in the cycle. Without
    loss of generality, re-index the cycle $o_1,o_2,...,o_m,o_1$ such that
    $o_1 \rt_\alpha o_2$. Note that since $o_1 \rt_\alpha o_2$, $\fq(o_1) < \fq(o_2)$. We proceed by cases.

    (2a) $o_m \caused'_\alpha o_1$. From the re-indexing of the cycle and by
    the assumptions of this case, the sequence of operations from
    $o_2,...,o_m,o_1$ alternates between $<_{x_i}$ and $\caused'_\alpha$, ending with $\caused'_\alpha$.
    Thus, by Lemma~\ref{lemma:zig-zag}, $\fq(o_2) < \fq(o_1)$, contradicting the inequality above.

    (2b) $o_m <_{x_m} o_1$. By
    the assumptions of this case, the sequence of operations from
    $o_2,...,o_{m-1},o_m$ alternates between $<_{x_i}$ and $\caused'_\alpha$, ending with $\caused'_\alpha$.
    Thus by Lemma~\ref{lemma:zig-zag}, $\fq(o_2) < \fq(o_m)$.
    
    By the definitions of $<_\psi$, $\rt_\alpha$, and $\fq(o_2)$, since $o_1 \rt_\alpha o_2$,
    $o_1 \in \mathcal{W}$ and $\pp(o_1) < \resp(o_1) < \inv(o_2) < \fq(o_2)$.
    This implies $\pp(o_1) < \fq(o_m)$.

    Since $\pp(o_1) < \fq(o_m)$, $o_1$'s write
    would be applied at at least one replica when that replica replies with a
    \timestamp{} in $o_m$'s first-round quorum. This implies $\mts_{o_1} <
    \mts_{o_m}$, which contradicts the assumption that $o_m <_{x_m} o_1$.
  \end{proof}
\end{lemma}

\begin{lemma}
  \label{lemma:correctness:gryff:legality}
  A topological sort $S$ of $<_\psi$ over complete operations to objects $X$ is in the sequential
  specification $\prod_{x \in X} \spec_x$.
  \begin{proof}
    Since each operation targets a single object $x \in X$, it
    suffices to only consider the orders of subsets of operations to each $x$.
    For a fixed $x$, this order is solely dictated by the partial order $<_x$ in the topological sort of $<_\psi$.
    Lemma~\ref{lemma:correctness:gryff:legality:single-object} shows that the sequence $S_x$ defined by
    $<_x$ is in $x$'s sequential specification $\spec_x$. Since $S$ is a topological sort
    over $<_\psi$, which by definition generalizes each $<_x$, it follows that the entire sequence $S$
    is in $\prod_{x \in X} \spec_x$.
  \end{proof}
\end{lemma}

\begin{theorem}
  \label{theorem:correctness:gryff}
  \gryffrs{} guarantees \rs{}.
  \begin{proof}
    Let $\alpha_1$ be a well-formed execution of \gryffrs{}. Extend $\alpha_1$ to $\alpha_2$ by adding a response action for any complete operation $o$
    that does not have one in $\alpha_1$.
    
    Let $S$ be a topological sort of $<_\psi$ on the operations in $\alpha_2$. Lemma~\ref{lemma:correctness:gryff:legality} implies that $S \in
    \prod_{x \in X} \spec_x$. We now show that $S$ satisfies the three properties of \rs{}:

    (1) By construction, $S$ contains the same invocations and responses as
    $\alpha_2$, which extends $\alpha_1$ with zero or more response actions.
    Further, as we show below, $S$ respects $\caused_{\alpha_1}$, which subsumes the
    clients' process orders. Thus, $S$ is equivalent to $\complete(\alpha_2)$.

    (2) Consider two operations $o_1$ and $o_2$ such that $o_1 \caused_{\alpha_1} o_2$.
    By the definition of $\alpha_2$, $\caused'_{\alpha_1} \subseteq \caused'_{\alpha_2}$. Further,
    since the definition of $<_\psi$ includes $\caused'_{\alpha_2}$, $<_x$ (including $o_2$ reading from $o_1$), and their transitive closure, $<_\psi$ thus also includes $\caused_{\alpha_1}$. Since $S$ is a topological sort of $<_\psi$, $o_1 <_S o_2$.
    
    (3) Consider two operations $o_1 \in \mathcal{W}$ and $o_2 \in \conflicts_{\alpha_1}(o_1)
    \cup \mathcal{W}$ such that $o_1 \rt_{\alpha_1} o_2$. By the definition of $<_\psi$,
    $o_1 <_\psi o_2$, and since $S$ is a topological sort of $<_\psi$, $o_1 <_S o_2$.
  \end{proof}
\end{theorem}